\documentclass{statsoc}
\usepackage{etoolbox}
\makeatletter
\patchcmd{\@makecaption}
  {\parbox}
  {\advance\@tempdima-\fontdimen2} % decrease the width!
  {}{}
\makeatother 

\usepackage[a4paper,top=1.25in, bottom=1.25in, left=1in, right=1in]{geometry}

\usepackage{graphicx}
\usepackage{amsmath}
\usepackage{amssymb}
\usepackage{graphicx}
\usepackage{natbib}
\usepackage{enumitem}
\usepackage{mathtools}
\usepackage{mathrsfs}
\mathtoolsset{showonlyrefs}
\usepackage{tikz}
\usepackage{anyfontsize}
\usepackage{dsfont}
\usepackage[ruled]{algorithm2e}
\usepackage{hyperref}
\hypersetup{
    colorlinks,%
    citecolor=blue,%
    filecolor=blue,%
    linkcolor=black,%
    urlcolor=blue
}
\newcommand{\R}{\mathbb{R}}

\newcommand{\E}{\mathbb{E}}
\newcommand{\p}{\mathbb{P}}

\newcommand{\indep}{\rotatebox[origin=c]{90}{$\models$}}

\newcommand{\Indc}{\mathbf{1}}
\newcommand{\argmax}[1]{\underset{#1}{\arg\!\max}}

\newcommand{\event}{\mathcal{V}}
\newcommand{\Z}{\mathcal{Z}}
\newcommand{\lb}{\left(}
\newcommand{\rb}{\right)}
\newcommand{\const}{B}

\newcommand{\D}{\mathcal{D}}

\usepackage{xspace}

\definecolor{myblue}{rgb}{.8, .8, 1}
\definecolor{mathblue}{rgb}{0.2472, 0.24, 0.6} % mathematica's Color[1, 1--3]
\definecolor{mathred}{rgb}{0.6, 0.24, 0.442893}
\definecolor{mathyellow}{rgb}{0.6, 0.547014, 0.24}

\newcommand{\tT}{{\widetilde{T}}}

\newcommand{\tX}{{\widetilde{X}}}
\newcommand{\tY}{{\widetilde{Y}}}

\newcommand{\calA}{{\mathcal{A}}}

\newcommand{\calC}{{\mathcal{C}}}
\newcommand{\calD}{{\mathcal{D}}}
\newcommand{\calE}{{\mathcal{E}}}
\newcommand{\calF}{{\mathcal{F}}}

\newcommand{\calI}{{\mathcal{I}}}

\newcommand{\calN}{{\mathcal{N}}}
\newcommand{\calO}{{\mathcal{O}}}

\newcommand{\calS}{{\mathcal{S}}}

\newcommand{\calU}{{\mathcal{U}}}

\newcommand{\calX}{{\mathcal{X}}}

\newcommand{\calZ}{{\mathcal{Z}}}

\newcommand{\td}{\tilde}

\graphicspath{{./figs/}}

\newtheorem{theorem}{Theorem}
\newtheorem{proposition}{Proposition}
\newtheorem{lemma}{Lemma}
\newtheorem{assumption}{Assumption}
\newtheorem{remark}{Remark}

\DeclareMathOperator{\ind}{\mathds{1}}  % Indicator

\newcommand{\neqd}{\stackrel{\textnormal{d}}{\neq}}
\newcommand{\tr}{\mathrm{tr}}
\newcommand{\ca}{\mathrm{ca}}
\newcommand{\quantile}{\mathrm{Quantile}}

\renewcommand{\citet}{\cite}

\renewcommand{\c}{c}
\newcommand{\eps}{\varepsilon}
\newcommand{\Cend}{C_{\mathrm{end}}}
\newcommand{\Closs}{C_{\mathrm{loss}}}
\newcommand{\condIC}{conditionally independent censoring assumption}
\newcommand{\compIC}{completely independent censoring assumption}
\newcommand{\revise}[1]{{\color{black} #1}}

\title[Conformalized Survival Analysis]{Conformalized Survival Analysis}
\author{Emmanuel Cand\`es}
\address{Department of Mathematics and Department of Statistics, Stanford University, CA, USA.}
\email{candes@stanford.edu}

\author{Lihua Lei}
\address{Graduate School of Business, Stanford University, Stanford, CA, USA.}
\email{lihualei@stanford.edu}

\author[Emmanuel Cand\`{e}s, Lihua Lei and Zhimei Ren]{Zhimei Ren}
\address{Department of Statistics, University of Chicago, Chicago, IL, USA.}
\email{zmren@uchicago.edu}

\begin{document}
\begin{abstract}
  Existing survival analysis techniques heavily rely on strong
  modelling assumptions and are, therefore, prone to model
  misspecification errors.  In this paper, we develop an inferential
  method based on ideas from conformal prediction, which can wrap
  around any survival prediction algorithm to produce calibrated,
  covariate-dependent lower predictive bounds on survival times. In
  the Type I right-censoring setting, when the censoring times are
  completely exogenous, the lower predictive bounds have guaranteed
  coverage in finite samples without any assumptions other than that
  of operating on independent and identically distributed data
  points. Under a more general conditionally independent censoring
  assumption, the bounds satisfy a doubly robust property which states
  the following: marginal coverage is approximately guaranteed if
  either the censoring mechanism or the conditional survival function
  is estimated well. Further, we demonstrate that the lower predictive
  bounds remain valid and informative for other types of censoring.
  The validity and efficiency of our procedure are demonstrated on
  synthetic data and real COVID-19 data from the UK
  Biobank.% We show strong
  % evidence that the proposed method can achieve desired coverage.
\end{abstract}

\section{Introduction}
\label{sec:intro}
%!TEX ROOT = main.tex

The COVID-19 pandemic has placed extraordinary demands on health
systems \citep[e.g.,][]{ranney2020critical}. In turn, these demands
create an unavoidable need for medical resource allocation and, in
response, several groups of researchers have communicated clinical
ethics recommendations \citep[e.g.,][]{emanuel2020fair,
  vergano2020clinical}. By and large, these recommendations require a
% reliable benefit assessment of receiving specific types of medical
% resources;
reliable individual risk assessment for patients who test positive;
see Table 2 of \cite{emanuel2020fair}. Clearly, one risk
measure of interest might be the survival time, the time lapse between
the confirmation of COVID-19 and an event such as death or reaching a
critical state, should this ever occur.

\subsection{Survival analysis}

Survival times are not always observed due to censoring
\citep{leung1997censoring}. A main goal of survival analysis is to
infer the survival function---the probability that a patient will
survive beyond any specified time---from censored data. The
Kaplan-Meier curve \citep{kaplan1958nonparametric} produces such an
inference when the population under study is a group of patients with
certain characteristics. On the positive side, the Kaplan-Meier curve
does not make any assumption on the distribution of survival times. On
the negative side, it can only be applied to a handful of
subpopulations because it requires sufficiently many events in each
subgroup \citep{kalbfleisch2011statistical}.  More often than not, the
scientist has available multiple categorical and continuous
covariates, and it thus becomes of interest to understand
heterogeneity by studying the conditional survival function; that is,
the dependence on the available factors.  In the conditional
setting, however, distribution-free inference for the conditional
survival function gets to be challenging.  Standard approaches make
parametric or nonparametric assumptions about the distribution of the
covariates and that of the survival times conditional on covariate
values. A well-known example is of course the celebrated Cox model
which posits a proportional hazards model in which an unspecified
nonparametric base line is modified via a parametric model describing
how the hazard varies in response to explanatory covariates
\citep{cox1972regression, breslow1975analysis}. Other popular models,
such as accelerated failure time (AFT) \citep{cox1972regression,
  wei1992accelerated} and proportional odds models
\citep{murphy1997maximum, harrell2015regression}, also combine
nonparametric and parametric model specifications. 

As medical technologies produce ever larger and more complex clinical
datasets, we have witnessed a rapid development of machine learning
methods adapted to high-dimensional and heterogeneous survival data
\citep[e.g.,][]{verweij1993cross, faraggi1995neural,
  tibshirani1997lasso, gui2005penalized, hothorn2006survival,
  zhang2007adaptive, ishwaran2008random, witten2010survival,
  goeman2010l1, simon2011regularization, katzman2016deep, lao2017deep,
  wang2019machine, li2020censored}. An appealing feature of these
methods is that they typically do not make modeling
assumptions.
%footnote
To quote from \cite{efron2020prediction}: ``
  Neither surface nor noise is required as input to randomForest, gbm,
  or their kin.''
The downside is that it is often challenging to
quantify the uncertainty for these methods. To be sure, blind
application of off-the-shelf uncertainty quantification tools, such as
the bootstrap \citep{efron1979bootstrap, efron1994introduction}, can
yield unreliable results since their validity 1) rests on implicit
modeling assumptions, and 2) holds only asymptotically
\citep[e.g.,][]{lei2020conformal,
  ratkovic2021estimation}. % Blind application of machine learning methods to predict the benefit can be dangerous and unethical due to the limited domain knowledge and difficulty of uncertainty quantification.

\subsection{Prediction intervals}

For decision-making in sensitive and uncertain environments---think of
the COVID-19 pandemic---it is preferable to produce prediction
intervals for the \emph{uncensored} survival time with guaranteed
coverage rather than point predictions.  In this regard, the use of
$(1-\alpha)$ prediction intervals is an effective way of summarizing
what can be learned from the available data; wide intervals reveal a
lack of knowledge and keep overconfidence at arm's length. Here and
below, an interval is said to be a $(1-\alpha)$ prediction interval if
it has the property that it contains the true label, here, the
survival time, at least $100(1-\alpha)$\% of the time (a formal
definition is in Section 2).  Prediction intervals have been widely
studied in statistics \citep[e.g.,][]{wilks1941determination,
  wald1943extension, aitchison1980statistical, stine1985bootstrap,
  geisser1993predictive, vovk2005algorithmic,
  krishnamoorthy2009statistical} and much research has been concerned
with the construction of covariate-dependent intervals.

Of special interest is the subject of conformal inference, a generic
procedure that can be used in conjunction with sophisticated machine
learning prediction algorithms to produce prediction intervals with
valid marginal coverage without making any distributional assumption
whatsoever \citep[e.g.,][]{saunders1999transduction, vovk2002line,
  vovk2005algorithmic, lei2014distribution,
  tibshirani2019conformal}. While coverage is only guaranteed in a
marginal sense, it has been theoretically proved and empirically observed that some conformal
prediction methods can also achieve near conditional coverage---that
is, coverage assuming a fixed value of the covariates---when some key
parameters of the underlying conditional distribution can be estimated
reasonably well \citep[e.g.,][]{sesia2020comparison,
  lei2020conformal}.

\subsection{Our contribution}

Standard conformal inference requires fully observed outcomes and is
not directly applicable to samples with censored outcomes. In this
paper, we extend conformal inference to handle right-censored outcomes
in the setting of Type-I censoring
\citep[e.g.,][]{leung1997censoring}. This setting assumes that the
censoring time is observed for every unit while the outcome is only
observed for uncensored units. In particular, we generate a
covariate-dependent lower prediction bound (LPB) on the uncensored
survival time, which can be regarded as a one-sided
$(1-\alpha)$-prediction interval. As we just argued, the LPB is a
conservative assessment of the survival time, which is particularly
desirable for high-stakes decision-making. A low LPB value suggests
either % that the patient is
% exposed to 
a high risk for the patient, or a high degree of uncertainty for
similar patients due to data scarcity. Either way, the signal to a
decision-maker is that the patient deserves some attention.

Under the \compIC\space defined below, which states that the censoring
time is independent of both the outcome and covariates, our LPB
provably yields a $(1-\alpha)$ prediction interval.  This property
holds in finite samples \emph{without any assumption other than that
  of operating on i.i.d.~samples}. Under the more general
\condIC\space introduced later, our LPB satisfies a \emph{doubly
  robust} property which states the following: marginal coverage is
approximately guaranteed if either the censoring mechanism or the
conditional survival function is estimated well. In the latter case,
the LPB even has approximately guaranteed conditional coverage.

Readers familiar with conformal inference would notice that the above
guarantees can be achieved by simply applying conformal inference to
the censored outcomes, i.e., by constructing an LPB on the censored
outcome treated as the response. This unsophisticated approach is
conservative. Instead, we will see how to provide tighter bounds and
sharper inference by applying conformal inference on a subpopulation
with large censoring times; that is, on which censored outcomes are
closer to actual outcomes. To achieve this, we shall see how to
carefully combine the selection of a subpopulation with ideas from
weighted conformal inference \citep{tibshirani2019conformal}.

Lastly, while we focus on clinical examples, it will be clear from our
exposition that our methods can be applied to other time-to-event
outcomes in a variety of other disciplines, such as industrial life
testing \citep{bain2017statistical}, sociology
\citep{allison1984event}, and economics \citep{powell1986censored,
  hong2003inference, sant2016program}.

%%% Local Variables:
%%% mode: latex
%%% TeX-master: "main"
%%% End:

\section{Prediction intervals for survival times}
\label{sec:setup}
\subsection{Problem setup}\label{subsec:setup}
Let $X_i, C_i, T_i$, $i = 1, \ldots, n$, be respectively the vector of covariates, the censoring time, and the survival
time of the $i$-th
unit/patient. Throughout the paper, we assume that $(X_i, C_i, T_i)$
are i.i.d.~copies of the random vector $(X,C,T)$.  We consider the
Type I right-censoring setting, where the observables for the $i$-th
unit include $X_i, C_i$, and the censored survival time $\tT_i$,
defined as the minimum of the survival and censoring time:
\[\tT_i = \min(T_i, C_i).\]
For instance, if $T_i$ measures the time lapse between the admission
into the hospital and death, and $C_i$ measures the time lapse between
the admission into the hospital and the day data analysis is
conducted, then $\tT_i = T_i$ if the $i$-th patient died before the
day of data analysis and $\tT_i = C_i$ if she survives beyond that
day.

The censoring time $C$ partially masks information from the inferential target $T$. As discussed by \cite{leung1997censoring}, it is necessary to impose constraints on the dependence structure between $T$ and $C$ to enable meaningful inference. In particular, we make the following \textbf{\condIC} \citep[e.g.,][]{kalbfleisch2011statistical}:
\begin{assumption}[conditionally independent censoring]
\begin{align}\label{eq:conditional_independent_censoring}
  T~\indep~C\mid X.
\end{align}
\end{assumption}
This assumes away any unmeasured confounder affecting both the
survival and censoring time; please see immediately below for an
example. In some cases, we also consider the \textbf{\compIC}, which is stronger in the sense that it implies the former:
\begin{assumption}[completely independent censoring]
\begin{align}\label{eq:independent_censoring}
  (T, X)~\indep~C.
\end{align}
\end{assumption}
For instance, in a randomized clinical trial, the end-of-study
censoring time $C$ is defined as the time lapse between the
recruitment and the end of the study. For single-site trials, $C$ is
often modelled as a draw from an exogenous stochastic process
\citep[e.g.,][]{carter2004application, gajewski2008predicting} and
thus obeys \eqref{eq:independent_censoring}. For multicentral trials,
$C$ is often assumed to depend on the site location only
\citep[e.g.,][]{carter2005practical, anisimov2007modelling,
  barnard2010systematic}, and thus
\eqref{eq:conditional_independent_censoring} holds as soon as the
vector of covariates includes the site of the trial. \revise{For an
  observational study such as the COVID-19 example discussed later in
  Section \ref{sec:application}, additional covariates would be
  included to make the conditionally independent censoring assumption
  plausible.}

Although \eqref{eq:conditional_independent_censoring} is a strong
assumption, it is a widely used starting point to study survival
analysis methods \citep{kalbfleisch2011statistical}. We leave the
investigation of informative censoring
\citep[e.g.,][]{lagakos1979general, wu1988estimation,
  scharfstein2002estimation} to future research. Additionally, whereas
the setting of Type I censoring appears to be restrictive, we will
show in Section \ref{subsec:beyond} that an LPB in this setting can
still be informative for other censoring types.

% \subsection{Coverage criteria for lower prediction bounds}
\subsection{Naive lower prediction bounds}\label{subsec:naive}
Our ultimate goal is to generate a covariate-dependent LPB as a
conservative assessment of the uncensored survival time $T$. Denote by
$\hat{L}(\cdot)$ a generic LPB estimated from the observed data
$(X_i, C_i, \tT_i)_{i=1}^{n}$. We say an LPB is \emph{calibrated} if
it satisfies the following coverage criterion:
\begin{align}\label{eq:marginal_criterion}
  \p\big(T\ge \hat{L}(X)\big) \ge 1-\alpha,
\end{align}
where $\alpha$ is a pre-specified level (e.g., $0.1$), and the
probability is computed over both $\hat{L}(\cdot)$ and a future unit
$(X,C,T)$ that is independent of $(X_i,C_i,T_i)_{i=1}^{n}$. 

Since $\tT\le T$, any calibrated LPB on the censored survival time
$\tT$ is also a calibrated LPB on the uncensored survival time
$T$. Consequently, a naive approach is to discard the censoring time
$C_i$'s and construct an LPB on $\tT$ directly. Since the samples
$(X_i, \tT_i)$ are i.i.d., a distribution-free calibrated LPB on $\tT$
can be obtained via standard techniques from conformal inference
\citep[e.g.,][]{vovk2005algorithmic, lei2018distribution,
  romano2019conformalized}. Our first result is somewhat negative:
indeed, it states that all distribution-free calibrated LPBs on $T$
must be LPBs on $\tT$.

\begin{theorem}
  \label{thm:distribution_free}
  Take $X\in \R^{p}$ and $C\ge 0$, $T\ge 0$. Assume that 
  $\hat{L}(\cdot)$ is a calibrated LPB on $T$ for all joint
  distributions of $(X, C, T)$ obeying the \condIC ~with $X$ being
  continuous and $(T, C)$ being continuous or discrete.
  % \footnote{Our
    % proof can be extended to include the case where either $C$ or $T$
    % or both are mixtures of discrete and continuous distributions but
    % we do not consider such extensions here.}
  Then for any such
  distribution,
  \[\p(\tT\ge \hat{L}(X))\ge 1 - \alpha.\]  
\end{theorem}

% footnote
Our proof can be extended to include the case where either $C$ or $T$
    or both are mixtures of discrete and continuous distributions but
    we do not consider such extensions here.
An LPB constructed by taking $\tT$ as the response may be calibrated
but also overly conservative because of the censoring mechanism. To
see this, note that the oracle LPB on $\tT$ is, by definition, the $\alpha$-th
conditional quantile of $\tT \mid X$, denoted by
$\td{q}_{\alpha}(X)$. Similarly, let $q_{\alpha}(X)$ be the oracle LPB
on $T$. Under the \condIC,
\begin{align}
  \p(T \ge q_{\alpha}(x)\mid X = x) = 1 - \alpha & = \p(\tT \ge \td{q}_{\alpha}(x)\mid X = x) \nonumber
  \\ & = \p(T \ge \td{q}_{\alpha}(x)\mid X = x)\p(C \ge \td{q}_{\alpha}(x)\mid X = x). \label{eq:qqtd}
  \end{align}
  If the censoring times are small, the gap between
  $\td{q}_{\alpha}(x)$ and $q_{\alpha}(x)$ can be large. For
  illustration, assume that $X, C$, and $T$ are mutually independent,
  and $T\sim \mathrm{Exp}(1), C\sim \mathrm{Exp}(b)$. It is easy to
  show that $q_{\alpha}(X) = -\log(1 - \alpha)$ and
  $\td{q}_{\alpha}(X) = -\log(1 - \alpha) / (1 + b)$. Thus, a naive
  approach taking $\tT$ as a target of inference can be arbitrarily
  conservative.

  In sum, Theorem \ref{thm:distribution_free} implies that any
  calibrated LPB on $T$ must be a calibrated LPB on $\tT$, under the
  \condIC ~only. This is why to make progress and overcome the
  limitations of the naive approach, we shall need additional 
  distributional assumptions.

\subsection{Leveraging the censoring mechanism}

We have just seen that the conservativeness of the naive approach is
driven by small censoring times. A heuristic way to mitigate this
issue is to discard units with small values of $C$. Consider a
threshold $c_0$, and extract the subpopulation on which $C \ge
c_0$. One immediate issue with this is that the selection induces a
distributional shift between the subpopulation and the whole
population, namely,
\[
(X, C, T) \neqd (X, C, T) \mid C \ge c_0.
\]
For instance, the patients with larger censoring times tend to be
healthier than the remaining ones. To examine the distributional shift
in detail, note that the joint distribution of $(X, \tT)$ on the
whole population is $P_{X}\times P_{\tT\mid X}$ while that on the
subpopulation is
\[
  P_{(X, \tT)\mid C\ge c_0} = P_{X\mid C \ge c_0}\times P_{\tT\mid X,
    C\ge c_{0}}.
\]
Next, observe that $P_{\tT\mid X, C\ge c_{0}}\neq P_{\tT\mid X}$ even
under the \compIC~because $(T, X)\indep C$ does not imply
$\tT \indep C\mid X$ in general. For example, as in Section
\ref{subsec:naive}, if $X, C$, and $T$ are mutually independent and
$T, C\stackrel{i.i.d.}{\sim} \mathrm{Exp}(1)$, then
$\p(\tT\ge a, C\ge a) = \p(\tT\ge a) > \p(\tT\ge a)\p(C\ge a)$, for
any $a > 0$.  As a result, both the covariate
distribution and the conditional distribution of $\tT$ given $X$
differ in the two populations.

Now consider a secondary censored outcome $\tT\wedge c_0$, where
$a\wedge b = \min\{a, b\}$. We have
\begin{align}
  \nonumber 
  P_{(X, \tT\wedge c_0)\mid C\ge c_0} = P_{X\mid C \ge c_0}\times P_{\tT\wedge c_0\mid X, C\ge c_{0}} & \stackrel{(a)}{=} P_{X\mid C \ge c_0}\times P_{T\wedge c_0\mid X, C\ge c_{0}} \\
   \label{eq:subpopulation_dist} & \stackrel{(b)}{=} P_{X\mid C \ge c_0}\times P_{T\wedge c_0\mid X}, 
\end{align}
where (a) uses the fact that
\[T\wedge c_0 = \tT \wedge c_0, \,\,\text{ if } C\ge c_0,\]
and (b) follows from the \condIC. On the other hand, the joint distribution of $(X, T\wedge c_0)$ on the whole population is
\begin{align}
  \label{eq:whole_population_dist}
  P_{(X, T\wedge c_0)} = P_{X}\times P_{T\wedge c_0 \mid X}.
\end{align}
Contrasting \eqref{eq:subpopulation_dist} with
\eqref{eq:whole_population_dist}, we observe that {\em there is only a
  covariate shift} between the subpopulation and the whole population.

The likelihood ratio between the two covariate distributions is
\begin{equation}
  \label{eq:covariate_shift}
  \frac{dP_{X}}{dP_{X\mid C\ge c_0}}(x) = \frac{\p(C \ge c_0)}{\p(C\ge c_0\mid X = x)}.
\end{equation}
\revise{While there is a distributional shift between the selected units and the target population, the special form of the covariate shift allows us to adjust for the bias by carefully reweighting the samples. In particular, }applying the one-sided version of weighted conformal inference
\citep{tibshirani2019conformal}, discussed in the next section, gives
a calibrated LPB on $T\wedge c_0$, and thus a calibrated LPB on
$T$. With sufficiently many units with large values of $C$, we can
choose a large threshold $c_0$ to reduce the loss of power caused by
censoring. We emphasize that there is no contradiction with Theorem
\ref{thm:distribution_free} because, as shown in Section
\ref{sec:method}, weighted conformal inference requires
$\p(C\ge c_0 \mid X)$ to be (approximately) known.

We refer to the denominator $\p(C\ge c_0\mid X = x)$ in
\eqref{eq:covariate_shift} as the \emph{censoring mechanism}, denoted
by $\c(x; c_0)$. We write it as $\c(x)$ for brevity when no confusion
can arise. This is the conditional survival function of $C$ evaluated
at $c_0$. Under a censoring of Type I, the $C_i$'s are fully observed
while the $T_i$'s are only partially observed. Thus, $\p(C \mid X)$ is
typically far easier to estimate than $\p(T \mid X)$. Practically, the
censoring mechanism is usually far better understood than the
conditional survival function of $T$; for example, as mentioned in
Section \ref{subsec:setup}, in randomized clinical trials, $C$ often
solely depends on the site location.

Under the \compIC, the covariate shift even disappears since
$P_{X} = P_{X\mid C\ge c_0}$. In this case, we can apply a one-sided
version of conformal inference to obtain a calibrated LPB on
$T\wedge c_0$, and hence a calibrated LPB on $T$
\citep[e.g.,][]{vovk2005algorithmic, lei2018distribution,
  romano2019conformalized}. With infinite samples, as
$c_0\rightarrow \infty$, the method is tight in the sense that the
censoring issue disappears. Again, this result does not contradict
Theorem \ref{thm:distribution_free}, which requires the LPB to be
calibrated under the weaker condition
\eqref{eq:conditional_independent_censoring}. With finite samples,
there is a tradeoff between the choice of the threshold $c_0$ and the
size of the induced subpopulation.

%%% Local Variables:
%%% mode: latex
%%% TeX-master: "main"
%%% End:

\section{Conformal inference for censored outcomes}
\label{sec:method}
\subsection{Weighted conformal inference}\label{subsec:weighted_conformal_inference}
Returning to \eqref{eq:subpopulation_dist} and
\eqref{eq:whole_population_dist}, the goal is to construct an LPB
$\hat{L}(\cdot)$ on $T\wedge c_0$ from training samples
$(X_i, \td{T}_i\wedge c_0)_{C_i\ge c_0} = (X_i, T_i\wedge c_0)_{C_i\ge
  c_0}$ such that
\[
\p\big(T\wedge c_0\ge \hat{L}(X)\big)\ge 1 - \alpha.
\]
Since $T\wedge c_0\le T$, $\hat{L}(\cdot)$ is a calibrated LPB on
$T$. We consider $c_0$ to be a fixed threshold in Section
\ref{subsec:weighted_conformal_inference} and
\ref{subsec:double_robustness}, and discuss a data-adaptive approach
to choosing this threshold in Section \ref{subsec:c0}.

To deal with covariate shifts, \citet{tibshirani2019conformal}
introduced weighted conformal inference, which extends standard
conformal inference \citep[e.g.,][]{vovk2005algorithmic,
  shafer2008tutorial, lei2014distribution,
  foygel2019limits,barber2019predictive,sadinle2019least,
  romano2020classification,cauchois2020knowing}). Imagine we have
i.i.d.~training samples $(X_i, Y_i)_{i=1}^{n}$ drawn from a
distribution $P_{X}\times P_{Y\mid X}$ and wish to construct
prediction intervals for test points drawn from the target
distribution $Q_X \times P_{Y\mid X}$ (in standard conformal
inference, $P_X = Q_X$). Assuming $w(x) = dQ_{X}(x)/dP_{X}(x)$ is
known, then weighted conformal inference produces prediction intervals
$\hat{C}(\cdot)$ with the property
\begin{align}\label{eq:covshift_cover}
\p_{(X,Y)\sim Q_{X} \times P_{Y\mid X}}
\big(Y\in\hat{C}(X)\big)\ge 1 - \alpha.
\end{align}
Above, the probability is computed over both the training set and the
test point $(X, Y)$.  In our case, the outcome is $T\wedge c_0$ and
the covariate shift $w(x) = \p(C\ge c_0) / \c(x)$, as shown in
\eqref{eq:covariate_shift}. 

% \ejc{I prefer to speak of conformity scores. This is simpler. Please adjust.}~\llzr{Done}

\revise{In Algorithm \ref{algo:weighted_split}, we sketched a version
  of weighted conformal inference based on data splitting, which is
  adapted to our setting and has low computational
  overhead. Operationally, it has three main steps:}
\begin{enumerate}
\item split the data into a training and a calibration fold;
\item apply any prediction algorithm on the training fold to generate
  a \emph{conformity score} indicating how atypical a value of the
  outcome is given observed covariate values;
  % footnote
  here, we generate a conformity score such that a large value indicates a
    lack of conformity to training data.
\item calibrate the predicted outcome by the distribution of
  conformity scores on the calibration fold. In the calibration step
  from Algorithm \ref{algo:weighted_split}, $\quantile(1 - \alpha; Q)$
  is the $(1-\alpha)$ quantile of the distribution $Q$ defined as
\begin{equation}
  \label{eq:quantile}
  \quantile(1 - \alpha; Q) = \sup\{z: Q(Z\le z) < 1 - \alpha\}.
\end{equation}
\end{enumerate}

\begin{algorithm}[H]
  \DontPrintSemicolon  
  \SetAlgoLined
  \BlankLine
  \caption{\revise{conformalized survival analysis}\label{algo:weighted_split}}
  \textbf{Input:} level $\alpha$; \revise{data $\calZ=(X_i,\tT_i,C_i)_{i\in\calI}$}; testing point $x$;\\
  \hspace{0.08\textwidth}function $V(x,y;\calD)$ to compute the conformity score between $(x,y)$ and 
  data $\calD$; \\
  \hspace{0.08\textwidth}function $\hat{w}(x;\calD)$ to fit the weight function at $x$ using 
  $\calD$ as data;\\
  \hspace{0.08\textwidth}function $\calC(\calD)$ to select the threshold $c_0$ using 
  $\calD$ as data.\;
  \vspace*{.3cm}
  \textbf{Procedure:}\\
  \vspace*{.1cm}
  \hspace{0.02\textwidth}1. Split $\calZ$ into a training fold $\calZ_{\text{tr}} \triangleq (X_i,Y_i)_{i\in\calI_{\text{tr}}}$  and a calibration fold $\calZ_{\text{ca}} \triangleq (X_i,Y_i)_{i\in\calI_{\text{ca}}}$.\; 
  \hspace{0.02\textwidth}\revise{2. Select $c_0 = \calC(\calZ_{\tr})$ and let $\calI_{\ca}' = 
  \{i\in\calI_{\ca}: C_i \ge c_0\}$.}\;
  \hspace{0.02\textwidth}3. For each $i\in \calI'_{\text{ca}}$, compute the conformity score 
  \revise{$V_i = V(X_i,\tT_i \wedge c_0;\calZ_{\text{tr}})$}.\;
  \hspace{0.02\textwidth}4. For each $i\in\calI'_{\text{ca}}$, compute the weight $W_i = \hat{w}(X_i;\calZ_{\text{tr}})\in [0, \infty)$.\;
  \hspace{0.02\textwidth}5. Compute the weights $\hat{p}_i(x) = \frac{W_i}{\sum_{i\in\calI'_{\text{ca}}}W_i +\hat{w}(x;\calZ_{\text{tr}})}$
  and $\hat{p}_{\infty}(x) = \frac{\hat{w}(x;\calZ_{\text{tr}})}{\sum_{i\in\calI'_{\text{ca}}}W_i + \hat{w}(x;\calZ_{\text{tr}})}$.\;
  \hspace{0.02\textwidth}6. Compute $\eta(x) = \quantile\left(1 - \alpha; \sum_{i\in\calI'_{\text{ca}}} 
  \hat{p}_i(x)\delta_{V_i} + \hat{p}_{\infty}(x)\delta_{\infty}\right)$.\;
  \vspace*{.3cm}
  \textbf{Output}: $\hat{L}(x) = \inf\{y: V(x,y;\calZ_{\text{tr}}) \le \eta(x) \}\revise{\wedge c_0}$
\end{algorithm}

A few comments regarding Algorithm \ref{algo:weighted_split} are in
order. First, when the covariate shift $w(x)$ is unknown, it can be
estimated using the training fold. Second, note that in step 4, if
$\hat{w}(x; \calZ_{\tr}) = \infty$, then
$\hat{p}_{i}(x) = 0\,\, (i\in \calZ_{\ca})$ and
$\hat{p}_{\infty}(x) = 1$. In this case, step 5 gives
$\hat{L}(x) = -\infty$. Third, the requirement that
$W_i\in [0, \infty)$ is natural because $X_{i}\sim P_{X}$ and
$w(X)\in [0, \infty)$ almost surely under $P_{X}$ even if $Q_{X}$ is
not absolutely continuous with respect to $P_{X}$. Fourth, it is worth
mentioning in passing that $\eta(x)$ is invariant to positive
rescalings of $\hat{w}(x)$. Thus, we can set $w(x) = 1 / \hat{\c}(x)$
in our case where $\hat{c}(x)$ is an estimate of $\c(x)$.
\revise{Finally, apart from fitting $V(\cdot, \cdot; \Z_\tr)$ and
  $\hat{w}(\cdot; \Z_\tr)$ once on the training fold, the additional
  computational cost of our algorithm comes from computing
  $|\calI'_{\ca}|$ conformity scores and finding the $(1-\alpha)$-th
  quantile. We provide a detailed analysis of time complexity in
  Section~\ref{sec:time_comp} of the Appendix.}

In the algorithm, the conformity score function $V(x, y; \calD)$ can
be arbitrary and we discuss three popular choices from the literature: % \ejc{Why do we have diamonds?}~\llzr{We have changed it to dash. What do you think?}
\begin{itemize}
\item Conformalized mean regression (CMR) scores are defined via
  $V(x, y; \calZ_\tr) = \hat{m}(x)- y$, where $\hat{m}(\cdot)$ is an
  estimate of the conditional mean of $Y$ given $X$. The resulting LPB
  is then $(\hat{m}(x) - \eta(x))\wedge c_0$. This is the one-sided version of the
  conformity score used in \cite{vovk2005algorithmic} and
  \cite{lei2014distribution}.
\item Conformalized quantile regression (CQR) scores are
  defined via $V(x, y; \calZ_\tr) = \hat{q}_{\alpha}(x) - y$, where
  $\hat{q}_\alpha(\cdot)$ is an estimate of the conditional
  $\alpha$-th quantile of $Y$ given $X$. The resulting LPB is then
  $(\hat{q}_\alpha(x) - \eta(x))\wedge c_0$. This score was proposed by
  \cite{romano2019conformalized}; it is more adaptive than CMR and
  usually has better conditional coverage.
\item Conformalized distribution regression (CDR) scores
  are defined via
  $V(x, y; \calZ_\tr) = \alpha - \hat{F}_{Y\mid X = x}(y)$, where
  $\hat{F}_{Y\mid X = x}(\cdot)$ is an estimate of the conditional
  distribution of $Y$ given $X$. The resulting LPB is then
  $\hat{F}_{Y\mid X = x}^{-1}(\alpha - \eta(x))\wedge c_0$, or equivalently, the
  $(\alpha - \eta(x))$-th quantile of the estimated conditional
  distribution. This score was proposed by
  \cite{chernozhukov2019distributional}. It is particularly suitable
  to our problem because most survival analysis methods estimate
  the whole conditional distribution.
\end{itemize}

Under the \compIC, $\p(C\ge c_0\mid X) = \p(C\ge c_0)$ almost surely. As a consequence, we can set $\hat{w}(x) = w(x) \equiv 1$ and obtain a calibrated LPB without any distributional assumption. 
\begin{proposition}\label{prop:exact_LPB}[Corollary 1 of \cite{tibshirani2019conformal}]
  Let $c_0$ be any threshold independent of $\calZ_\ca$. Consider
  Algorithm \ref{algo:weighted_split} with $Y_i = T_i\wedge c_0$ and
  $\hat{w}(x; \calD)\equiv 1$. Under the \compIC, $\hat{L}(X)$ is
  calibrated.
\end{proposition}

\subsection{Doubly robust lower prediction bounds}
\label{subsec:double_robustness}

% \ejc{We do not talk about the case where $c(x)$ is known as in a
%   clinical trial?}~\llzr{Strictly speaking, $c(x)$ is unknown even in a clinical trial. Usually, the protocol specifies the sample size required per month/year and the lab tries to recruit the patients following the protocol. However, the actual patient recruitment process may not be the same as the design. But we can add a proposition on the case where $c(x)$ is known. What do you think?}

Under the more general \condIC, the censoring mechanism needs to be
estimated. We can apply any distributional regression techniques such
as the kernel method or the newly invented distribution boosting
\citep{friedman2020contrast} to estimate
$\c(x) = \p(C\ge c_0\mid X = x)$. For two-sided weighted split-CQR,
\cite{lei2020conformal} prove that the intervals satisfy a doubly
robust property which states the following: the average coverage is
guaranteed if either the covariate shift or the conditional quantiles
are estimated well, and the conditional coverage is approximately
controlled if the latter is true. In Section \ref{app:double_robustness} 
in the Appendix, \revise{we present more general results, both
non-asymptotic and asymptotic, that are applicable to a broad
class of conformity scores proposed by \cite{gupta2019nested},
including the CMR-, CQR- and CDR-based scores.}

In this section, we first present a version of the asymptotic result
tailored to the CQR-LPB for simplicity.
\begin{theorem}\label{thm:double_robustness_CQR}
  % \ejc{You have to say that you use $\hat{c}$ and where you got it.}~\llzr{Yes we added that below.}
  Let $N = |\calZ_\tr|, n = |\calZ_\ca|$, $c_0$ be any threshold
  independent of $\calZ_\ca$, and $q_{\alpha}(x; c_0)$ denote the $\alpha$-th conditional quantile of
  $T\wedge c_0$ given $X = x$. Further, % \ejc{We cannot overload symbols
    % like this: $q_{\alpha}(x)$ was something else a few pages ago.}~\llzr{We changed $q_{\alpha}(x)$ to $q_{\alpha}(x; c_0)$ and $\hat{q}_{\alpha}(x)$ to $\hat{q}_{\alpha}(x; c_0)$.}
  let $\hat{\c}(x)$ and $\hat{q}_{\alpha}(x; c_0)$ be estimates of $\c(x)$ and $q_{\alpha}(x; c_0)$ respectively using $\calZ_\tr$, and $\hat{L}(x)$ be the corresponding
  CQR-LPB. Assume that there exists $\delta > 0 $ such that
  $\E\left[1 / \hat{\c}(X)^{1 + \delta}\right] < \infty$ and
  $\E[1 / \c(X)^{1 + \delta}] < \infty$. Suppose that either A1 or A2
  (or both) holds:
\begin{enumerate}[label = A\arabic*]
  \item $\underset{N\rightarrow \infty}{\lim} \E\Big[ \big|1 / \hat{\c}(X) - 1 / \c(X) \big|\Big] = 0$.
  \item \begin{enumerate}[label = (\roman*)]
      \item There exists $b_2 > b_1 > 0$ and $r > 0$ such that, for any $x$ and $\eps \in [0, r]$, 
        \revise{\[\p(T\wedge c_0\ge q_{\alpha}(x; c_0) + \eps\mid X = x)\in {[}1 - \alpha - b_2\eps, 1 - \alpha - b_1 \eps{]}, \quad \text{if }q_{\alpha}(x; c_0) + \eps < c_{0}.\]}
      \item $\underset{N\rightarrow\infty}{\lim}\E\left[\calE(X) / \hat{c}(X)\right] = 
        \underset{N\rightarrow\infty}{\lim}\E\left[\calE(X) / c(X)\right] = 0$, where $\calE(x) = |\hat{q}_{\alpha}(x; c_0) - q_{\alpha}(x; c_0)|$.
      \end{enumerate}
  \end{enumerate}
  Then 
  \begin{align}
    \label{eq:double_robust_marginal_LPB}
    \lim_{N,n\rightarrow\infty} \p\left( T\wedge c_0\ge \hat{L}(X)\right) \ge 1 - \alpha.
  \end{align}
  Furthermore, under A2, for any $\varepsilon > 0$,
  \begin{align}\label{eq:conditional_coverage_LPB}
  \lim_{N,n\rightarrow \infty} \p\Big(\E\big[\Indc\big\{T\wedge c_0\ge \hat{L}(X)\big\} \mid X\big] > 1-\alpha-\varepsilon \Big) = 1.
  \end{align}
\end{theorem}

\revise{
  \begin{remark}
    The condition A2 (i) holds if $\,T$ has a bounded and absolutely
    continuous density conditional on $X$ in a neighborhood of
    $q_{\alpha}(x)$. In fact, noting that
    $q_{\alpha}(x; c_0) = q_{\alpha}(x) \wedge c_0$, when
    $q_{\alpha}(x; c_0) + \eps \le c_0$, we have
    $q_{\alpha}(x) \le c_0$ and thus
    $T\wedge c_0 \ge q_{\alpha}(x) \wedge c_0$ if and only if
    $T\ge q_{\alpha}(x)$.
  \end{remark}
}

Intuitively, if $\hat{c}(x)\approx c(x)$, then the procedure
approximates the oracle version of weighted split-CQR with the true
weights, and the LPBs should be approximately calibrated. On the other
hand, if $\hat{q}_{\alpha}(x; c_0)\approx q_{\alpha}(x; c_0)$, then
$V_i \approx q_{\alpha}(X_i; c_0) - T_i\wedge c_0$. As a result,
\[\p(V_i \le 0\mid X_i)\approx \p(T_i\wedge c_0\le q_{\alpha}(X_i; c_0)\mid X_i) = \alpha.\]
Thus, the $(1 - \alpha)$-th quantile of the $V_i$'s conditional on
$\calZ_\tr$ is approximately $0$. To keep on going, recall that
$\eta(x)$ is the $(1 - \alpha)$-th quantile of the random distribution
$\sum_{i\in\calZ_{\ca}}\hat{p}_{i}(x)\delta_{V_i} +
\hat{p}_{\infty}(x)\delta_{\infty}$, and set $G$ to be the cumulative
distribution function of this random distribution. Then,
\[G(0) \approx \E[G(0)\mid \calZ_\tr] = \sum_{i\in \calZ_{\ca}}\hat{p}_i(x)\p(V_i\le 0\mid \calZ_\tr)\approx \sum_{i\in \calZ_{\ca}}\hat{p}_i(x)(1 - \alpha)\approx 1 - \alpha,\]
implying that $\eta(x)\approx 0$. Therefore, $\hat{L}(x)\approx q_{\alpha}(x; c_0)$, which approximately achieves the desired conditional coverage. 

With the same intuition, we can establish a similar result for the
CDR-LPB with a slightly more complicated version of Assumption
A2.
\begin{theorem}\label{thm:double_robustness_CDR}
Let $F(\cdot\mid x)$ denote the conditional distribution of $T\wedge c_0$ given $X = x$. With the same settings and assumptions as in Theorem \ref{thm:double_robustness_CQR}, the same conclusions hold if A2 is replaced by the following conditions:
  \begin{enumerate}[label = (\roman*)]
  \item there exists $r > 0$ such that, for any $x$ and $\eps \in [0, r]$, 
    \revise{\[\p(T\wedge c_0\ge q_{\alpha + \eps}(x; c_0)\mid X = x) = 1 - \alpha - \eps, \quad \text{if }q_{\alpha + \eps}(x; c_0) < c_0.\]}
  \item $\underset{N\rightarrow\infty}{\lim}\E\left[\calE(X) / \hat{c}(X)\right] = 
    \underset{N\rightarrow\infty}{\lim}\E\left[\calE(X) / c(X)\right] = 0$, where
    \[\calE(x) = \sup_{s\in [\alpha - r, \alpha + r]}|F(\hat{q}_{s}(x; c_0) \mid x) - F(q_{s}(x; c_0)\mid x)|.\]
  \end{enumerate}
\end{theorem}

The double robustness of weighted split conformal inference has some
appeal; indeed, the researcher can leverage knowledge about both the
conditional survival function and the censoring mechanism without any
concern for which is more accurate. Suppose the Cox model is adequate
in a randomized clinical trial; then it can be used to produce
$\hat{q}_{\alpha}(x; c_0)$ in conjunction with the known censoring
mechanism. If the model is indeed correctly specified, the LPB is
conditionally calibrated, as are classical prediction intervals
derived from the Cox model \citep{kalbfleisch2011statistical}; if the
model is misspecified, however, the LPB is still calibrated.

\revise{
  \begin{remark}
    A special case is when the completely independent censoring
    assumption holds, yet the researcher is unaware of this and still
    applies the estimated $\hat{c}(\cdot)$ to obtain the prediction
    intervals.  As implied by Theorem~\ref{thm:double_robustness_CQR}
    and~\ref{thm:double_robustness_CDR}, if $\hat{c}(\cdot)$ is
    approximately a constant function, the prediction interval is
    approximately calibrated. Notably, even if $\hat{c}(\cdot)$
    deviates from a constant, our prediction interval still achieves
    coverage as long as the estimated weights are non-decreasing in
    the conformity scores. We present this additional robustness
    result in Section~\ref{appx:coverage_independent_censoring} of the
    Appendix.
\end{remark}
}

\revise{ As a concluding remark, the prediction interval can become
  numerically and statistically unstable in the presence of extreme
  weights since the proposed method depends on $c(x)$ (or the
  estimated $\hat{c}(x)$) through its inverse. The reader may have
  observed that $c(x)$ plays a role similar to that of the propensity
  score in causal inference; the reweighting step in Algorithm
  \ref{algo:weighted_split} is analogous to inverse propensity score
  weighting-type methods. Assumption A1 in Theorem
  \ref{thm:double_robustness_CQR} mimics the overlap condition
  \citep[e.g.][]{d2021overlap} in the causal inference
  literature. That said, there is a crucial difference. In a typical
  causal setting, the overlap condition is an assumption about the
  unknown data generating process, which cannot be manipulated. In
  contrast, in our work Assumption A1 can always be satisfied by
  selecting a sufficiently low threshold $c_0$.
  % For instance, imagine
  % there is prior information on the censoring mechanism such that for
  % some $c^*$, $\p(C \ge c^* \mid X)$ is known to be bounded away from
  % $0$ uniformly over $X$. Then we may select $c_0 \le c^*$. When there
  % is no prior information, we can still use the training data to learn
  % $c^*$. \ejc{This is vague and I am not sure what people we will get
  %   out of the last sentence.}
  We provide a detailed discussion in
  Section~\ref{appx:c_up} of the Appendix.}
% what is  
% different in our case is that the quantity $c(x)$ depends
% on $c_0$, which we have the freedom to select---indeed, 
% we only need $c_0$ to be ``not too large'' in order to
% avoid the extremely small $c(x)$.
% We provide
% our guidance on the selection of (upper bounds) of $c_0$
% in the following section.

% \paragraph{Avoiding irregular weights}
% \revise{As discussed before, another important aspect 
% in choosing $c_0$ is to avoid extremely small
% $\hat{c}(x)$. When there is prior information
% on the censoring mechanism such that we know
% for some $c^*$, $\p(C \ge c^* \mid X) \ge \eta$,
% for some $\eta >0$ almost surely, we may select
% $c_0 \le c^*$; when there is no prior information,
% under certain conditions we can still use the training 
% data to learn such a $c^*$. We provide the details 
% in Section~\ref{appx:c_up}
% of the supplementary material.}

\revise{\subsection{Adaptivity to high-quality modeling} We have seen
  that when the quantiles of survival times are well estimated,
  $\hat{L}(x)\approx q_{\alpha}(x; c_0)$, which is the oracle lower
  prediction bound for $T\wedge c_0$, had the true survival function
  been known. This holds without knowing whether the survival model is
  well estimated or not. This suggests that conformalized survival
  analysis has favorable adaptivity properties, as formalized below. 
  \begin{theorem}\label{thm:adaptivity}
    \begin{enumerate}
    \item Under the settings and assumptions of Theorem \ref{thm:double_robustness_CQR}, assume further that A2 (ii) holds and a modified version of A2 (i) holds: there exists $b_1 > 0$ and $r > 0$ such that, for any $x$ and $\eps \in [0, r]$, 
        \[\p(T\wedge c_0\ge q_{\alpha}(x; c_0) - \eps\mid X = x)\ge 1 - \alpha + b_1 \eps.\]
      Then, for any $\eps > 0$,
      \[\lim_{N, n\rightarrow \infty}\p_{X\sim Q_X}(\hat{L}(X) \ge q_{\alpha}(X; c_0) - \eps) = 1.\]
    \item Under the settings and assumptions of Theorem \ref{thm:double_robustness_CDR}, assume further the condition (ii) and the modified version of condition (i): there exists $r > 0$ such that, for any $x$ and $\eps \in [0, r]$,
      \[\p(T\wedge c_0\ge q_{\alpha - \eps}(x; c_0)\mid X = x) \ge 1 - \alpha + \eps.\]
      Then, for any $\eps > 0$,
      \[\lim_{N, n\rightarrow \infty}\p_{X\sim Q_X}(\hat{L}(X) \ge q_{\alpha - \eps}(X; c_0)) = 1.\]
    \end{enumerate}
  \end{theorem}

  In theory, if $c_0$ is allowed to grow with $n$ and $C$ exceeds
  $c_0$ with sufficient probability, then
  $\hat{L}(x)\approx q_{\alpha}(x)$ (see Appendix
  \ref{app:adaptivity}). In practice, it would however be wiser to
  tune $c_0$ in a data-adaptive fashion (discussed in the next
  subsection) than to prescribe a predetermined growing sequence. }

\subsection{Choice of threshold}
\label{subsec:c0}

% \paragraph{Optimizing efficiency}
The threshold $c_0$ induces an estimation-censoring tradeoff: a larger
$c_0$ mitigates the censoring effect, closing the gap between the target
outcome $T$ and the operating outcome $T\wedge c_0$, but reduces
the sample size to estimate the censoring mechanism and the
conditional survival function.  It is thus important to pinpoint the
optimal value of $c_0$ to maximize efficiency.

To avoid double-dipping, we choose $c_0$ on the training fold
$\calZ_\tr$. In this way, $c_0$ is independent of the calibration fold
$\calZ_\ca$ and we are not using the same data twice. In particular,
Proposition \ref{prop:exact_LPB}, Theorem
\ref{thm:double_robustness_CQR} and \ref{thm:double_robustness_CDR}
all apply. Concretely, we (1) set a grid of values for $c_0$, (2)
randomly sample a holdout set from $\calZ_\tr$, (3) apply Algorithm
\ref{algo:weighted_split} on the rest of $\calZ_\tr$ for each value of
$c_0$ to generate LPBs for each unit in the holdout set, and (4)
select $c_0$ which maximizes the average LPBs on the holdout set. One
way to see all of this is to pretend that the training fold is the
whole dataset and measure efficiency as the average realized LPBs. In
practice, we choose $25\%$ units from $\calZ_\tr$ as the holdout
set. The procedure is convenient to implement, though it is by no
means the most powerful approach. 
%We leave the investigation of more
%principled selection procedures to future research.

\revise{Under suitable conditions, we can choose $c_0$ by using the
  calibration fold $\calZ_{\ca}$ and have the resulting LPBs still be
  (approximately) calibrated. To be specific, given a candidate set
  $\calC$ for $c_0$, we simply maximize the average LPB on
  $\calZ_\ca$:
\begin{align*}
  \hat{c}_0 = \argmax{c_0 \in \calC}~\frac{1}{|\calI_{\ca}|}
  \sum_{i\in\calI_{\ca}} \hat{L}_{c_0}(X_i),
\end{align*}
where $\hat{L}_{c_0}(X)$ is given by the conformalized survival
analysis with the threshold $c_0$. In Section~\ref{sec:uniform_c} of
the Appendix we derive uniform results for the $c_0$'s
in $\calC$, and prove coverage guarantees for $L_{\hat{c}_0}(X)$ via a
generalization of the techniques for unweighted conformal inference by
\cite{yang2021finite}.}

%%% Local Variables:
%%% mode: latex
%%% TeX-master: "main"
%%% End:

\section{Simulation studies}
\label{sec:sim}
In this section, we design simulation studies to evaluate the
performance of our method.  Specifically, we run four sets of
experiments detailed in Table \ref{tab:params_1}. In each experiment,
we compare the CQR- and CDR-LPB with the following alternatives:
\begin{itemize}
  \item Cox model: we generate the LPB as the $\alpha$-th quantile from an estimated Cox model. The method is implemented via the \texttt{survival} R-package \citep{survival-package}.
  \item Accelerated failure time (AFT) model: we generate the LPB as
    the $\alpha$-th quantile from an estimated AFT model with Weibull
    noise. The method is implemented in the \texttt{survival} R
    package.
  \item Censored quantile regression: we consider three variants of quantile regression methods, proposed by \citet{powell1986censored}, \citet{portnoy2003censored}, and \citet{peng2008survival}, respectively. All three procedures are implemented in the \texttt{quantreg} R package \citep{roger2020quantreg}.
  \item Censored quantile regression forest \citep{li2020censored}: this is a variant of quantile random forest \citep{athey2019generalized} designed to handle time-to-event outcomes. We reimplement the method based on the code provided in \url{https://github.com/AlexanderYogurt/censored\_ExtremelyRandomForest}.
  \item Naive CQR: we apply split-CQR \citep{romano2019conformalized} naively to $(X_i,\tT_i)_{i=1}^n$, where the quantiles are estimated by the \texttt{quantreg} R package.
\end{itemize}
For the CQR-LPB, the conditional quantiles are estimated via censored
quantile regression forest or distribution boosting~\citep{friedman2020contrast}; for the CDR-LPB, the conditional survival function is estimated via distribution boosting, which is implemented in
the R package \texttt{conTree}~\citep{Jerome2020conTree}. % \ejc{What would happen to CQR-LPB if we used quantiles extracted from distribution boosting?}
% \llzr{The results from applying CQR-LPB with distributiona boosting is added. To make the name shorter, we are calling them ``CQR-conTree'', ``CDR-conTree'', and so on...
% what do you think?}
 
In each experiment, we generate $200$ independent datasets, each
containing a training set of size $n = 3000$, and a test set of size
$n = 3000$. For conformal methods, $50\%$ of the training set is used for
fitting the predictive model, and the remaining $50\%$ of the training set is reserved for
calibration.
% footnote
The splitting ratio between the training set and
  the test set is slightly different from the recommendation
  by~\citet{sesia2020comparison}, where they suggest using $75\%$ of
  the data for training and $25\%$ for calibration.  We reserve more
  data for calibration to ensure there are still enough samples in the
  calibration set after the selection and to decrease the variability
  of the LPBs.
  We then evaluate the coverage of LPBs as
$(1 / n_{\text{test}})\sum_{i=1}^{n_{\text{test}}} \Indc\{T_i \ge
\hat{L}(X_i)\}$.  All the results in this section can be replicated
with the code available at
\url{https://github.com/zhimeir/cfsurv\_paper}.  In addition, the
proposed CQR- and CDR-LPB are implemented in the R package
\texttt{cfsurvival}, available at
\url{https://github.com/zhimeir/cfsurvival}.

The covariate vector $X\in\R^p$ is generated from $P_X$. The survival
time $T$ is generated from an AFT model with Gaussian noise, i.e.
\begin{align*}
  \log{T} \mid X \sim \calN \Big( \mu(X) , \sigma^2(X) \Big).
\end{align*}
We consider $2\times 2$ settings with univariate or multivariate
covariates plus homoscedastic or heteroscedastic errors.
% footnote
Here the term ``homoscedastic'' or ``heteroscedastic'' is applied to
  $\log{T}$. The choice of the parameters in each setting is
specified in Table~\ref{tab:params_1}.

Finally, we apply all the methods with target coverage level
$1 - \alpha = 90\%$. In each experiment, we estimate $c(x)$ by
distribution boosting.

\begin{table}
  \caption{\label{tab:params_1}
		Parameters used in the simulation study. ``Homosc.''~and
    ``Heterosc.''~are short for homoscedastic and heteroscedastic;
    ``Uvt.''~and ``Mvt.''~are short for univariate and
    multivariate. $\calU(a,b)$ denotes the uniform distribution
    supported on $[a,b]$; $\calE(\lambda)$ denotes the exponential
    distribution with rate $\lambda$.
    }
  \centering
  \begin{tabular}{c|ccccc}
     & dimension $p$ & $P_X$ &  $P_{C\mid X}$ & $\mu(x)$ & $\sigma(x)$\\\hline
    Uvt. + Homosc. & $1$ & $\calU(0,4)$ & $\calE(0.4)$ &  $2+0.37\sqrt{x}$ &  $1.5$ \\
   Uvt. + Heterosc. & $1$ & $\calU(0,4)$ & $\calE(0.4)$ &$2+0.37\sqrt{x}$ & $1+x/5$\\ 
    Mvt. + Homosc. & $100$ & $\calU([-1,1]^{p})$  
                   & $\calE(0.4)$ &  $\log{2}+1+0.55(x_1^2-x_3x_5)$ & $1$\\
    Mvt. + Heterosc. & $100$ & $\calU([-1,1]^p)$  & $\calE(0.4)$
                     & $\log{2}+1+0.55(x_1^2-x_3x_5)$ & $|x_{10}|+1$\\
  \end{tabular}
\end{table}
Figure~\ref{fig:sim1_marginal_coverage} presents the empirical
coverage of the LPBs on uncensored survival times. Censored random
forests, the Cox model, the AFT model, and the three quantile regression
methods fail to achieve the target coverage in most cases. On the
other hand, the naive CQR attains the desired coverage but at the price of
being overly conservative. In contrast, both the CQR- and CDR-LPB
achieve near-exact marginal coverage, as predicted by our theory.

\begin{figure}[ht]
  \centering
  \includegraphics[width = \textwidth]{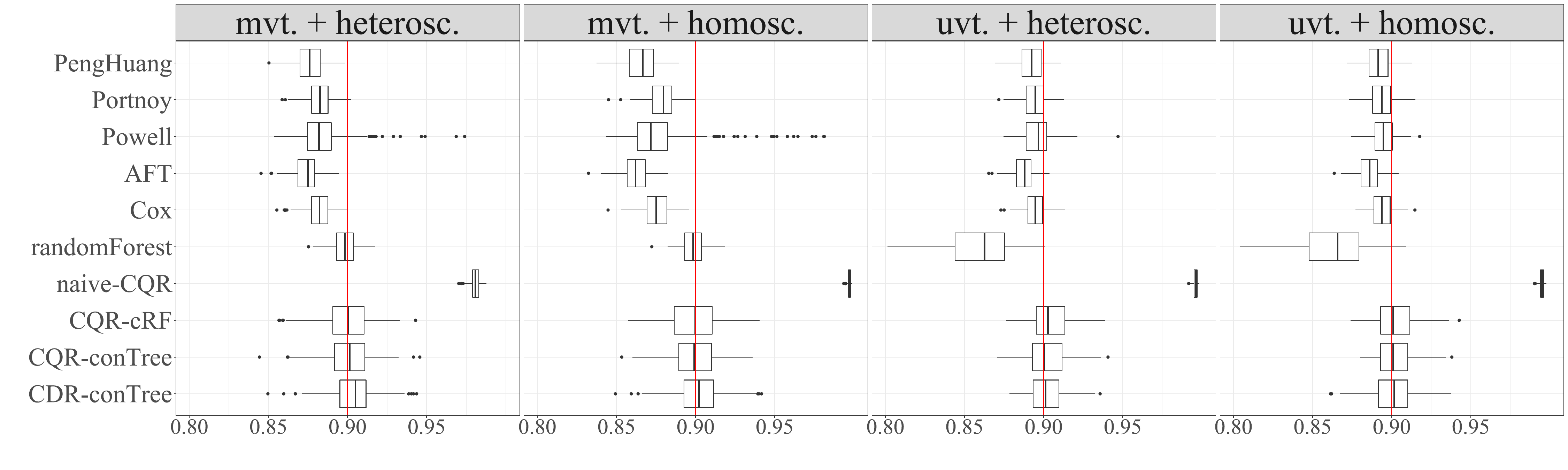}
  \caption{Empirical $90\%$ coverage of the uncensored survival time $T$. ``CQR-cRF'' is short for the CQR-LPB with censored quantile regression forest; ``CQR-conTree'' and ``CDR-conTree'' are short for the CQR- and CDR-LPB with distribution boosting. The other abbreviations are the same as in Table~\ref{tab:params_1}.
  }
  \label{fig:sim1_marginal_coverage}
\end{figure}

Next, we investigate the conditional coverage and efficiency of these
methods.  In Figure~\ref{fig:sim1_conditional_coverage}(a), we plot
the empirical conditional coverage as a function of the conditional
variance of $T$ on $X$. In particular, we stratify the data into $10$
groups based on equispaced percentiles of $\mathrm{Var}(T\mid X)$ and
plot the average coverage within each stratum along with a $90\%$
confidence band obtained via repeated sampling. Note that in either
the homoscedastic or the heteroscedastic case, $\mathrm{Var}(T\mid X)$
is varying with $X$. Not surprisingly, the naive CQR is conditionally
conservative. In the univariate case, both the CQR- and CDR-LPB
approximately achieve desired conditional coverage; in the
multivariate case, the conditional coverage is slightly uneven, though
still concentrating around the target
line. Figure~\ref{fig:sim1_conditional_coverage}(b) presents the ratio
between the LPBs and the true $\alpha$-th conditional quantile as a
function of $\text{Var}(T \mid X)$. This is a measure of efficiency
since the true conditional quantile is the oracle LPB. Here, we
observe that naive CQR-LPBs are close to zero, confirming that they
are overly conservative, while the CQR- and CDR-LPBs are fairly close
to the oracle LPB, implying that both methods are relatively
efficient.

\begin{figure}[ht]
  \begin{minipage}{.5\linewidth}
    \centering
    \includegraphics[width = \textwidth]{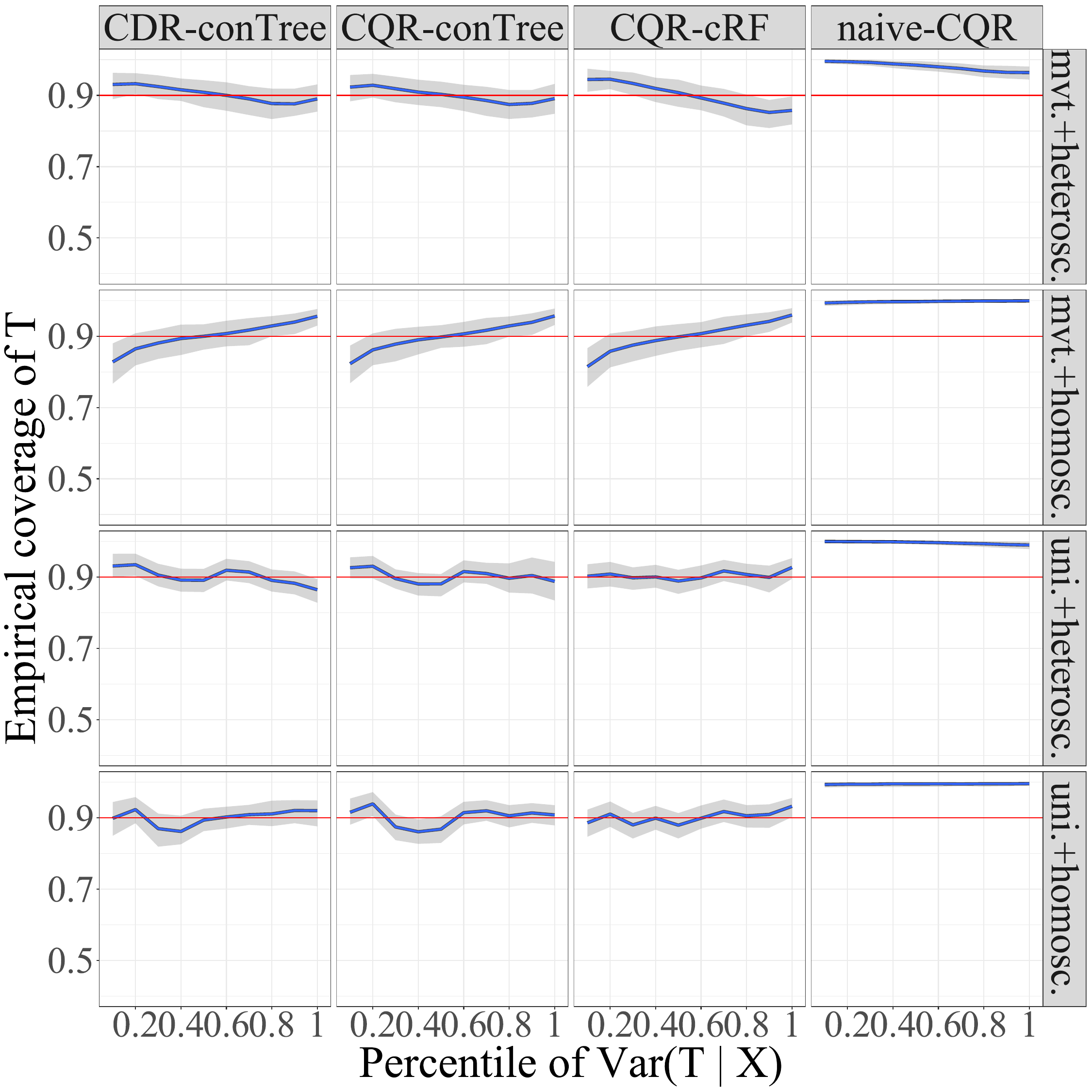}
    (a)
  \end{minipage}
  \hfill
  \begin{minipage}{.5\linewidth}
    \centering
    \includegraphics[width = \textwidth]{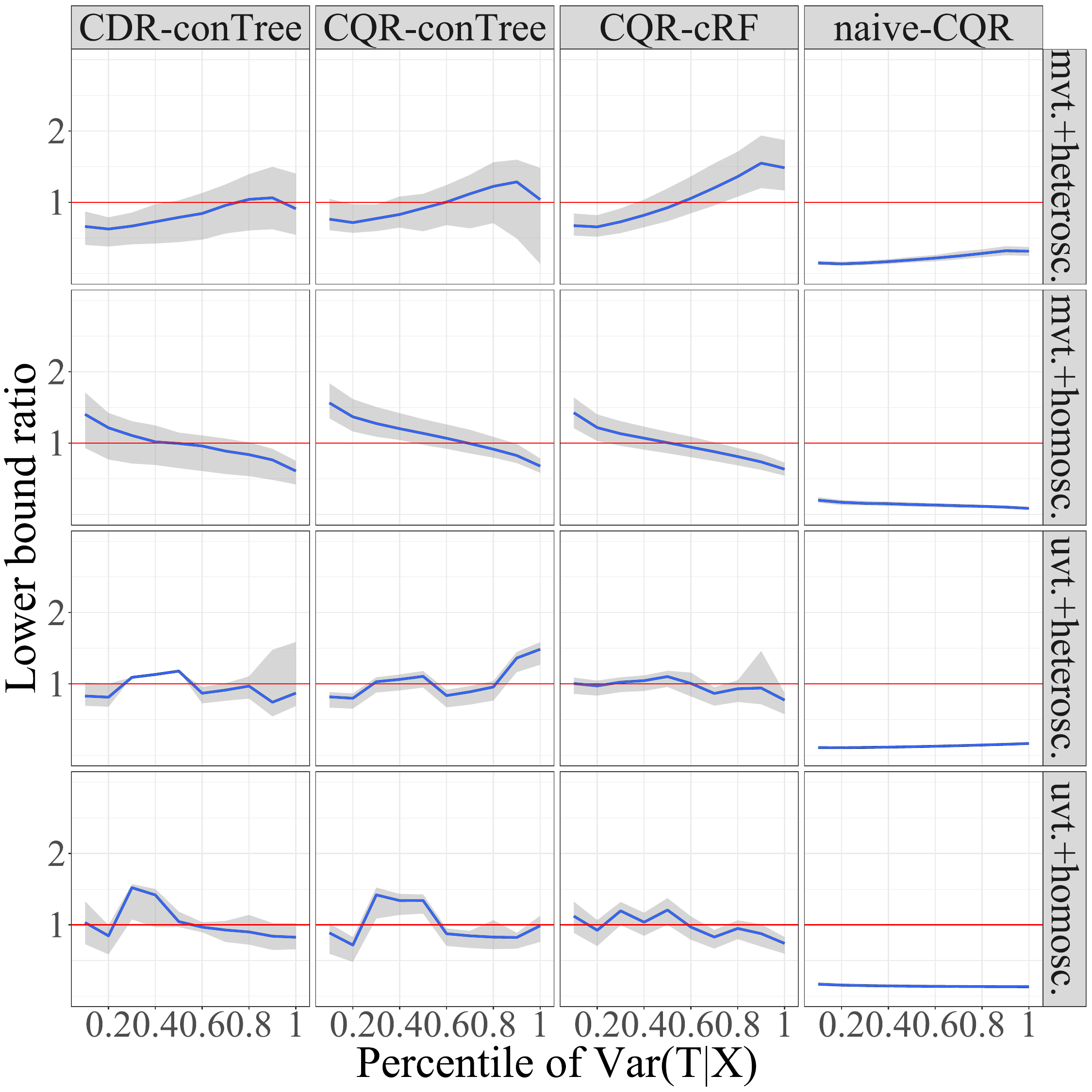}
    (b)
  \end{minipage}
  \caption{Results from the experiments detailed in Table \ref{tab:params_1}: (a) empirical $90\%$ conditional coverage and (b) ratio between the LPB and the theoretical quantile as a function of $\text{Var}(T\mid X)$. The blue curves correspond to the mean coverage in (a) and the median ratio in (b). The gray confidence bands correspond to the $95\%$ and $5\%$ quantiles of the estimates over repeated sampling. The abbreviations are the same as in Figure~\ref{fig:sim1_marginal_coverage}.}
    \label{fig:sim1_conditional_coverage}
\end{figure}

%%% Local Variables:
%%% mode: latex
%%% TeX-master: "main"
%%% End:

\section{Application to UK Biobank COVID-19 data}
\label{sec:application}
%!TEX root = mainEC.tex
We apply our method to the UK Biobank COVID-19 dataset to demonstrate
robustness and practicability.  UK Biobank~\citep{bycroft2018uk} is a 
large-scale biomedical database and research resource, containing
in-depth genetic and health information from half a million UK
participants. In April 2020, UK Biobank started to release COVID-19
testing data, and has since continued to regularly provide
updates. This gives researchers access to a cohort of COVID-19
patients, along with their date of confirmation, survival status,
pre-existing conditions, and other demographic covariates.

We include in our analysis all individuals in UK Biobank who received
a positive COVID-19 test result before January 21st, 2021.  This
results in a dataset of size $n=14,\!861$ with $484$ events, defined
as a COVID-related death.  We extract eight covariate features,
namely, \emph{age, gender, body mass index (BMI), waist size,
  cardiovascular disease status, diabetes status, hypothyroidism
  status, and respiratory disease status}. As in
Section~\ref{sec:setup}, the censoring time is the time lapse between
the date of a positive test and January 21st, 2021. The survival time
is the time lapse between the date of a positive test and the event
(which may have yet to occur).

We wish to harness this data to produce an LPB on the survival time of
each COVID-19 patient. To apply the CQR- or CDR-LPB, we set the
threshold $c_0$ to be $14$ days. Since survival time assessment likely
informs high-stakes decision-making, we set the target level to $99\%$
for reliability.

\subsection{Semi-synthetic examples}

To demonstrate robustness, we start our analysis with two
semi-synthetic examples so that the ground truth is known and
calibration can be assessed (results on real outcomes are presented
next). We keep the covariate matrix $X$ from the UK Biobank COVID-19
data.  In the first simulation study, we substitute the censoring time
with a synthetic $C$. In the second, each
survival time, observed or not, is substituted with a synthetic version. Details follow:  
\begin{itemize}
 \item {\em Synthetic $C$}: we take the censored survival time $\tT$
   as the uncensored survival time and generate the censoring time
   $C_{\text{syn}}$ as
   \begin{align*}
      C_{\text{syn}}  \sim \calE(0.001\cdot \text{age} + 0.01\cdot \text{gender}).
   \end{align*}
   In this setting, the observables are $(X, C_{\text{syn}}, \tT \wedge C_{\text{syn}})$, and
   we wish to construct LPBs on $\tT$.
\item {\em Synthetic $T$}: we keep the real censoring time $C$, and
  generate a survival time $T_{\text{syn}}$ as: % \ejc{$T'$ may not be good notation. How about $T_{\text{syn}}$?}~\llzr{Good idea!}
   \begin{align*}
     \log{T_{\text{syn}}} \mid X \sim \calN(2 + 0.05\cdot \text{age} + 0.1 \cdot\text{gender}, 1). 
   \end{align*}
   In this setting, the observables are $(X, C, T_{\text{syn}}\wedge C)$, and we
   wish to construct LPBs on $T_{\text{syn}}$.
\end{itemize}

\begin{figure}[ht]
  \centering
  \includegraphics[width = 0.8\textwidth]{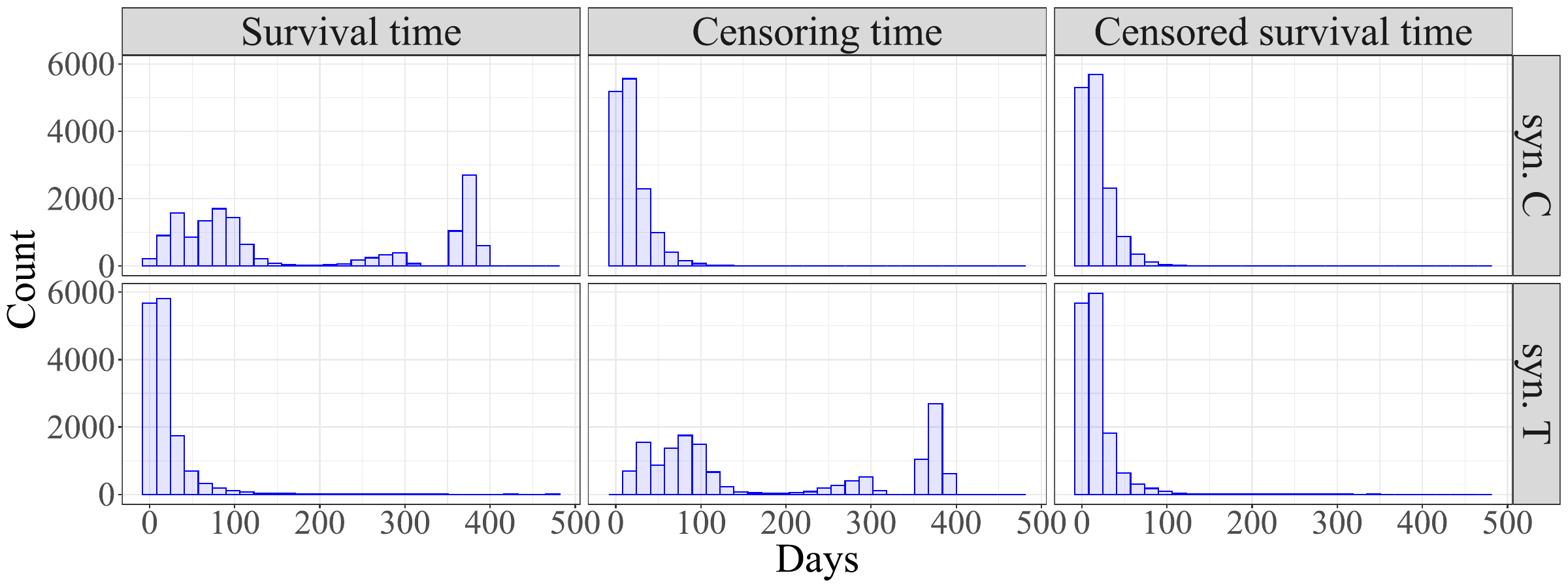}
  \caption{Histograms of the survival time, censoring time, and
    censored survival time defined as the minimum between the two, in
    each simulation setting. % \ejc{Remove $T$, $C$ and $\min(T,C)$ as
      % this is not consistent with your earlier notation. Call is
      % survival time, censoring time and censored survival time. Also,
      % since you discuss the $T$-setting before the $C$-setting, I
      % don't understand why the presentation then reverses the order.}~\llzr{We reversed the order of the two settings and 
      % changed the labels in the plot.}
  }
  \label{fig:histogram_synthetic_data}
\end{figure}

Figure~\ref{fig:histogram_synthetic_data} shows the histograms of the
survival time, censoring time, and censored survival time from the two
simulated datasets. We apply the CDR-LPB (with $c_0 = 14$) to
both. For comparison, we also apply the AFT and naive CQR. To evaluate
the LPBs, we randomly split the data into a training set with $75\%$
of the data and a holdout set with the remaining $25\%$. Each method
is applied to the training set, and the resulting LPBs are evaluated
on the holdout set. We repeat the above procedure $100$ times to
create 100 pairs of training and test data sets.

To visualize conditional calibration, we fit a Cox model on the data
to generate a predicted risk score for each unit and stratify all
units into $10$ subgroups defined by deciles of the predicted
risk. The results for synthetic $C$ and $T$ are plotted in
Figures~\ref{fig:synthetic_c} and~\ref{fig:synthetic_t},
respectively. % For each setting, we plot the marginal and conditional
% coverage of all three methods and the median value of CDR-LPBs within
% each group, averaged across $100$ independent samples.
As in the simulation studies from Section \ref{sec:sim}, we see that
the naive CQR is overly conservative. Notably, although the AFT-LPB is
well calibrated in the synthetic-$C$ setting, this method is overly conservative in the synthetic-$T$ setting, even though the model
is correctly specified. In contrast, the CDR-LPB is calibrated in both
examples. From the middle panels of Figures \ref{fig:synthetic_c} and
\ref{fig:synthetic_t}, we also observe that the CDR-LPB is
approximately conditionally calibrated. Finally, the right panels show
that CDR-LPB nearly preserves the rank of the predicted risk given by
the Cox model. The flat portion of the LPB towards the left end corresponds to the threshold, implying that at least $99\%$ of people with predicted risk scores lower than $0.5$ can survive beyond $14$ days. % \ejc{I think that you should comment on the flat portion of the LPB 
 % since this corresponds to the threshold.}~\llzr{Done.}
\begin{figure}[ht]
  \centering
  \begin{minipage}{0.3\textwidth}
    \includegraphics[width = \textwidth]{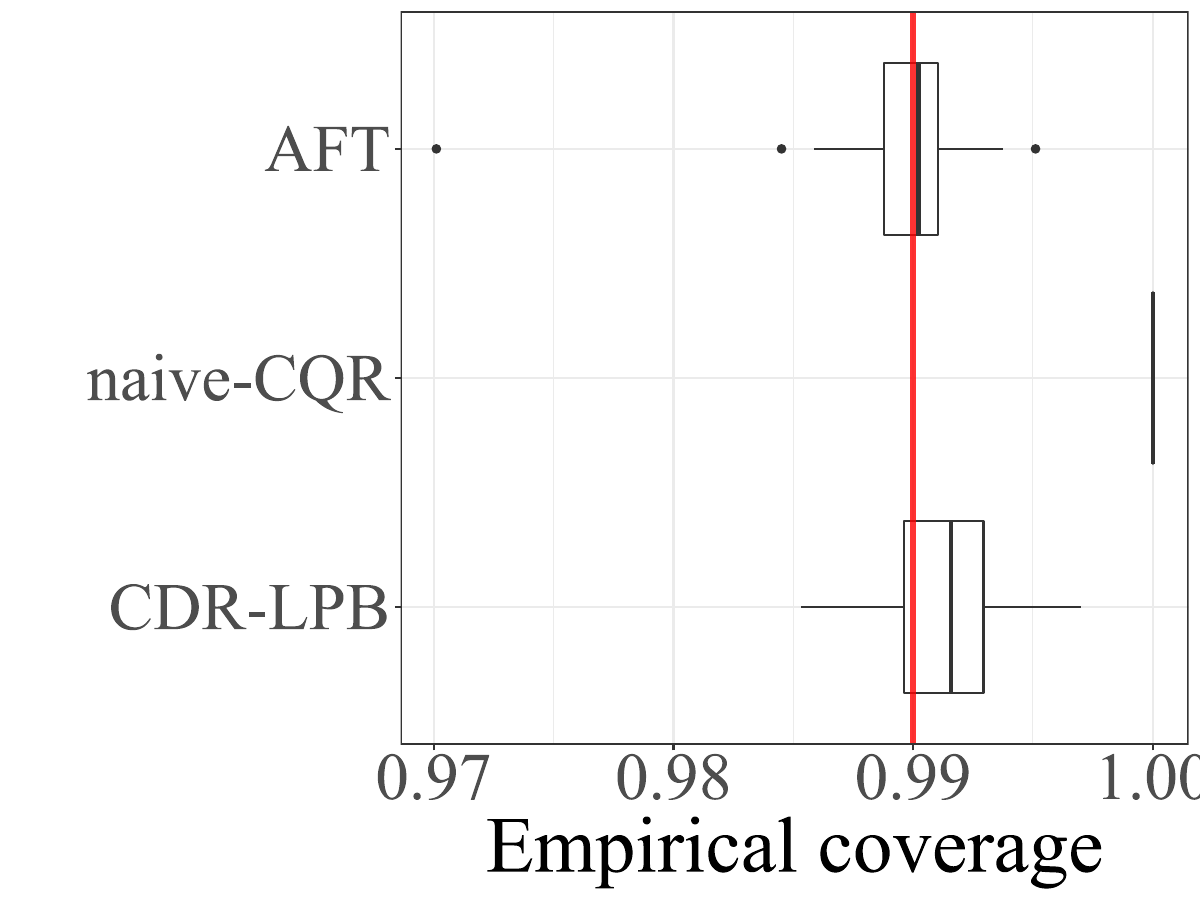}
  \end{minipage}
  \hfill
  \begin{minipage}{0.3\textwidth}
    \includegraphics[width =\textwidth]{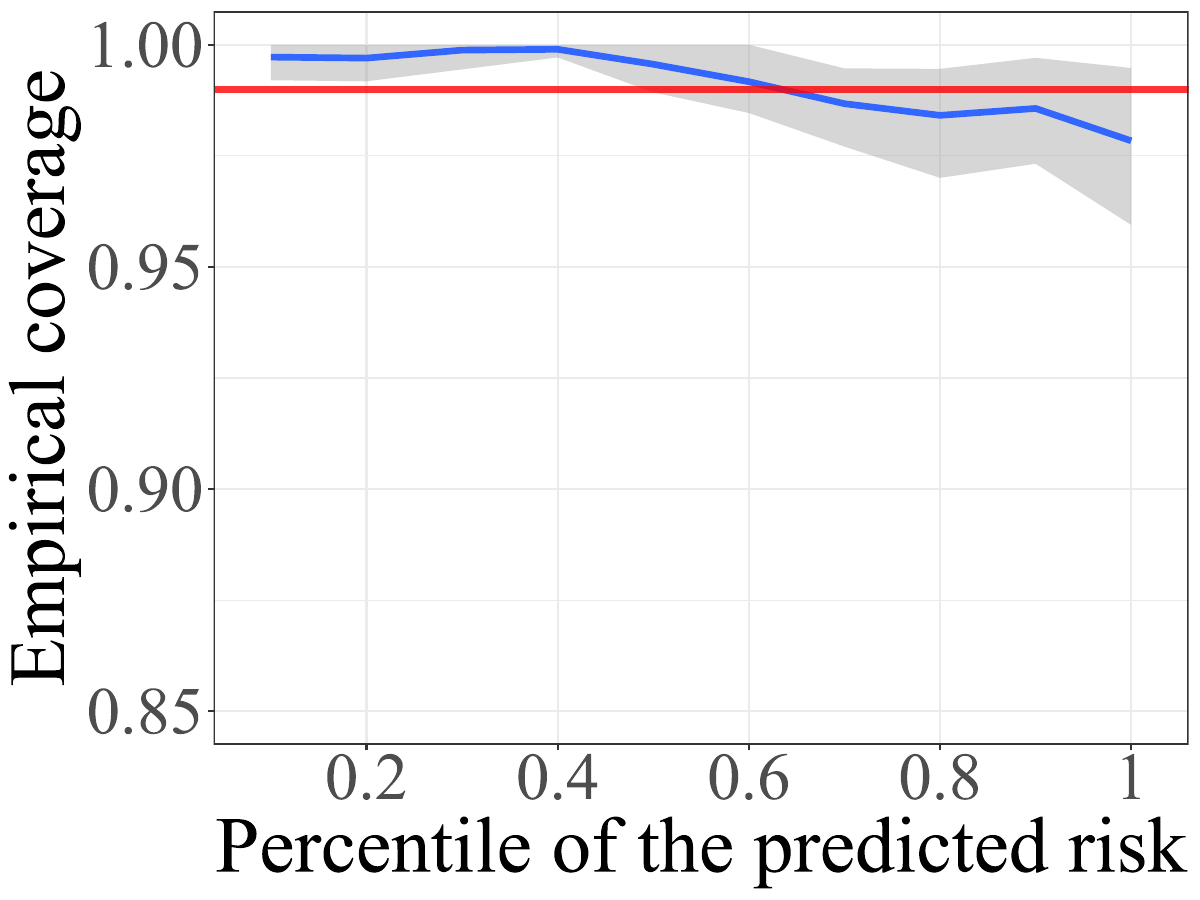}
  \end{minipage}
  \hfill
  \begin{minipage}{0.3\textwidth}
    \includegraphics[width = \textwidth]{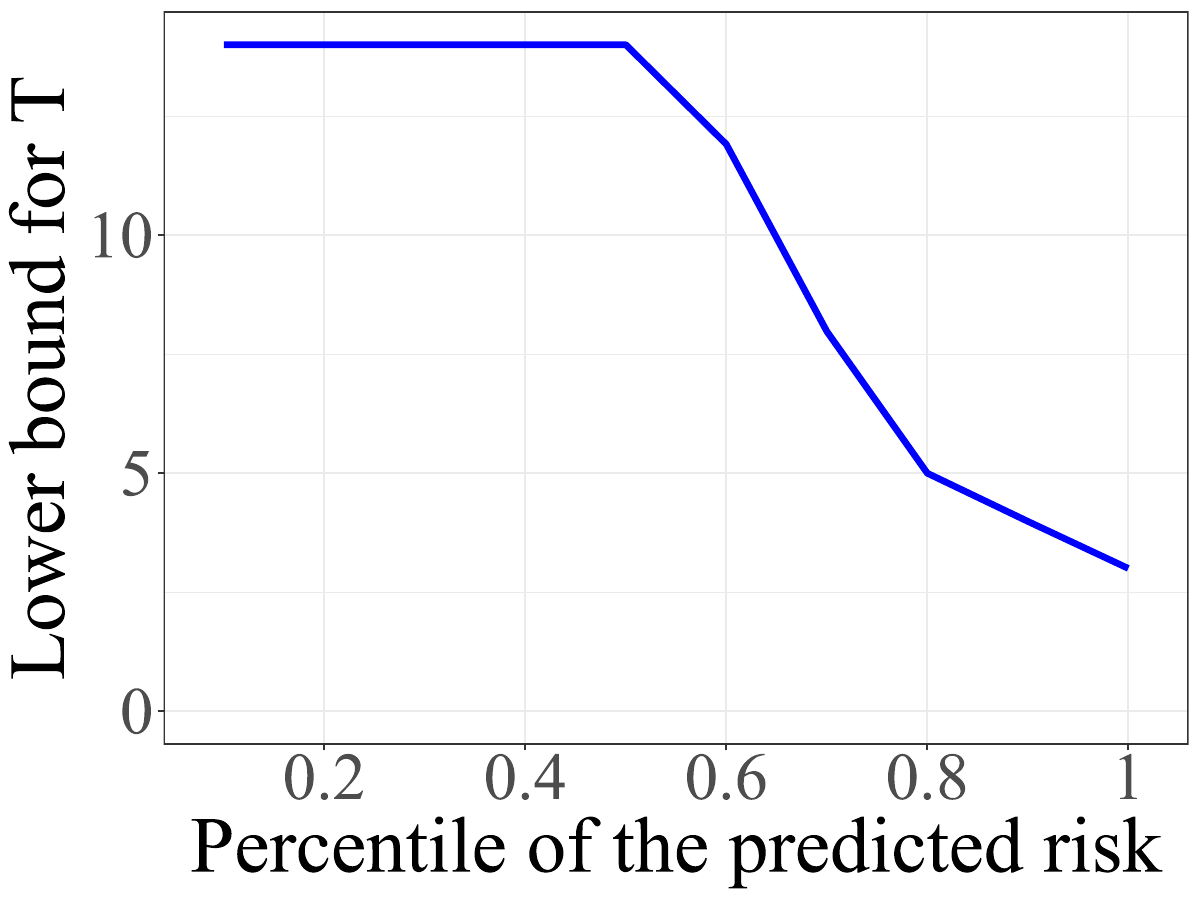}
  \end{minipage}
  \caption{Results for synthetic censoring times across 100
    replications: empirical coverage (left), empirical conditional
    coverage of the CDR-LPB (middle), and CDR-LPB as a function of
    the percentile of the predicted risk (right).  The target coverage level
    is $99\%$. The blue curves correspond to the mean coverage in the
    middle panel and the median LPB in the right panel; the gray
    confidence bands correspond to the $5\%$ and $95\%$ quantiles of
    the estimates across $100$ independent replications.}
  \label{fig:synthetic_c}
\end{figure}

\begin{figure}[ht]
  \centering
  \begin{minipage}{0.3\textwidth}
    \includegraphics[width = \textwidth]{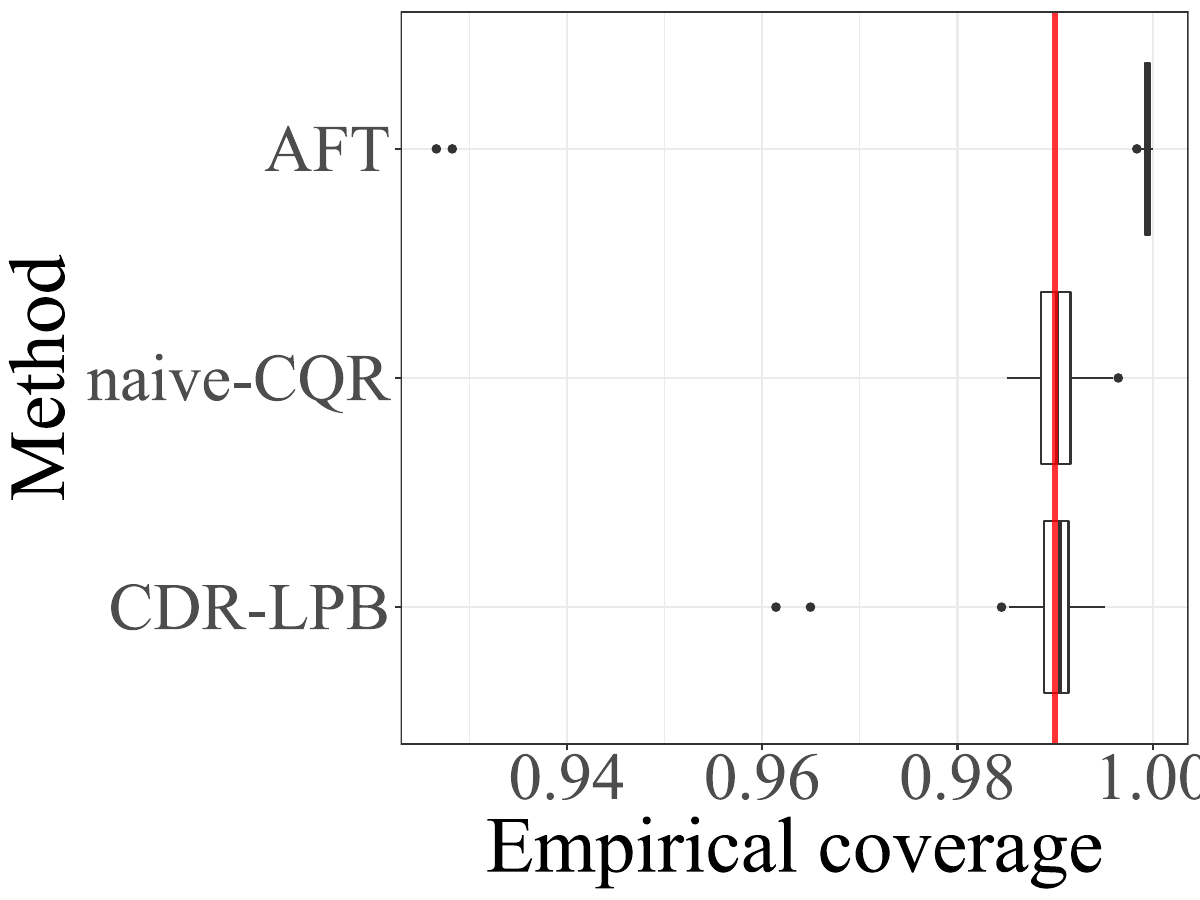}
  \end{minipage}
  \hfill
  \begin{minipage}{0.3\textwidth}
    \includegraphics[width =\textwidth]{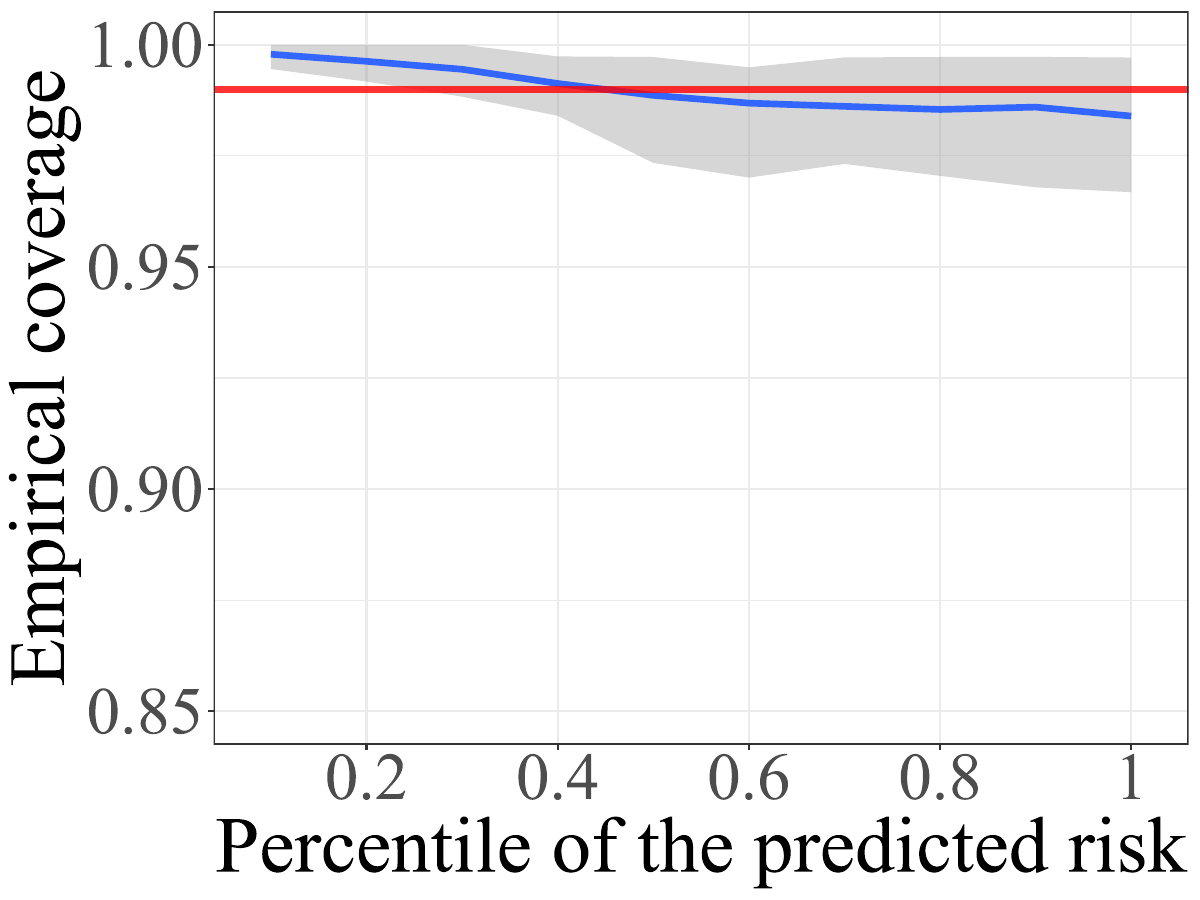}
  \end{minipage}
  \hfill
  \begin{minipage}{0.3\textwidth}
    \includegraphics[width = \textwidth]{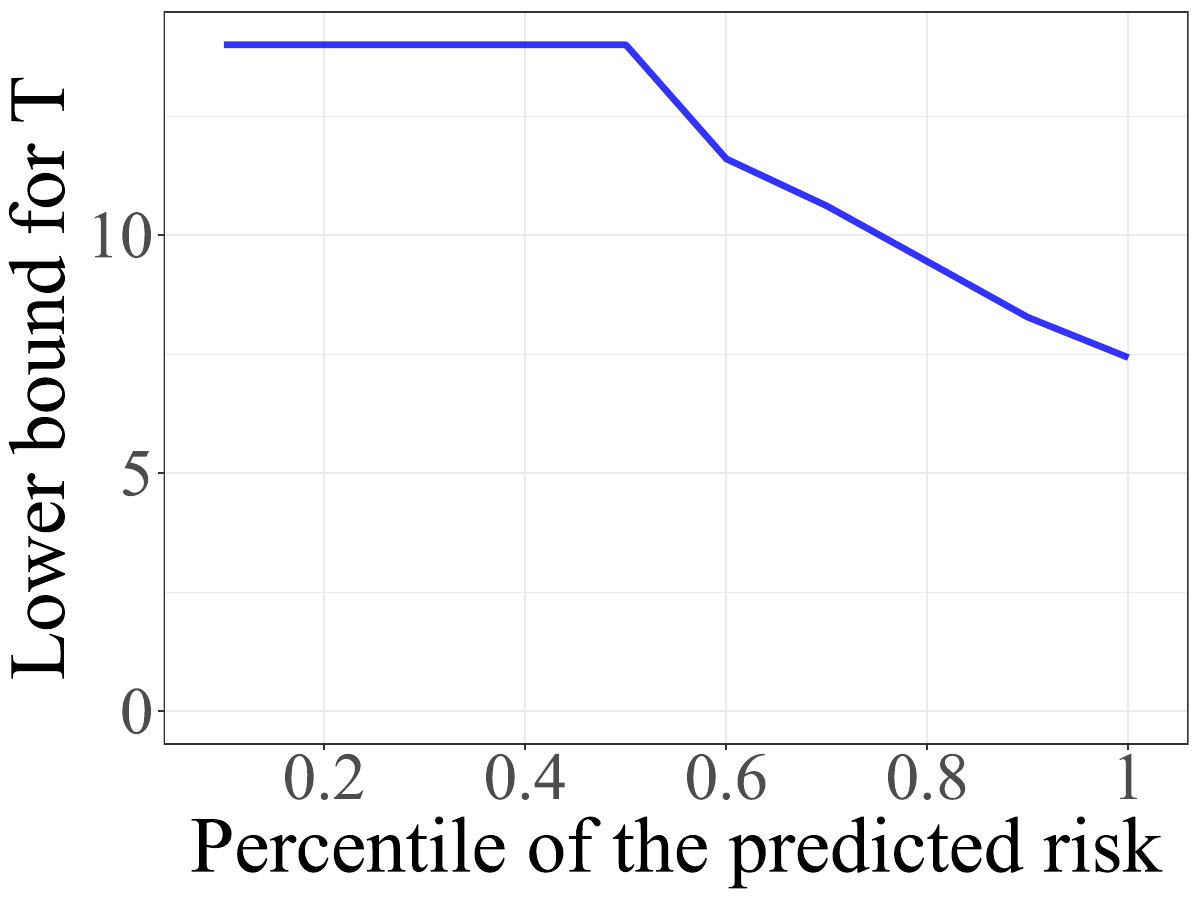}
  \end{minipage}
  \caption{Results for synthetic survival times: everything else is as
    in Figure~\ref{fig:synthetic_c}.}
  \label{fig:synthetic_t}
\end{figure}

\subsection{Real data analysis}

We now turn attention to actual COVID-19 responses. Again, we randomly
split the data into a training set including $75\%$ of data and a
holdout set including the remaining $25\%$.  Then we run the CDR on
the training set and validate the LPBs on the holdout set. The issue
is that the actual survival time is only partially observed, and thus,
the coverage of a given LPB cannot be assessed accurately (this is
precisely why we generated semi-synthetic responses in the previous
section.)  Nevertheless, we note that
\begin{align*}
  \beta_{\text{lo}} := \p\big( \tT \ge \hat{L}(X) \big)  \le  \p\big(T \ge \hat{L}(X)\big) \le 1 - \p \big( \tT < \hat{L}(X), T \le C\big) =: \beta_{\text{hi}},
\end{align*}
where both $\beta_{\text{lo}}$ and $\beta_{\text{hi}}$ are estimable
from the data. This says that we can assess the marginal coverage of
the LPBs by evaluating a lower and upper bound on the
coverage. Of course, this extends to conditional
coverage.

To assess the stability, we evaluate our method on $100$ independent sample splits. Figure~\ref{fig:bounds_of_coverage} presents the empirical lower and
upper bound of the marginal coverage and those of the conditional
coverage as functions of the predicted risk (as in the semi-synthetic
examples), together with their variability across $100$ sample splits. The left panel shows that the upper bound is very close to
the lower bound, and both concentrate around the target level. Thus we
can be assured that the CDR-LPB is well calibrated. Similarly, the
other panels show that the CDR-LPB is approximately conditionally
calibrated.  We conclude this section by showing in
Figure~\ref{fig:covid_conditional_bound} the LPBs as functions of the
percentiles of the predicted risk, age, and BMI, respectively.

\begin{figure}[ht]
  \centering
  \begin{minipage}{0.3\textwidth}
  \centering
  \includegraphics[width = \textwidth]{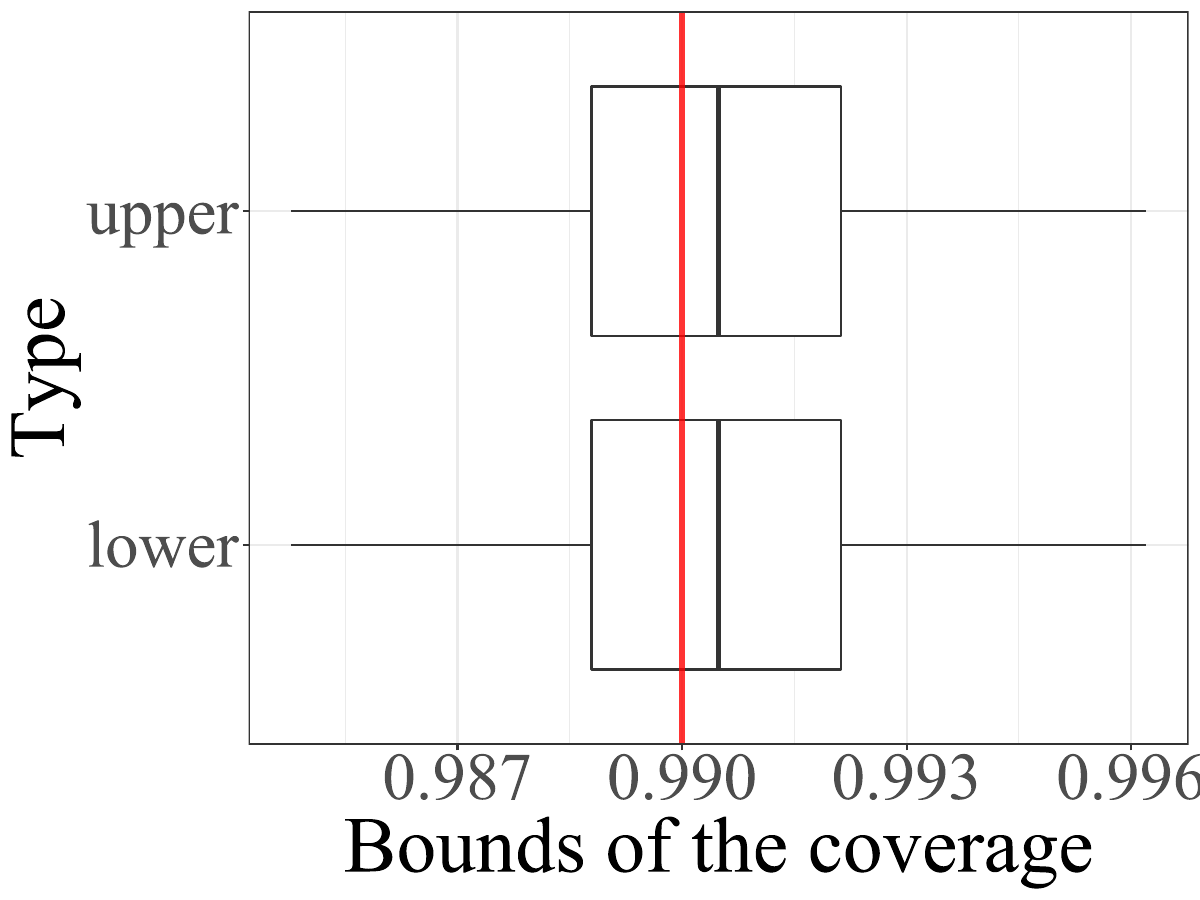}
  (a)
  \end{minipage}
  \hfill
  \begin{minipage}{0.3\textwidth}
  \centering
  \includegraphics[width = \textwidth]{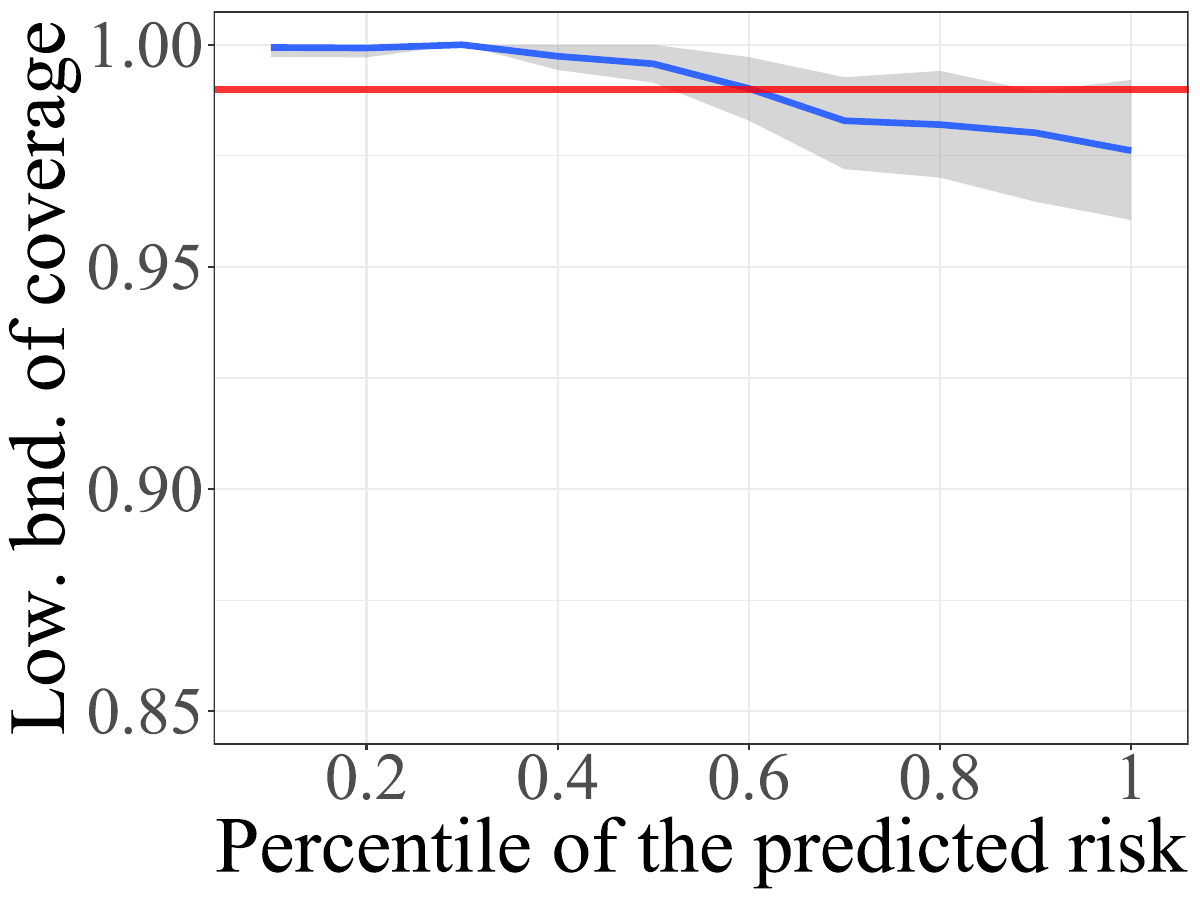}
  (b)
  \end{minipage}
  \hfill
  \begin{minipage}{0.3\textwidth}
  \centering
  \includegraphics[width = \textwidth]{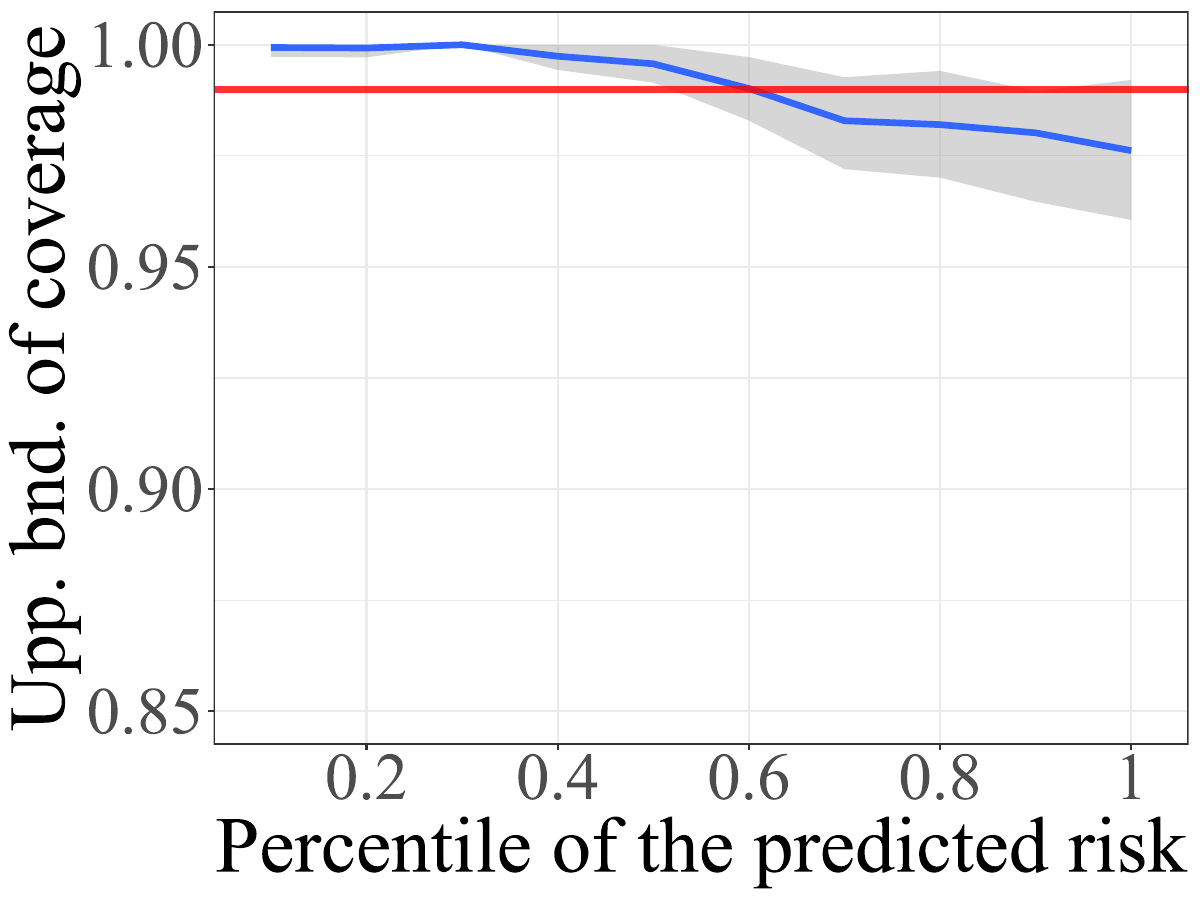}
  (c)
  \end{minipage}
  \caption{Analysis of the UK Biobank COVID-19 data: (a) lower and upper bounds of the empirical coverage;
    (b) lower and (c) upper bounds of empirical coverage as a function of the predicted risk. The target coverage level is $99\%$. The blue 
           curves correspond to the mean coverage, and the gray confidence bands correspond to the $5\%$ and $95\%$ quantiles of 
         the estimates across $100$ sample splits. % \ejc{I don't think so.}. 
         % \llzr{Thanks for pointing this out! It should be across 100 sample splits. We added a sentence at the beginning of the paragraph above. }
       } 
  \label{fig:bounds_of_coverage}
\end{figure}

\begin{figure}[ht]
  \centering
  \begin{minipage}{0.3\textwidth}
    \includegraphics[width = \textwidth]{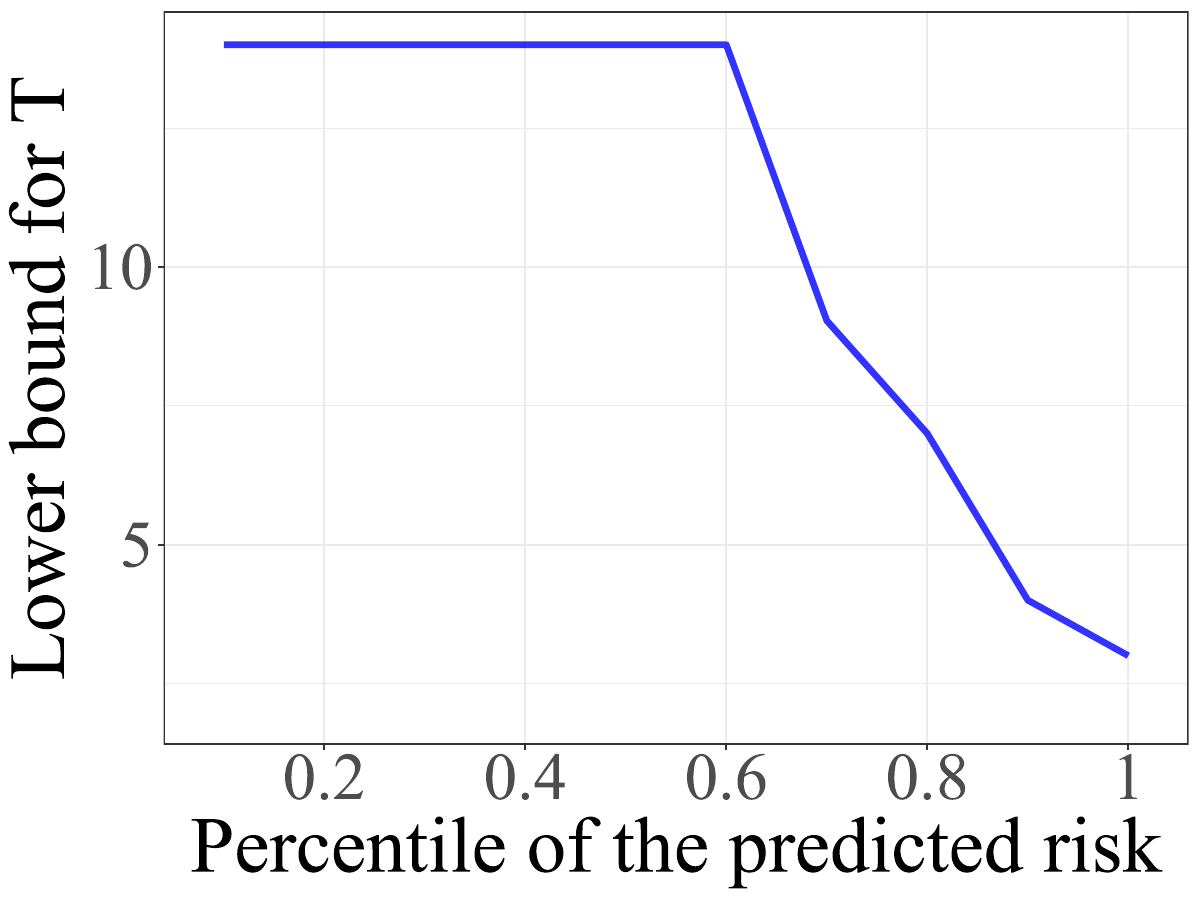}
  \end{minipage}
  \begin{minipage}{0.3\textwidth}
    \includegraphics[width = \textwidth]{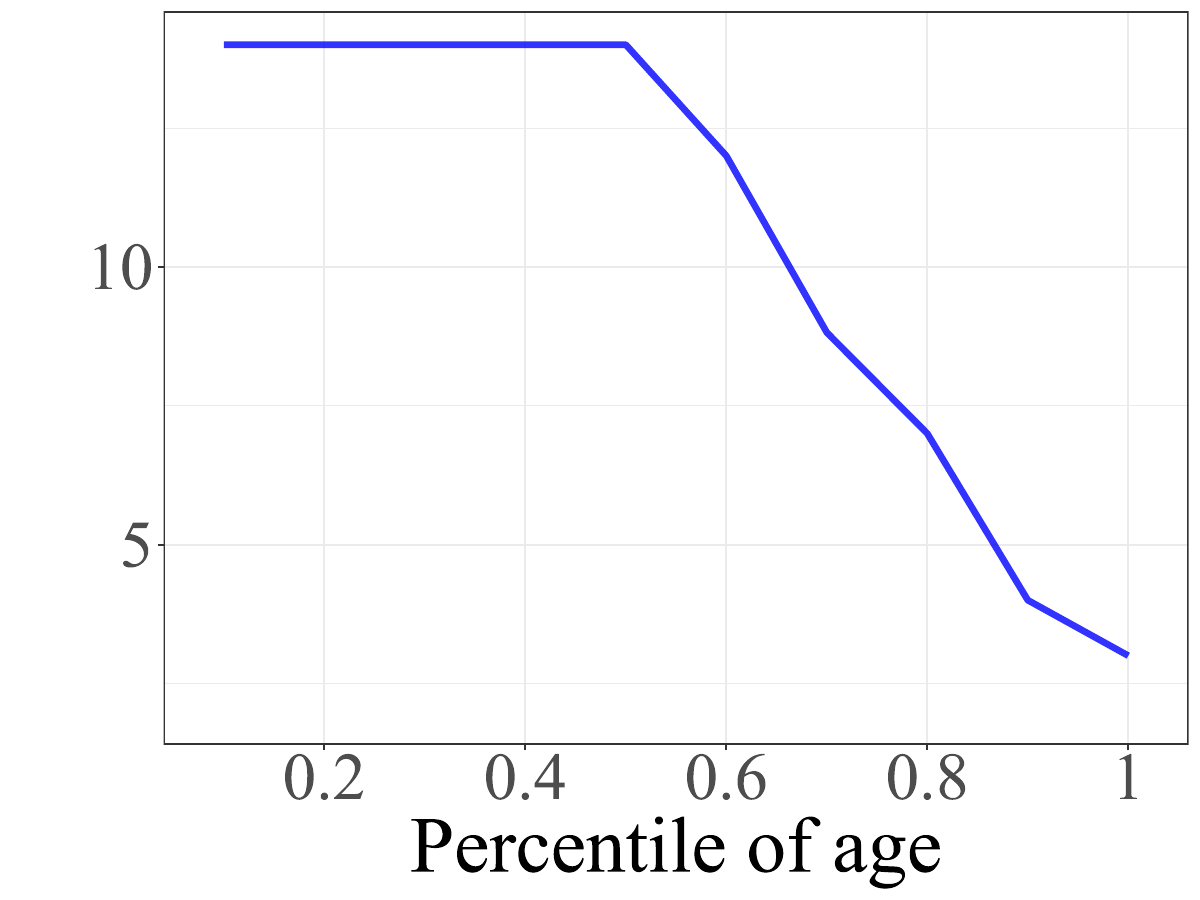}
  \end{minipage}
  \begin{minipage}{0.3\textwidth}
    \includegraphics[width = \textwidth]{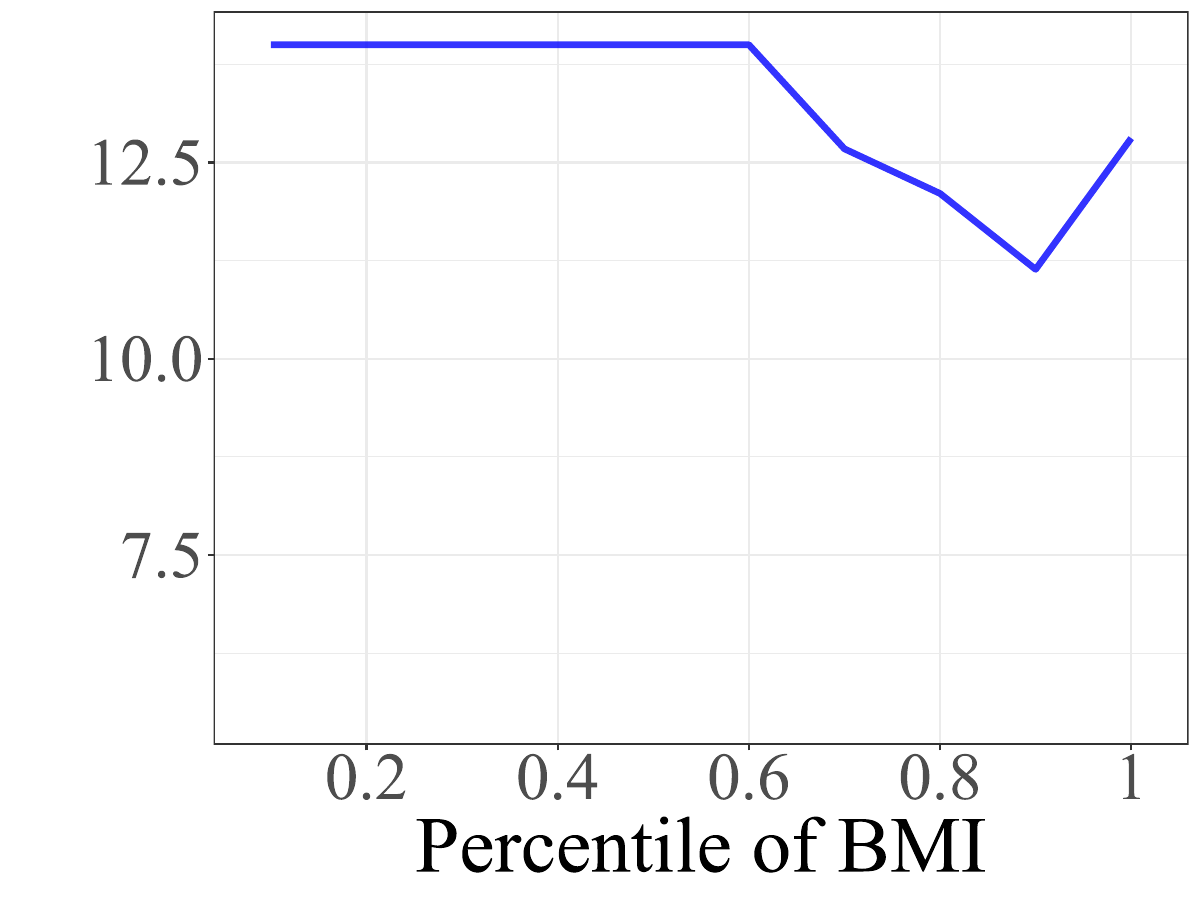}
  \end{minipage}
  \caption{Analysis of the UK Biobank COVID-19 data: LPBs on the survival time of COVID-19 patients as a function of the 
    percentiles of predicted risk (left), age (middle) and BMI (right). The target coverage level is $99\%$. The blue 
  curves correspond to the median LPB across $100$ sample splits. }
  \label{fig:covid_conditional_bound}
\end{figure}

%%% Local Variables:
%%% mode: latex
%%% TeX-master: "main"
%%% End:

\section{Discussion and extensions}
\label{sec:discussion}
\subsection{Beyond Type-I censoring}\label{subsec:beyond}
In practice, censoring can be driven by multiple factors. As discussed
in \cite{leung1997censoring}, the two most common types of right
censoring in a clinical study are the end-of-study censoring caused by
the trial termination and the loss-to-follow-up censoring caused by
unexpected attrition; \revise{see also \cite{korn1986censoring} and
  \cite{schemper1996note} for an account of the two types of
  censoring}. Let $\Cend$ denote the former and $\Closs$ the
latter. By definition, $\Cend$ is observable for every patient, as
long as the entry times are accurately recorded. \revise{When the
  event is not death (e.g., the patient's returning visit), $\Closs$
  is observable if all patients are tracked until
  the end of the study. } However, when the event is death, $\Closs$
can only be observed for surviving patients. This is because for dead
patients, it is impossible to know when they would have been lost to
follow-up, had they survived.

In survival analysis without loss-to-follow-up censoring, or
time-to-event analysis with non-death events, the setting of Type I
censoring considered in this paper is plausible. However, it is found
that both the end-of-study and loss-to-follow-up censoring are
involved in many applications \citep{leung1997censoring}. In these
cases, the effective censoring time $C$ is the minimum of $\Cend$ and
$\Closs$, and is only observable for surviving patients, namely the
patients with $T > C$. This situation prevents us from applying
Algorithm \ref{algo:weighted_split} below because the subpopulation
with $C\ge c_0$ is not fully observed. If we use the subpopulation
whose $C$ is 1) observed and 2) larger than or equal to a threshold
$c_0$ instead, then the joint distribution of $(X, T)$ becomes
$P_{X\mid C\ge c_0, T > C}\times P_{T\mid X, C\ge c_0, T > C}$. The
extra conditioning event $T > C$ induces a shift of the conditional
distribution, since
$P_{T\mid X, C\ge c_0, T > C}\neq P_{T\mid X, C\ge c_0}$ in general,
rendering the weighted split conformal inference invalid.

Our method can nevertheless be adapted to yield meaningful inference
under an additional assumption:
\begin{equation}
  \label{eq:Cend}
  (T, \Closs) ~ \indep ~ \Cend\mid X.
\end{equation}
\revise{
Unlike \cite{korn1986censoring} and \cite{schemper1996note},
\eqref{eq:Cend} does not impose any restrictions on the dependence
between $T$ and $\Closs$, which is harder to conceptualize.} The
assumption \eqref{eq:Cend} tends to be plausible, \revise{especially
  when the total length of follow-up is short,} since the randomness
of the end-of-study censoring time often comes from the entry time of
a patient, which is arguably exogenous to the survival time and
attrition, at least when conditioning on a few demographic
variables. \revise{There are certain cases where \eqref{eq:Cend} could
  be violated. For example, if new treatments become available during
  the course of a study, subjects who enter later are different from
  those who enter earlier as they could have been given the
  alternative treatments, but were not.}

Let $T' = T\wedge \Closs$, the survival time censored merely by the loss to follow-up. Then the censored survival time $\td{T} = T \wedge C = T'\wedge \Cend$, and \eqref{eq:Cend} implies that $T' \indep \Cend\mid X$, an analog of the \condIC~\eqref{eq:conditional_independent_censoring}. Since $\Cend$ is observed for every patient, Algorithm \ref{algo:weighted_split} can be applied to produce an LPB $\hat{L}(\cdot)$ such that
\[\p(T' \ge \hat{L}(X))\ge 1 - \alpha\Longrightarrow \p(T \ge \hat{L}(X))\ge 1 - \alpha.\]
% \revise{When even Assumption~\ref{eq:Cend} does not hold, we always 
% apply conformal inference to $(X,\tilde{T})$ and obtain 
% LPBs on $\tT$, which is still a valid LPB on $T$. Without
% furthur assumptions, this is also optimal by 
% Theorem~\ref{thm:distribution_free}.}

\revise{
  In Section~\ref{sec:add_sim} of the Appendix, 
  we provide an additional simulation illustrating
  the result of our method in this setting.
  An observation in conjunction with this line of reasoning is
  that, unlike most survival analysis techniques, our method
  distinguishes two sources of censoring and takes advantage of the
  censoring mechanism itself.}
% This is  in a similar vein to~\citet{korn1986censoring,schemper1996note} 
% who considers a similar 
% setting and but a different condition $T ~\indep ~ \Closs \wedge \Cend \mid X$.}
It can be regarded as a building block to remove the
adverse effect of $\Cend$. It remains an interesting question whether
the censoring issue induced by $\Closs$ can be resolved or alleviated in this context.

\subsection{Sharper coverage criteria}\label{subsec:localized}

It is more desirable to achieve a stronger conditional coverage
criterion:
\begin{align}
  \label{eq:conditional_criterion}
  \p\left(T\ge \hat{L}(X)\mid X = x\right) \ge 1-\alpha,
\end{align}
which states that $\hat{L}(X)$ is a conditionally calibrated LPB. Clearly, \eqref{eq:conditional_criterion} implies valid marginal coverage. Theorem \ref{thm:double_robustness_CQR} and \ref{thm:double_robustness_CDR} show that the CQR- and CDR-LPB are approximately conditionally calibrated if the conditional quantiles are estimated well. However, without distributional assumptions, we can show that \eqref{eq:conditional_criterion} can only be achieved by trivial LPBs. % We present the details in Appendix \ref{subapp:conditional_coverage}.
\begin{theorem}\label{thm:conditional_coverage}
  Assume that $X\in \R^{p}$ and $C\ge 0, T\ge 0$. Let $P_{(X, C)}$ be
  any given distribution of $(X, C)$. If $\hat{L}(\cdot)$ satisfies
  \eqref{eq:conditional_criterion} uniformly for all joint
  distributions of $(X, C, T)$ with $(X, C)\sim P_{(X, C)}$, % \ejc{Why
    % do we care about the distribution of $(X,C)$?}~\llzr{It makes the result stronger -- the impossibility result holds even if the procedure knows the distribution of $(X, C)$. See explanations below. }
  then for all such
  distributions,
  \[\p(\hat{L}(x) = 0) \ge 1 - \alpha,\]
  at almost surely all points $x$ aside from the atoms of $P_X$.
\end{theorem}
Theorem \ref{thm:conditional_coverage} implies that no nontrivial LPB exists even if the distribution of $(X, C)$ is known. Put another way, it is impossible to achieve desired conditional coverage while being agnostic to the conditional survival function. This impossibility result is inspired by previous works on uncensored outcomes and two-sided intervals \citep{vovk2012conditional, foygel2019limits}. 

It is valuable to find other achievable coverage criteria which are
sharper than the marginal coverage criterion
\eqref{eq:marginal_criterion}. Without censoring and covariate shift,
\cite{vovk2003mondrian} introduced Mondrian conformal inference to
achieve desired marginal coverage over multiple subpopulations. The
idea is further developed from different perspectives
\citep{vovk2012conditional, lei2013distribution, guan2019conformal,
  foygel2019limits, romano2019malice}. Given a partition of the
covariate space $\{\calX_1, \ldots, \calX_K\}$, Mondrian conformal
inference guarantees that
\[\p(Y\in \hat{C}(X)\mid X\in \calX_k) \ge 1 - \alpha, \quad k = 1, \ldots, K.\]
% footnote
Mondrian conformal inference allows
  the subgroups to also depend on the outcome; see
  \cite{vovk2005algorithmic}, which refers to the rule of forming
  subgroups as a ``taxonomy.'' Besides, the subgroups can also be
  overlapping; see \cite{foygel2019limits}.
Following their techniques, we can extend Mondrian conformal inference to our case by modifying the calibration term $\eta(x)$ (in Algorithm \ref{algo:weighted_split}):
\begin{equation}
  \label{eq:mondrian_eta}
  \eta(x) = \quantile\left(1 - \alpha; \frac{\sum_{i\in\calI_{\text{ca}}, X_i\in \calX_k} 
    \hat{p}_i(x)\delta_{V_i} + \hat{p}_{\infty}(x)\delta_{\infty}}{\sum_{i\in\calI_{\text{ca}, X_i\in \calX_k}}\hat{p}_i(x) + \hat{p}_{\infty}(x)}\right), \quad \forall x\in \calX_k.
\end{equation}
Suppose $\calX_1$ and $\calX_2$ correspond to male and female
subpopulations. Then $\eta(x)$ is a function of both the testing point
$x$ and the gender. That said, estimation of censoring mechanisms and
conditional survival functions can still depend on the whole training
fold $\calZ_\tr$ as joint training may be more powerful than
separate training on each subpopulation \citep{romano2019malice}.

When the censoring mechanism is known, we can prove that
\begin{equation}
  \label{eq:mondrian_criterion}
  \p(T\wedge c_0\ge \hat{L}(X)\mid X\in \calX_k)\ge 1 - \alpha, \quad k = 1, \ldots, K.
\end{equation}
By the \condIC, the target distribution in the localized criterion \eqref{eq:mondrian_criterion} for a given $k$ can be rewritten as
\begin{equation}
  \label{eq:local_subpopulation_dist}
  (X, T\wedge c_0) \mid C\ge c_0, X\in \calX_k~\sim~P_{X\mid C\ge c_0, X\in \calX_k}\times P_{T\wedge c_0 \mid X}.
\end{equation}
The covariate shift between the observed and target distributions is
\[w_{k}(x) = \frac{d P_{X\mid C\ge c_0, X\in \calX_k}}{dP_{X}}(x) \propto \frac{I(x\in \calX_k)}{\p(C\ge c_0\mid X = x)}.\]
This justifies the calibration term \eqref{eq:mondrian_eta} in the weighted Mondrian conformal inference. Since the weighted Mondrian conformal inference is a special case of Algorithm \ref{algo:weighted_split}, it also enjoys the double robustness property, implied by Theorem \ref{thm:double_robustness} in Section \ref{app:double_robustness} in the Appendix.

\subsection{Survival counterfactual prediction}\label{subsec:causal}

The proposed method in this paper is designed for a single cohort. In practice, patients are often exposed to multiple conditions, and the goal is to predict the counterfactual survival times had the cohort been exposed to a different condition. For example, a clinical study typically involves a treatment group and a control group. For a new patient,  it is of interest to predict her survival time had she been assigned the treatment. For uncensored outcomes, \cite{lei2020conformal} proposed a method based on weighted conformal inference for counterfactual prediction under the potential outcome framework \citep{neyman1923application, rubin1974estimating}. We can extend their strategy to handle censored outcomes and apply it to the survival counterfactual prediction.

Suppose each patient has a pair of potential survival times
$(T(1), T(0))$, where $T(1)$ (resp. $T(0)$) denotes the survival time
had the patient been assigned into the treatment (resp. control)
group. Our goal is to construct a calibrated LPB on $T(1)$, given
i.i.d.~observations $(X_i, W_i, C_i, T_i)_{i=1}^{n}$ with $W_i$
denoting the treatment assignment and
\[T_i = \left\{
      \begin{array}{ll}
        T_i(1), & W_i = 1,\\
        T_i(0), & W_i = 0.
      \end{array}
\right.\]
Without further assumptions on the correlation structures between $T(1)$ and $T(0)$, it is natural to conduct inference based on the observed treated group since the control group contains no information about $T(1)$. The joint distribution of $(X, T(1)\wedge c_0)$ on this group becomes
\[(X, T(1)\wedge c_0) \mid C\ge c_0, W = 1~\sim~P_{X\mid C\ge c_0, W = 1}\times P_{T(1)\wedge c_0 \mid X, C\ge c_0, W = 1}.\]
% Under the strong ignorability assumption that $(T(1), T(0))\indep W \mid X$ \citep{rubin1978bayesian}, a widely accepted starting point in causal inference, the conditional distribution of $T(1)\wedge c_0$ matches the target:
Under the assumption that $(T(1), T(0))\indep (W, C) \mid X$, the conditional distribution of $T(1)\wedge c_0$ matches the target:
\[P_{T(1)\wedge c_0 \mid X, C\ge c_0, W = 1} = P_{T(1)\wedge c_0\mid X}.\]
The assumption is a combination of the strong ignorability assumption \citep{rubin1978bayesian}, a widely accepted starting point in causal inference, and the \condIC.
% The strong ignorability assumption plays essentially the same role as the
% conditionally independent censoring assumption---both reduce the
% distributional shift to a pure covariate shift. 
The density ratio of
the two covariate distributions can be characterized by
\[w(x) = \frac{dP_{X\mid C\ge c_0, W = 1}}{dP_{X}}(x)\propto \frac{1}{\p(C\ge c_0, W = 1\mid X = x)}.\]
In many applications, it is plausible to further assume that $C\indep W\mid X$. In this case,
\[\p(C\ge c_0, W = 1\mid X = x) = \p(C\ge c_0\mid X = x)\p(W = 1\mid X = x),\]
where the first term is the censoring mechanism and the second term is the propensity score \citep{rosenbaum1983central}. Therefore, we can obtain calibrated LPBs on counterfactual survival times if both the censoring mechanism and the propensity score are known. This assumption is often plausible for randomized clinical trials. Furthermore, it has a doubly robust guarantee of coverage that is similar to Theorems \ref{thm:double_robustness_CQR} and \ref{thm:double_robustness_CDR}.

%%% Local Variables:
%%% mode: latex
%%% TeX-master: "main"
%%% End:

\subsection*{Acknowledgment}

E. C. was supported by Office of Naval Research grant N00014-20-12157,
by the National Science Foundation grants OAC 1934578 and DMS 2032014,
by the Army Research Office (ARO) under grant W911NF-17-1-0304, and by
the Simons Foundation under award 814641. L.~L.~and Z.~R.~were
supported by the same NSF OAC grant and by the Discovery Innovation
Fund for Biomedical Data Sciences.  L.L.~was also supported by NIH
grant R01MH113078. The authors are grateful to Steven Goodman, Lucas Janson, Ying Jin, Yan Min,
Chiara Sabatti, Matteo Sesia, Lu Tian and Steve Yadlowsky for their
constructive feedback. This research has been conducted using the UK Biobank Resource under Application Number 27837.
% \clearpage
%\appendix
%\input{appendixEC}

\bibliographystyle{rss}
\bibliography{cfsurv}

\begin{thebibliography}{91}
\expandafter\ifx\csname natexlab\endcsname\relax\def\natexlab#1{#1}\fi
\expandafter\ifx\csname url\endcsname\relax
  \def\url#1{\texttt{#1}}\fi
\expandafter\ifx\csname urlprefix\endcsname\relax\def\urlprefix{URL: }\fi

\bibitem[{Aitchison and Dunsmore(1980)}]{aitchison1980statistical}
Aitchison, J. and Dunsmore, I.~R. (1980) \textit{Statistical prediction
  analysis}.
\newblock CUP Archive.

\bibitem[{Allison(1984)}]{allison1984event}
Allison, P.~D. (1984) \textit{Event history analysis: Regression for
  longitudinal event data}.
\newblock No.~46. Sage.

\bibitem[{Anisimov and Fedorov(2007)}]{anisimov2007modelling}
Anisimov, V.~V. and Fedorov, V.~V. (2007) Modelling, prediction and adaptive
  adjustment of recruitment in multicentre trials.
\newblock \textit{Statistics in medicine}, \textbf{26}, 4958--4975.

\bibitem[{Athey et~al.(2019)Athey, Tibshirani and Wager}]{athey2019generalized}
Athey, S., Tibshirani, J. and Wager, S. (2019) Generalized random forests.
\newblock \textit{The Annals of Statistics}, \textbf{47}, 1148--1178.

\bibitem[{von Bahr and Esseen(1965)}]{vonbahr65}
von Bahr, B. and Esseen, C.-G. (1965) Inequalities for the $r$-th absolute
  moment of a sum of random variables, $1\le r \le 2$.
\newblock \textit{The Annals of Mathematical Statistics}, \textbf{36},
  299--303.

\bibitem[{Bain(2017)}]{bain2017statistical}
Bain, L. (2017) \textit{Statistical analysis of reliability and life-testing
  models: theory and methods}.
\newblock Routledge.

\bibitem[{Barber et~al.(2019{\natexlab{a}})Barber, Cand{\`e}s, Ramdas and
  Tibshirani}]{foygel2019limits}
Barber, R.~F., Cand{\`e}s, E.~J., Ramdas, A. and Tibshirani, R.~J.
  (2019{\natexlab{a}}) The limits of distribution-free conditional predictive
  inference.
\newblock \textit{Information and Inference: A Journal of the IMA}.

\bibitem[{Barber et~al.(2019{\natexlab{b}})Barber, Candes, Ramdas and
  Tibshirani}]{barber2019predictive}
Barber, R.~F., Candes, E.~J., Ramdas, A. and Tibshirani, R.~J.
  (2019{\natexlab{b}}) Predictive inference with the {Jackknife}+.
\newblock \textit{arXiv preprint arXiv:1905.02928}.

\bibitem[{Barnard et~al.(2010)Barnard, Dent and Cook}]{barnard2010systematic}
Barnard, K.~D., Dent, L. and Cook, A. (2010) A systematic review of models to
  predict recruitment to multicentre clinical trials.
\newblock \textit{BMC medical research methodology}, \textbf{10}, 1--8.

\bibitem[{Breslow(1975)}]{breslow1975analysis}
Breslow, N.~E. (1975) Analysis of survival data under the proportional hazards
  model.
\newblock \textit{International Statistical Review/Revue Internationale de
  Statistique}, 45--57.

\bibitem[{Bycroft et~al.(2018)Bycroft, Freeman, Petkova, Band, Elliott, Sharp,
  Motyer, Vukcevic, Delaneau, O’Connell et~al.}]{bycroft2018uk}
Bycroft, C., Freeman, C., Petkova, D., Band, G., Elliott, L.~T., Sharp, K.,
  Motyer, A., Vukcevic, D., Delaneau, O., O’Connell, J. et~al. (2018) The
  {UK} {Biobank} resource with deep phenotyping and genomic data.
\newblock \textit{Nature}, \textbf{562}, 203--209.

\bibitem[{Carter(2004)}]{carter2004application}
Carter, R.~E. (2004) Application of stochastic processes to participant
  recruitment in clinical trials.
\newblock \textit{Controlled clinical trials}, \textbf{25}, 429--436.

\bibitem[{Carter et~al.(2005)Carter, Sonne and Brady}]{carter2005practical}
Carter, R.~E., Sonne, S.~C. and Brady, K.~T. (2005) Practical considerations
  for estimating clinical trial accrual periods: application to a multi-center
  effectiveness study.
\newblock \textit{BMC medical research methodology}, \textbf{5}, 1--5.

\bibitem[{Cauchois et~al.(2020)Cauchois, Gupta and Duchi}]{cauchois2020knowing}
Cauchois, M., Gupta, S. and Duchi, J. (2020) Knowing what you know: valid
  confidence sets in multiclass and multilabel prediction.
\newblock \textit{arXiv preprint arXiv:2004.10181}.

\bibitem[{Chernozhukov et~al.(2019)Chernozhukov, W{\"u}thrich and
  Zhu}]{chernozhukov2019distributional}
Chernozhukov, V., W{\"u}thrich, K. and Zhu, Y. (2019) Distributional conformal
  prediction.
\newblock \textit{arXiv preprint arXiv:1909.07889}.

\bibitem[{Cox(1972)}]{cox1972regression}
Cox, D.~R. (1972) Regression models and life-tables.
\newblock \textit{Journal of the Royal Statistical Society: Series B
  (Methodological)}, \textbf{34}, 187--202.

\bibitem[{D’Amour et~al.(2021)D’Amour, Ding, Feller, Lei and
  Sekhon}]{d2021overlap}
D’Amour, A., Ding, P., Feller, A., Lei, L. and Sekhon, J. (2021) Overlap in
  observational studies with high-dimensional covariates.
\newblock \textit{Journal of Econometrics}, \textbf{221}, 644--654.

\bibitem[{Efron(1979)}]{efron1979bootstrap}
Efron, B. (1979) Bootstrap methods: Another look at the {Jackknife}.
\newblock \textit{The Annals of Statistics}, 1--26.

\bibitem[{Efron(2020)}]{efron2020prediction}
--- (2020) Prediction, estimation, and attribution.
\newblock \textit{International Statistical Review}, \textbf{88}, S28--S59.

\bibitem[{Efron and Tibshirani(1994)}]{efron1994introduction}
Efron, B. and Tibshirani, R.~J. (1994) \textit{An introduction to the
  bootstrap}.
\newblock CRC press.

\bibitem[{Emanuel et~al.(2020)Emanuel, Persad, Upshur, Thome, Parker, Glickman,
  Zhang, Boyle, Smith and Phillips}]{emanuel2020fair}
Emanuel, E., Persad, G., Upshur, R., Thome, B., Parker, M., Glickman, A.,
  Zhang, C., Boyle, C., Smith, M. and Phillips, J. (2020) Fair allocation of
  scarce medical resources in the time of {Covid}-19.
\newblock \textit{The New England Journal of Medicine}, \textbf{382}.

\bibitem[{Faraggi and Simon(1995)}]{faraggi1995neural}
Faraggi, D. and Simon, R. (1995) A neural network model for survival data.
\newblock \textit{Statistics in medicine}, \textbf{14}, 73--82.

\bibitem[{Friedman and Narasimhan(2020)}]{Jerome2020conTree}
Friedman, J. and Narasimhan, B. (2020) \textit{conTree: Contrast Trees and
  Boosting}.
\newblock \urlprefix\url{https://jhfhub.github.io/conTree_tutorial}.
\newblock R package version 0.2-8.

\bibitem[{Friedman(2020)}]{friedman2020contrast}
Friedman, J.~H. (2020) Contrast trees and distribution boosting.
\newblock \textit{Proceedings of the National Academy of Sciences},
  \textbf{117}, 21175--21184.

\bibitem[{Gajewski et~al.(2008)Gajewski, Simon and
  Carlson}]{gajewski2008predicting}
Gajewski, B.~J., Simon, S.~D. and Carlson, S.~E. (2008) Predicting accrual in
  clinical trials with bayesian posterior predictive distributions.
\newblock \textit{Statistics in medicine}, \textbf{27}, 2328--2340.

\bibitem[{Geisser(1993)}]{geisser1993predictive}
Geisser, S. (1993) \textit{Predictive inference}, vol.~55.
\newblock CRC press.

\bibitem[{Goeman(2010)}]{goeman2010l1}
Goeman, J.~J. (2010) L1 penalized estimation in the {Cox} proportional hazards
  model.
\newblock \textit{Biometrical journal}, \textbf{52}, 70--84.

\bibitem[{Guan(2019)}]{guan2019conformal}
Guan, L. (2019) Conformal prediction with localization.
\newblock \textit{arXiv preprint arXiv:1908.08558}.

\bibitem[{Gui and Li(2005)}]{gui2005penalized}
Gui, J. and Li, H. (2005) Penalized {Cox} regression analysis in the
  high-dimensional and low-sample size settings, with applications to
  microarray gene expression data.
\newblock \textit{Bioinformatics}, \textbf{21}, 3001--3008.

\bibitem[{Gupta et~al.(2019)Gupta, Kuchibhotla and Ramdas}]{gupta2019nested}
Gupta, C., Kuchibhotla, A.~K. and Ramdas, A.~K. (2019) Nested conformal
  prediction and quantile out-of-bag ensemble methods.
\newblock \textit{arXiv preprint arXiv:1910.10562}.

\bibitem[{Harrell~Jr(2015)}]{harrell2015regression}
Harrell~Jr, F.~E. (2015) \textit{Regression modeling strategies: with
  applications to linear models, logistic and ordinal regression, and survival
  analysis}.
\newblock Springer.

\bibitem[{Hong and Tamer(2003)}]{hong2003inference}
Hong, H. and Tamer, E. (2003) Inference in censored models with endogenous
  regressors.
\newblock \textit{Econometrica}, \textbf{71}, 905--932.

\bibitem[{Hothorn et~al.(2006)Hothorn, B{\"u}hlmann, Dudoit, Molinaro and Van
  Der~Laan}]{hothorn2006survival}
Hothorn, T., B{\"u}hlmann, P., Dudoit, S., Molinaro, A. and Van Der~Laan, M.~J.
  (2006) Survival ensembles.
\newblock \textit{Biostatistics}, \textbf{7}, 355--373.

\bibitem[{Ishwaran et~al.(2008)Ishwaran, Kogalur, Blackstone, Lauer
  et~al.}]{ishwaran2008random}
Ishwaran, H., Kogalur, U.~B., Blackstone, E.~H., Lauer, M.~S. et~al. (2008)
  Random survival forests.
\newblock \textit{Annals of Applied Statistics}, \textbf{2}, 841--860.

\bibitem[{Kalbfleisch and Prentice(2011)}]{kalbfleisch2011statistical}
Kalbfleisch, J.~D. and Prentice, R.~L. (2011) \textit{The statistical analysis
  of failure time data}, vol. 360.
\newblock John Wiley \& Sons.

\bibitem[{Kaplan and Meier(1958)}]{kaplan1958nonparametric}
Kaplan, E.~L. and Meier, P. (1958) Nonparametric estimation from incomplete
  observations.
\newblock \textit{Journal of the American statistical association},
  \textbf{53}, 457--481.

\bibitem[{Katzman et~al.(2016)Katzman, Shaham, Cloninger, Bates, Jiang and
  Kluger}]{katzman2016deep}
Katzman, J.~L., Shaham, U., Cloninger, A., Bates, J., Jiang, T. and Kluger, Y.
  (2016) Deep survival: A deep {Cox} proportional hazards network.
\newblock \textit{stat}, \textbf{1050}.

\bibitem[{Koenker(2020)}]{roger2020quantreg}
Koenker, R. (2020) \textit{quantreg: Quantile Regression}.
\newblock \urlprefix\url{https://CRAN.R-project.org/package=quantreg}.
\newblock R package version 5.75.

\bibitem[{Korn(1986)}]{korn1986censoring}
Korn, E.~L. (1986) Censoring distributions as a measure of follow-up in
  survival analysis.
\newblock \textit{Statistics in medicine}, \textbf{5}, 255--260.

\bibitem[{Krishnamoorthy and Mathew(2009)}]{krishnamoorthy2009statistical}
Krishnamoorthy, K. and Mathew, T. (2009) \textit{Statistical tolerance regions:
  theory, applications, and computation}, vol. 744.
\newblock John Wiley \& Sons.

\bibitem[{Lagakos(1979)}]{lagakos1979general}
Lagakos, S.~W. (1979) General right censoring and its impact on the analysis of
  survival data.
\newblock \textit{Biometrics}, 139--156.

\bibitem[{Lao et~al.(2017)Lao, Chen, Li, Li, Zhang, Liu and Zhai}]{lao2017deep}
Lao, J., Chen, Y., Li, Z.-C., Li, Q., Zhang, J., Liu, J. and Zhai, G. (2017) A
  deep learning-based radiomics model for prediction of survival in
  glioblastoma multiforme.
\newblock \textit{Scientific reports}, \textbf{7}, 1--8.

\bibitem[{Lei et~al.(2018)Lei, G’Sell, Rinaldo, Tibshirani and
  Wasserman}]{lei2018distribution}
Lei, J., G’Sell, M., Rinaldo, A., Tibshirani, R.~J. and Wasserman, L. (2018)
  Distribution-free predictive inference for regression.
\newblock \textit{Journal of the American Statistical Association},
  \textbf{113}, 1094--1111.

\bibitem[{Lei et~al.(2013)Lei, Robins and Wasserman}]{lei2013distribution}
Lei, J., Robins, J. and Wasserman, L. (2013) Distribution-free prediction sets.
\newblock \textit{Journal of the American Statistical Association},
  \textbf{108}, 278--287.

\bibitem[{Lei and Wasserman(2014)}]{lei2014distribution}
Lei, J. and Wasserman, L. (2014) Distribution-free prediction bands for
  non-parametric regression.
\newblock \textit{Journal of the Royal Statistical Society: Series B:
  Statistical Methodology}, 71--96.

\bibitem[{Lei and Cand{\`e}s(2021)}]{lei2020conformal}
Lei, L. and Cand{\`e}s, E.~J. (2021) Conformal inference of counterfactuals and
  individual treatment effects.
\newblock \textit{Journal of the Royal Statistical Society: Series B
  (Statistical Methodology)}.

\bibitem[{Leung et~al.(1997)Leung, Elashoff and Afifi}]{leung1997censoring}
Leung, K.-M., Elashoff, R.~M. and Afifi, A.~A. (1997) Censoring issues in
  survival analysis.
\newblock \textit{Annual review of public health}, \textbf{18}, 83--104.

\bibitem[{Li and Bradic(2020)}]{li2020censored}
Li, A.~H. and Bradic, J. (2020) Censored quantile regression forest.
\newblock In \textit{International Conference on Artificial Intelligence and
  Statistics}, 2109--2119. PMLR.

\bibitem[{Murphy et~al.(1997)Murphy, Rossini and van~der
  Vaart}]{murphy1997maximum}
Murphy, S., Rossini, A. and van~der Vaart, A.~W. (1997) Maximum likelihood
  estimation in the proportional odds model.
\newblock \textit{Journal of the American Statistical Association},
  \textbf{92}, 968--976.

\bibitem[{Neyman(1923/1990)}]{neyman1923application}
Neyman, J. (1923/1990) On the application of probability theory to agricultural
  experiments. {Essay} on principles. {Section} 9.
\newblock \textit{Statistical Science}, \textbf{5}, 465--472.
\newblock Translated and edited by D. M. Dabrowska and T. P. Speed from the
  Polish original, which appeared in Roczniki Nauk Rolniczyc, Tom X (1923):
  1--51 (Annals of Agricultural Sciences).

\bibitem[{Peng and Huang(2008)}]{peng2008survival}
Peng, L. and Huang, Y. (2008) Survival analysis with quantile regression
  models.
\newblock \textit{Journal of the American Statistical Association},
  \textbf{103}, 637--649.

\bibitem[{Portnoy(2003)}]{portnoy2003censored}
Portnoy, S. (2003) Censored regression quantiles.
\newblock \textit{Journal of the American Statistical Association},
  \textbf{98}, 1001--1012.

\bibitem[{Powell(1986)}]{powell1986censored}
Powell, J.~L. (1986) Censored regression quantiles.
\newblock \textit{Journal of econometrics}, \textbf{32}, 143--155.

\bibitem[{Ranney et~al.(2020)Ranney, Griffeth and Jha}]{ranney2020critical}
Ranney, M.~L., Griffeth, V. and Jha, A.~K. (2020) Critical supply
  shortages—the need for ventilators and personal protective equipment during
  the {Covid}-19 pandemic.
\newblock \textit{New England Journal of Medicine}, \textbf{382}, e41.

\bibitem[{Ratkovic and Tingley(2021)}]{ratkovic2021estimation}
Ratkovic, M. and Tingley, D. (2021) Estimation and inference on nonlinear and
  heterogeneous effects.
\newblock \textit{Tech. rep.}
\newblock
  \urlprefix\url{https://scholar.harvard.edu/files/dtingley/files/mdei.pdf}.

\bibitem[{Romano et~al.(2019{\natexlab{a}})Romano, Barber, Sabatti and
  Cand{\`e}s}]{romano2019malice}
Romano, Y., Barber, R.~F., Sabatti, C. and Cand{\`e}s, E.~J.
  (2019{\natexlab{a}}) With malice towards none: Assessing uncertainty via
  equalized coverage.
\newblock \textit{arXiv preprint arXiv:1908.05428}.

\bibitem[{Romano et~al.(2019{\natexlab{b}})Romano, Patterson and
  Candes}]{romano2019conformalized}
Romano, Y., Patterson, E. and Candes, E. (2019{\natexlab{b}}) Conformalized
  quantile regression.
\newblock In \textit{Advances in Neural Information Processing Systems},
  3543--3553.

\bibitem[{Romano et~al.(2020)Romano, Sesia and
  Cand{\`e}s}]{romano2020classification}
Romano, Y., Sesia, M. and Cand{\`e}s, E.~J. (2020) Classification with valid
  and adaptive coverage.
\newblock \textit{arXiv preprint arXiv:2006.02544}.

\bibitem[{Rosenbaum and Rubin(1983)}]{rosenbaum1983central}
Rosenbaum, P.~R. and Rubin, D.~B. (1983) The central role of the propensity
  score in observational studies for causal effects.
\newblock \textit{Biometrika}, \textbf{70}, 41--55.

\bibitem[{Rosenthal(1970)}]{rosenthal1970subspaces}
Rosenthal, H.~P. (1970) On the subspaces of $l^{p} (p> 2)$ spanned by sequences
  of independent random variables.
\newblock \textit{Israel Journal of Mathematics}, \textbf{8}, 273--303.

\bibitem[{Rubin(1974)}]{rubin1974estimating}
Rubin, D.~B. (1974) Estimating causal effects of treatments in randomized and
  nonrandomized studies.
\newblock \textit{Journal of educational Psychology}, \textbf{66}, 688.

\bibitem[{Rubin(1978)}]{rubin1978bayesian}
--- (1978) Bayesian inference for causal effects: The role of randomization.
\newblock \textit{The Annals of statistics}, 34--58.

\bibitem[{Sadinle et~al.(2019)Sadinle, Lei and Wasserman}]{sadinle2019least}
Sadinle, M., Lei, J. and Wasserman, L. (2019) Least ambiguous set-valued
  classifiers with bounded error levels.
\newblock \textit{Journal of the American Statistical Association},
  \textbf{114}, 223--234.

\bibitem[{Sant'Anna(2016)}]{sant2016program}
Sant'Anna, P.~H. (2016) Program evaluation with right-censored data.
\newblock \textit{arXiv preprint arXiv:1604.02642}.

\bibitem[{Saunders et~al.(1999)Saunders, Gammerman and
  Vovk}]{saunders1999transduction}
Saunders, C., Gammerman, A. and Vovk, V. (1999) Transduction with confidence
  and credibility.
\newblock In \textit{Proceedings of the Sixteenth International Joint
  Conference on Artificial Intelligence}, 722--726.

\bibitem[{Scharfstein and Robins(2002)}]{scharfstein2002estimation}
Scharfstein, D.~O. and Robins, J.~M. (2002) Estimation of the failure time
  distribution in the presence of informative censoring.
\newblock \textit{Biometrika}, \textbf{89}, 617--634.

\bibitem[{Schemper and Smith(1996)}]{schemper1996note}
Schemper, M. and Smith, T.~L. (1996) A note on quantifying follow-up in studies
  of failure time.
\newblock \textit{Controlled clinical trials}, \textbf{17}, 343--346.

\bibitem[{Sesia and Cand{\`e}s(2020)}]{sesia2020comparison}
Sesia, M. and Cand{\`e}s, E.~J. (2020) A comparison of some conformal quantile
  regression methods.
\newblock \textit{Stat}, \textbf{9}, e261.

\bibitem[{Shafer and Vovk(2008)}]{shafer2008tutorial}
Shafer, G. and Vovk, V. (2008) A tutorial on conformal prediction.
\newblock \textit{Journal of Machine Learning Research}, \textbf{9}, 371--421.

\bibitem[{Shah et~al.(2020)Shah, Peters et~al.}]{shah2020hardness}
Shah, R.~D., Peters, J. et~al. (2020) The hardness of conditional independence
  testing and the generalised covariance measure.
\newblock \textit{Annals of Statistics}, \textbf{48}, 1514--1538.

\bibitem[{Simon et~al.(2011)Simon, Friedman, Hastie and
  Tibshirani}]{simon2011regularization}
Simon, N., Friedman, J., Hastie, T. and Tibshirani, R. (2011) Regularization
  paths for {Cox}’s proportional hazards model via coordinate descent.
\newblock \textit{Journal of statistical software}, \textbf{39}, 1.

\bibitem[{Stine(1985)}]{stine1985bootstrap}
Stine, R.~A. (1985) Bootstrap prediction intervals for regression.
\newblock \textit{Journal of the American Statistical Association},
  \textbf{80}, 1026--1031.

\bibitem[{Therneau(2020)}]{survival-package}
Therneau, T.~M. (2020) \textit{A Package for Survival Analysis in R}.
\newblock \urlprefix\url{https://CRAN.R-project.org/package=survival}.
\newblock R package version 3.2-7.

\bibitem[{Tibshirani(1997)}]{tibshirani1997lasso}
Tibshirani, R. (1997) The lasso method for variable selection in the {Cox}
  model.
\newblock \textit{Statistics in medicine}, \textbf{16}, 385--395.

\bibitem[{Tibshirani et~al.(2019)Tibshirani, Foygel~Barber, Candes and
  Ramdas}]{tibshirani2019conformal}
Tibshirani, R.~J., Foygel~Barber, R., Candes, E. and Ramdas, A. (2019)
  Conformal prediction under covariate shift.
\newblock \textit{Advances in Neural Information Processing Systems},
  \textbf{32}, 2530--2540.

\bibitem[{Tsybakov(2008)}]{tsybakov2008introduction}
Tsybakov, A.~B. (2008) \textit{Introduction to nonparametric estimation}.
\newblock Springer Science \& Business Media.

\bibitem[{Vergano et~al.(2020)Vergano, Bertolini, Giannini, Gristina, Livigni,
  Mistraletti, Riccioni and Petrini}]{vergano2020clinical}
Vergano, M., Bertolini, G., Giannini, A., Gristina, G., Livigni, S.,
  Mistraletti, G., Riccioni, L. and Petrini, F. (2020) Clinical ethics
  recommendations for the allocation of intensive care treatments in
  exceptional, resource-limited circumstances: the italian perspective during
  the {COVID}-19 epidemic.
\newblock \textit{Critical Care}, \textbf{24}.

\bibitem[{Vershynin(2018)}]{vershynin2018high}
Vershynin, R. (2018) \textit{High-dimensional probability: An introduction with
  applications in data science}, vol.~47.
\newblock Cambridge university press.

\bibitem[{Verweij and Van~Houwelingen(1993)}]{verweij1993cross}
Verweij, P.~J. and Van~Houwelingen, H.~C. (1993) Cross-validation in survival
  analysis.
\newblock \textit{Statistics in medicine}, \textbf{12}, 2305--2314.

\bibitem[{Vovk(2002)}]{vovk2002line}
Vovk, V. (2002) On-line confidence machines are well-calibrated.
\newblock In \textit{The 43rd Annual IEEE Symposium on Foundations of Computer
  Science, 2002. Proceedings.}, 187--196. IEEE.

\bibitem[{Vovk(2012)}]{vovk2012conditional}
--- (2012) Conditional validity of inductive conformal predictors.
\newblock In \textit{Asian conference on machine learning}, 475--490.

\bibitem[{Vovk et~al.(2005)Vovk, Gammerman and Shafer}]{vovk2005algorithmic}
Vovk, V., Gammerman, A. and Shafer, G. (2005) \textit{Algorithmic learning in a
  random world}.
\newblock Springer Science \& Business Media.

\bibitem[{Vovk et~al.(2003)Vovk, Lindsay, Nouretdinov and
  Gammerman}]{vovk2003mondrian}
Vovk, V., Lindsay, D., Nouretdinov, I. and Gammerman, A. (2003) Mondrian
  confidence machine.
\newblock \textit{Technical Report}.

\bibitem[{Wald(1943)}]{wald1943extension}
Wald, A. (1943) An extension of wilks' method for setting tolerance limits.
\newblock \textit{The Annals of Mathematical Statistics}, \textbf{14}, 45--55.

\bibitem[{Wang et~al.(2019)Wang, Li and Reddy}]{wang2019machine}
Wang, P., Li, Y. and Reddy, C.~K. (2019) Machine learning for survival
  analysis: A survey.
\newblock \textit{ACM Computing Surveys (CSUR)}, \textbf{51}, 1--36.

\bibitem[{Wei(1992)}]{wei1992accelerated}
Wei, L.-J. (1992) The accelerated failure time model: a useful alternative to
  the {Cox} regression model in survival analysis.
\newblock \textit{Statistics in medicine}, \textbf{11}, 1871--1879.

\bibitem[{Wilks(1941)}]{wilks1941determination}
Wilks, S.~S. (1941) Determination of sample sizes for setting tolerance limits.
\newblock \textit{The Annals of Mathematical Statistics}, \textbf{12}, 91--96.

\bibitem[{Witten and Tibshirani(2010)}]{witten2010survival}
Witten, D.~M. and Tibshirani, R. (2010) Survival analysis with high-dimensional
  covariates.
\newblock \textit{Statistical methods in medical research}, \textbf{19},
  29--51.

\bibitem[{Wu and Carroll(1988)}]{wu1988estimation}
Wu, M.~C. and Carroll, R.~J. (1988) Estimation and comparison of changes in the
  presence of informative right censoring by modeling the censoring process.
\newblock \textit{Biometrics}, 175--188.

\bibitem[{Yang and Kuchibhotla(2021)}]{yang2021finite}
Yang, Y. and Kuchibhotla, A.~K. (2021) Finite-sample efficient conformal
  prediction.
\newblock \textit{arXiv preprint arXiv:2104.13871}.

\bibitem[{Zhang and Lu(2007)}]{zhang2007adaptive}
Zhang, H.~H. and Lu, W. (2007) Adaptive lasso for {Cox}'s proportional hazards
  model.
\newblock \textit{Biometrika}, \textbf{94}, 691--703.

\end{thebibliography}

\appendix
%!TEX ROOT=main.tex
\numberwithin{equation}{section}
\numberwithin{theorem}{section}
\numberwithin{lemma}{section}
\section{Proofs of impossibility results}

\subsection{Proof of Theorem \ref{thm:distribution_free}}
Let $(X_i, C_i, T_i)_{i=1}^{n+1}\stackrel{i.i.d.}{\sim} (X, C, T)$. For notational convenience, we put $Z_i = (X_i, C_i, T_i)$. To avoid confusion, we expand $\hat{L}(x)$ into $\hat{L}(Z_1, \ldots, Z_n; x)$. Note that $\hat{L}$ depends on $(Z_1, \ldots, Z_n)$ through $(X_i, C_i, \tT_i)_{i=1}^{n}$. Let 
\[\phi(Z_1, \ldots, Z_{n+1}) = I(T_{n+1} < \hat{L}_{n}(Z_1, \ldots, Z_n; X_{n+1})).\]
Since $\hat{L}_n$ satisfies \eqref{eq:marginal_criterion} under the conditionally independent censoring assumption \eqref{eq:conditional_independent_censoring}, we have that
\[\p\left(\phi(Z_1, \ldots, Z_{n+1}) = 1\right) \le \alpha.\]
As a result, if we treat $T\indep C\mid X$ as a null hypothesis,
$\phi(Z_1, \ldots, Z_{n+1})$ is an $\alpha$-level test. Note that $X$
is continuous, and $(T, C)$ are continuous or discrete. By Theorem 2
and Remark 4 of \cite{shah2020hardness}, for any joint distribution
$Q$ of $Z$ with the same continuity conditions on $(X, C, T)$,
\begin{align}\label{eq:phiQ}
  \p_{Z_i\stackrel{i.i.d.}{\sim}Q}\left(\phi(Z_1, \ldots, Z_{n+1}) = 1\right) \le \alpha.
\end{align}
Let $\td{Z}_i = (X_i, C_i, \tT_i)$ and $Q$ denote its distribution. Then $\tT_i\wedge C_i = T_i\wedge C_i$ and thus
\[\hat{L}_{n}(\td{Z}_1, \ldots, \td{Z}_n; x) = \hat{L}_{n}(Z_1, \ldots, Z_n; x).\]
Clearly, $X$ is absolutely continuous with respect to the Lebesgue measure and $\tT, C$ are absolutely continuous with respect to the Lebesgue measure or the counting measure. 
By \eqref{eq:phiQ} and the definition of $\phi$, we have
\[\p\left(\td{T}_{n+1} \ge \hat{L}_{n}(Z_1, \ldots, Z_n; X_{n+1})\right)\ge 1 - \alpha.\]
The proof is then completed by replacing $(T_{n+1}, X_{n+1})$ with $(T, X)$. 

\subsection{Proof of Theorem \ref{thm:conditional_coverage}}\label{subapp:conditional_coverage}

We prove the theorem by modifying the proof of Proposition 4 from \cite{vovk2012conditional}. To avoid confusion, we expand $\hat{L}(x)$ into $\hat{L}(Z_1, \ldots, Z_n; x)$ where $Z_{i} = (X_i, C_i, \tT_i)$. Fix any distribution $P$ with $(X, C)\sim P_{(X, C)}$ and $C\ge 0, T\ge 0$ almost surely. Suppose there exists a set $\event$ of $P_{X}$-non-atom $x$ such that $P_{X}(\event) > 0$, and for any $x\in \event$, 
  \[\p_{Z_i\stackrel{i.i.d.}{\sim}P}(\hat{L}(Z_1, \ldots, Z_n, x) > 0) > \alpha.\]
  Since $\event$ only includes non-atom $x$'s, there exists $t_{0} > 0$ and $\delta > 0$ such that
  \begin{align}
    \label{eq:iidP}
    \p_{Z_i\stackrel{i.i.d.}{\sim}P}(\hat{L}(x) > t_{0}) > \alpha + \delta, \quad \forall x\in \event.
  \end{align}
  We can further shrink $\event$ so that
  \begin{align}
    \label{eq:Pevent}
    \sqrt{2 - 2(1 - P_{X}(\event))^{n}} \le \delta / 2.
  \end{align}
  Fix any $t_{1}\in (0, t_{0})$. Define a new probability distribution $Q$ on $(X, C, T)$ with $(X, C) \sim P_{(X, C)}$ and the regular conditional probability
  \[Q(T\in A\mid X = x, C = c) = \left\{
      \begin{array}{ll}
        P(T\in A\mid X = x, C = c), & x\not\in \event,\\
        \delta_{t_{1}}(A), & x\in \event,
      \end{array}
    \right.\]
  where $\delta_{t_{1}}$ defines the point mass on $t_{1}$. Let $d_{\mathrm{TV}}$ denote the total-variation distance. Then,
  \begin{align}
    d_{\mathrm{TV}}(P, Q)
    &= \sup_{\calA_{X}, \calA_{C}, \calA_{T}}|P(X\in \calA_{X}, C\in \calA_{C}, T\in \calA_{T}) - Q(X\in \calA_{X}, C\in \calA_{C}, T\in \calA_{T})|\\
    & = \sup_{\calA_{X}, \calA_{C}, \calA_{T}}|P(X\in \calA_{X}\cap \event, C\in \calA_{C}, T\in \calA_{T}) - Q(X\in \calA_{X}\cap \event, C\in \calA_{C}, T\in \calA_{T})|\\
    & \le \sup_{\calA_{X}, \calA_{C}, \calA_{T}}\max\{P(X\in \calA_{X}\cap \event, C\in \calA_{C}, T\in \calA_{T}), Q(X\in \calA_{X}\cap \event, C\in \calA_{C}, T\in \calA_{T})\}\\
    & \le \sup_{\calA_{X}}\max\{P(X\in \calA_{X}\cap \event), Q(X\in \calA_{X}\cap \event)\}\\
    & = \sup_{\calA_{X}}P_{X}(\calA_{X}\cap \event)\le P_{X}(\event).
  \end{align}
  Using the tensorization inequality for the total-variation distance (see e.g., \cite{tsybakov2008introduction}, Section 2.4) and \eqref{eq:Pevent}, we obtain that
  \[d_{\mathrm{TV}}(P^{n}, Q^{n})\le \sqrt{2 - 2(1 - d_{\mathrm{TV}}(P, Q))^{n}}\le \delta / 2.\]
  Together with \eqref{eq:iidP}, this implies that
  \[\p_{Z_i\stackrel{i.i.d.}{\sim}Q}(\hat{L}(x) > t_{0}) > \alpha + \delta / 2, \quad \forall x\in \event.\]
  Let $Z = (X, C, T)$ be an independent draw from $Q$. The above inequality can be reformulated as
  \[\p_{Z_i, Z\stackrel{i.i.d.}{\sim}Q}(\hat{L}(X) > t_{0}\mid X = x) > \alpha + \delta / 2, \quad \forall x\in \event.\]
  Marginalizing over $x\in \event$, it implies that
  \begin{equation}
    \label{eq:contradiction}
    \p_{Z_i, Z\stackrel{i.i.d.}{\sim}Q}(\hat{L}(X) > t_{0}, X\in \event) > (\alpha + \delta / 2)Q_{X}(\event).
  \end{equation}
  By definition of $Q$, $T = t_1 < t_0$ almost surely conditional on $X \in \event$. Thus,
  \[\p_{Z_i, Z\stackrel{i.i.d.}{\sim}Q}(T < \hat{L}(X), X\in \event) > (\alpha + \delta / 2)Q_{X}(\event).\]
  On the other hand, since $Q$ is a distribution with the same marginal distribution of $(X, C)$ and $T\ge 0$ almost surely, for any $x$, 
  \[\p_{Z_i, Z\stackrel{i.i.d.}{\sim}Q}(T < \hat{L}(X)\mid X = x)\le \alpha.\]
  Marginalizing over $x\in \event$, it implies that
  \[\p_{Z_i, Z\stackrel{i.i.d.}{\sim}Q}(T < \hat{L}(X), X\in \event) \le \alpha Q_{X}(\event).\]
  This contradicts \eqref{eq:contradiction} since $Q_{X}(\event) = P_{X}(\event) > 0$. The theorem is proved by contradiction. 
  
  \section{Double robustness of weighted conformal inference}\label{app:double_robustness}

  \revise{
  Throughout this section, we will focus on the generic weighted conformal inference sketched in Algorithm \ref{algo:weighted_split_general} below.

  \begin{algorithm}[H]\label{algo:weighted_split_general}    
  \DontPrintSemicolon  
  \SetAlgoLined
  \BlankLine
  \caption{generic weighted split conformal inference}
  \textbf{Input:} level $\alpha$; data $\calZ=(X_i,Y_i)_{i\in\calI}$; testing point $x$;\\
  \hspace{0.08\textwidth}function $V(x,y;\calD)$ to compute the conformity score between $(x,y)$ and 
  data $\calD$; \\
  \hspace{0.08\textwidth}function $\hat{w}(x;\calD)$ to fit the weight function at $x$ using 
  $\calD$ as data.\;
  \vspace*{.3cm}
  \textbf{Procedure:}\\
  \vspace*{.1cm}
  \hspace{0.02\textwidth}1. Split $\calZ$ into a training fold $\calZ_{\text{tr}} \triangleq (X_i,Y_i)_{i\in\calI_{\text{tr}}}$
  and a calibration fold $\calZ_{\text{ca}} \triangleq (X_i,Y_i)_{i\in\calI_{\text{ca}}}$.\; 
  \hspace{0.02\textwidth}2. For each $i\in \calI_{\text{ca}}$, compute the conformity score $V_i = V(X_i,Y_i;\calZ_{\text{tr}})$.\;
  \hspace{0.02\textwidth}3. For each $i\in\calI_{\text{ca}}$, compute the weight $W_i = \hat{w}(X_i;\calZ_{\text{tr}})\in [0, \infty)$.\;
  \hspace{0.02\textwidth}4. Compute the weights $\hat{p}_i(x) = \frac{W_i}{\sum_{i\in\calI_{\text{ca}}}W_i +\hat{w}(x;\calZ_{\text{tr}})}$
  and $\hat{p}_{\infty}(x) = \frac{\hat{w}(x;\calZ_{\text{tr}})}{\sum_{i\in\calI_{\text{ca}}}W_i + \hat{w}(x;\calZ_{\text{tr}})}$.\;
  \hspace{0.02\textwidth}5. Compute $\eta(x) = \quantile\left(1 - \alpha; \sum_{i\in\calI_{\text{ca}}} 
  \hat{p}_i(x)\delta_{V_i} + \hat{p}_{\infty}(x)\delta_{\infty}\right)$.\;
  \vspace*{.3cm}
  \textbf{Output}: $\hat{L}(x) = \inf\{y: V(x,y;\calZ_{\text{tr}}) \le \eta(x) \}$
\end{algorithm}
}
\subsection{Conformity score via nested sets}\label{subsec:score}

\cite{gupta2019nested} introduced a broad class of conformity scores
characterized by nested sets. Suppose we have a totally ordered index
set $\calS$ (e.g., $\R$) and a sequence of nested sets
$\{\calF_{s}(x; \calD): s\in \calS\}$ in the sense that
$\calF_{s_1}(x; \calD)\subset \calF_{s_2}(x; \calD)$ for any
$s_1\le s_2\in\calS$. Define a score as the index of the minimal set that
includes $y$, i.e.
\begin{align}\label{eq:nested_sets}
  V(x,y;\calD) = \inf\{s\in\calS: y\in \calF_{s}(x;\calD)\}.
\end{align}
Without loss of generality, we assume throughout that $\calF_{\inf \calS}(x)$ is the empty set and $\calF_{\sup \calS}(x)$ is the full domain of $Y$. The CMR-, CQR-, and CDR-based scores are instances of this: 
\begin{itemize}
\item[-] CMR score: $\calF_{s}(x; \calD) = [\hat{m}(x) - s, \infty)$.
\item[-] CQR score: $\calF_{s}(x; \calD) = [\hat{q}_{\alpha}(x) - s, \infty)$.
\item[-] CDR score: $\calF_{s}(x; \calD) = [\hat{q}_{\alpha - s}(x), \infty)$.
\end{itemize}
We refer the readers to Table 1 of \cite{gupta2019nested} for a list of other conformity scores. 

\revise{
\subsection{Nonasymptotic theory for weighted conformal inference}
In this section, we establish nonasymptotic bounds for the coverage which would imply the asymptotic results (e.g., Theorem \ref{thm:double_robustness_CQR} and Theorem \ref{thm:double_robustness_CDR}), formally proved in Section \ref{app:double_robustness_asym}. The first result is identical to Theorem A.1 of \cite{lei2020conformal}. 

\begin{theorem}\label{thm:double_robustness_w}
  Let
  $(X_{i}, Y_{i})\stackrel{i.i.d.}{\sim}(X, Y) \sim P_{X}\times
  P_{Y\mid X}$ and $Q_{X}$ be another distribution on the domain of
  $X$. Set $N = |\Z_{\tr}|$ and $n = |\Z_{\ca}|$. Further, let $\hat{w}(x) = \hat{w}(x; \Z_{\tr})$ be an
  estimate of $w(x) = (dQ_{X} / dP_{X})(x)$, and $\hat{C}(x)$
  be the conformal interval resulting from Algorithm
  \ref{algo:weighted_split_general} with an arbitrary conformity score (not necessarily the ones defined in Section \ref{subsec:score}). Assume that
  $\E[\hat{w}(X) \mid \Z_{\tr}] < \infty$, where $\E$ denotes
  expectation over $X\sim P_{X}$. Redefine
  $\hat{w}(x)$ as
  $\hat{w}(x) / \E[\hat{w}(X) \mid \Z_{\tr}]$ so that
  $\E[\hat{w}(X) \mid \Z_{\tr}] = 1$. Then
  \begin{equation}\label{eq:unconditional_coverage_rate_w}
    \p_{(X, Y)\sim Q_{X}\times P_{Y\mid X}}\lb Y\in \hat{C}(X)\rb\ge 1 - \alpha - \frac{1}{2}\E_{X\sim P_{X}}\big[|\hat{w}(X) - w(X)|\big].
  \end{equation}
  \end{theorem}

  The second result generalizes Theorem A.2 of \cite{lei2020conformal}. 
  \begin{theorem}\label{thm:double_robustness_q}
  In the setting of Theorem \ref{thm:double_robustness_w}, assume further that
  \begin{enumerate}[label = C\arabic*]
    \item $\p_{X\sim Q_{X}}(w(X) < \infty) = 1$, and there exist $\delta, M > 0$ such that $\lb\E\left[\hat{w}(X)^{1 + \delta}\right]\rb^{1 / (1 + \delta)} \le M$;
      \item There exists $r > 0, s_0\in \calS$, $b_2 > b_1 >0$, and a sequence of oracle nested sets $\{\calO_s(x)\}_{s\in\calS}$, such that
      \begin{enumerate}[label = (\roman*)]
      \item for any $\eps \in [0, r]$,
        \begin{align}
          % 1 - \alpha - b_2\eps \le \E\left\{\Indc\left( Y \in \calO_{s_0-\eps}(X) \mid X\right)\right\}
          1 - \alpha - b_2\eps \le \p\left( Y \in \calO_{s_0-\eps}(X) \mid X\right)
          \le 1 - \alpha - b_1\eps, \quad \text{almost surely};
        \end{align}
        % \ejc{Notation above is not clear. Do you mean $\ldots \le \E\{1(\ldots) | X\} \le \ldots$ a.s.?}~\llzr{We have changed the notation.}
      \item there exist $k, \ell > 0$ such that $\displaystyle \lim_{N\rightarrow \infty} \E[\hat{w}(X)\Delta^{k}(X)] = \lim_{N\rightarrow \infty} \E[w(X)\Delta^{\ell}(X)] = 0$, where $\Delta(x) =\sup_{s\in [s_0 - r, s_0]}~\Delta_s(x)$, and
        \begin{align}
          \Delta_s(x) = \inf\left\{\Delta\ge 0: 
          s-\Delta\in\calS , \calF_{s - \Delta}(x; \calZ_\tr) \subset \calO_{s}(x) \text{ and } \calO_{s - \Delta}(x) \subset \calF_{s}(x; \calZ_\tr)\right\}.
        \end{align}
      \end{enumerate}
    \end{enumerate}
    Then there are constants $\const_{1}$ and $\const_2$ that only depend on $r, b_1, b_2, \delta, M, k, \ell$ such that
    \begin{multline}
      \p_{(X, Y)\sim Q_{X}\times P_{Y\mid X}}(Y \in \hat{C}(X))\ge 1 - \alpha \\
       -\const_{1} \left\{\frac{(\log n)^{(1 + \delta') / 2(2 + \delta')}}{n^{\delta' / (2 + \delta')}} + \lb\E[\hat{w}(X)\Delta^{k}(X)]\rb^{1 / (2+k)} + \lb\E[w(X)\Delta^{\ell}(X)]\rb^{1/ (1+\ell)}\right\},\label{eq:unconditional_coverage_rate_q}
    \end{multline}
    and, with probability at least $1 - \beta$, 
    \begin{multline}
      \p_{(X, Y)\sim Q_{X}\times P_{Y\mid X}}(Y \in \hat{C}(X) \mid X )\ge 1 - \alpha \\
       - \const_{2}\left\{\frac{(\log n)^{(1 + \delta') / 2(2 + \delta')}}{n^{\delta' / (2 + \delta')}} + \lb\E[\hat{w}(X)\Delta^{k}(X)]\rb^{1 / (2+k)} + \lb\frac{\E[w(X)\Delta^{\ell}(X)]}{\beta}\rb^{1/ \ell}\right\}.\label{eq:conditional_coverage_rate_q}
     \end{multline}
     where $\delta' = \min\{\delta, 1\}$. In particular, $\const_1$ and $\const_2$ depends on $M$ polynomially.
  \end{theorem}

  The next theorem shows that the conformal prediction interval is approximately the oracle interval when the outcome model is well estimated.

  \begin{theorem}\label{thm:prediction_intervals}
  In the setting of Theorem \ref{thm:double_robustness_w}, assume further that
  \begin{enumerate}[label = C'\arabic*]
    \item There exist $\delta_0, \delta, M > 0$ such that $\lb\E\left[w(X)^{1 + \delta_0}\right]\rb^{1 / (1 + \delta_0)} \le M, \lb\E\left[\hat{w}(X)^{1 + \delta}\right]\rb^{1 / (1 + \delta)} \le M$;
      \item There exists $r > 0, s_0\in \calS$, $b_1 >0$, and a sequence of oracle nested sets $\{\calO_s(x)\}_{s\in\calS}$, such that
      \begin{enumerate}[label = (\roman*)]
      \item for any $\eps \in [0, r]$,
        \begin{align}
          \p\left( Y \in \calO_{s_0+\eps}(X) \mid X\right)
          \ge 1 - \alpha + b_1\eps, \quad \text{almost surely};
        \end{align}
        % \ejc{Notation above is not clear. Do you mean $\ldots \le \E\{1(\ldots) | X\} \le \ldots$ a.s.?}~\llzr{We have changed the notation.}
      \item there exist $k, \ell > 0$ such that $\displaystyle \lim_{N\rightarrow \infty} \E[\hat{w}(X)\Delta^{'k}(X)] = \lim_{N\rightarrow \infty} \E[w(X)\Delta^{'\ell}(X)] = 0$, where $\Delta'(x) =\sup_{s\in [s_0, s_0+r]}~\Delta'_s(x)$,
        \begin{align}
          \Delta'_s(x) = \inf\left\{\Delta\ge 0: 
          s+\Delta\in\calS , \calF_{s}(x; \calZ_\tr) \subset \calO_{s+\Delta}(x) \text{ and } \calO_{s}(x) \subset \calF_{s+\Delta}(x; \calZ_\tr)\right\}.
        \end{align}
      \end{enumerate}
    \end{enumerate}
    Then, for any $\beta\in (0, 1)$, 
    \[\p_{X \sim Q_{X}}\lb \hat{C}(X)\subset \mathcal{O}_{s_0 + \eps_{n}(\beta)}(X)\rb\ge 1 - \beta,\]
    where 
    \[\eps_{n}(\beta) = \frac{\const_{3}}{\beta^a}\log\lb\frac{1}{\beta}\rb \left\{\frac{1}{n^{\delta' / (1 + \delta')}} + \lb\E[\hat{w}(X)\Delta^{'k}(X)]\rb^{1/(1+k)} + \lb\E[w(X)\Delta^{'\ell}(X)]\rb^{1/\ell}\right\},\]
    the constant $\const_{3}$ only depends on $r, b_1, \delta, M, k, \ell$, and
    \[a = \min\left\{\frac{1}{2}, \frac{\delta_0}{1 + \delta_0}, \frac{1}{1+k}, \frac{1}{\ell}\right\}.\]
    In particular, $\const_{3}$ depends on $M$ polynomially.
  \end{theorem}

% Then $\p\lb \hat{w}(\td{X})\ge \frac{nb_1 \eps_n}{64}\rb\le \beta$ with
% \[\eps_{n} = \frac{1}{n}\cdot \max\left\{\frac{64}{b_1}\lb\frac{M}{\beta}\rb^{(1 + \delta_0) / \delta_0}, 1\right\}.\]

\subsection{Proof of Theorem \ref{thm:double_robustness_q}}

% We start with the following two Rosenthal-type inequalities for sums of independent random variables with finite $(1 + \delta)$-th moments.

  % \begin{proposition}[Theorem 3 of \cite{rosenthal1970subspaces}]\label{prop:rosenthal}
  %   Let $\{Z_{i}\}_{i = 1, \ldots, n}$ be independent mean-zero random
  %   variables. Then for any $\delta \ge 1$, there exists $L(\delta) > 0$
  %   that only depends on $\delta$ such that
  %   \[\E \bigg|\sum_{i=1}^{n}Z_{i}\bigg|^{1 + \delta}\le L(\delta)\left\{\sum_{i=1}^{n}\E |Z_{i}|^{1 + \delta} + \lb\sum_{i=1}^{n}\E |Z_{i}|^{2}\rb^{(1 + \delta) / 2}\right\}.\]
  % \end{proposition}

  % \begin{proposition}[Theorem 2 of \cite{vonbahr65}]\label{prop:vonbahresseen}
  %   Let $\{Z_{i}\}_{i = 1, \ldots, n}$ be independent mean-zero random variables. Then for any $\delta \in [0, 1)$,
  %   \[\E \bigg|\sum_{i=1}^{n}Z_{i}\bigg|^{1 + \delta}\le 2\sum_{i=1}^{n}\E |Z_{i}|^{1 + \delta}.\]
  % \end{proposition}
We first state a technical lemma whose proof is deferred to Section \ref{subapp:lemma_what}. 
\begin{lemma}\label{lem:what}
Under Assumption C1, there exists a constant $A > 0$ that only depends on $\delta$ and $M$, such that
    \[\p\lb \sum_{i=1}^{n}\hat{w}(X_i)^2\ge nt\rb\le \frac{An^{(1 - \delta')/2}}{t^{(1 + \delta')/2}}, \quad \text{for any }t > 0,\]
    and
    \[\p\lb \bigg|\sum_{i=1}^{n} \big(\hat{w}(X_i) - 1\big)\bigg|\ge \frac{n}{2}\rb\le \frac{A}{n^{(\delta' + \delta) / 2}},\]
    where $\delta' = \min\{\delta, 1\}$. In particular, $A$ depends on $M$ polynomially.
\end{lemma}

Throughout the proof, we treat $\Z_\tr$, hence $\hat{w}(\cdot)$ and  $\hat{F}_s(\cdot; \Z_\tr)$, as fixed. We shall prove the $\Z_\tr$-conditional versions of \eqref{eq:unconditional_coverage_rate_q} and \eqref{eq:conditional_coverage_rate_q} (with $\p(\cdot)$ and $\E[\cdot]$ replaced by $\p(\cdot \mid \Z_\tr)$ and $\E[\cdot\mid \Z_\tr]$). The results then follow by the law of iterated expectations and the H\"{o}lder's inequality.

By definition, $w(X)$ is almost surely finite under $P_{X}$.  Assumption C1 implies that $w(X)$ is almost surely finite under $Q_{X}$ and $\hat{w}(X)$ is almost surely finite under $P_{X}$. As a result, for any measurable function $f$,
  \begin{equation}
    \label{eq:transformation}
    \E_{X\sim Q_{X}}[f(X)] = \E_{X\sim P_{X}}[w(X)f(X)].
  \end{equation}
  In addition, the assumption $\E_{X\sim P_{X}}[\hat{w}(X)\mid \Z_{\tr}] < \infty$ implies that $\p_{X\sim P_{X}}(\hat{w}(X) < \infty) = 1$. By \eqref{eq:transformation},
\[\p_{X\sim Q_{X}}(\hat{w}(X) < \infty) = 1 - \E_{X\sim P_{X}}[w(X)I(\hat{w}(X) = \infty)].\]
Thus, $\hat{w}(X)$ is almost surely finite under $Q_{X}$. 
  % We will prove the following result on $\eta(x)$:
  % \begin{equation}
  %   \label{eq:eta}
  %   \lim_{N, n \rightarrow \infty}\p_{X\sim Q_{X}}(\eta(X)\ge -\eps) = 1, \quad \mbox{for any }\eps \in (0, r / 2).
  % \end{equation}

  Let $\eps < r / 3$ and $(\td{X}, \td{Y})$ denote a generic random vector drawn from
  $Q_{X}\times P_{Y\mid X}$, which is independent of the data. Then
\begin{align}
  		&\p\left(\tY \in \hat{C}(\tX) \mid \tX\right)\\
  & =  \p\left(V(\tX, \tY; \calZ_{\tr}) \le \eta(\tX) \mid \tX\right)\\
  & \ge  \p\left(\tY \in \calF_{\eta(\tX)}(\tX;\calZ_\tr) \mid \tX\right)\\
	& \ge \p\left(\tY \in \calF_{s_0 - \eps}(\tX;\calZ_\tr) \mid \tX\right)
		-\p\left(\eta(\tX) < s_0 - \eps \mid \tX\right)\\
  &\stackrel{(1)}{\ge} \p\left(\tY \in \calO_{s_0 - 2\eps - \Delta(\tX)}(\tX) \mid \tX \right) - 
    \p\left(\eta(\tX) < s_0 - \eps \mid \tX\right)\\
  & \ge  \p\left(\tY \in \calO_{s_0 - 2\eps - \Delta(\tX)\Indc\{\Delta(\tX)\le \eps\}  }(\tX) 
    \mid \tX \right) - \Indc\{\Delta(\tX) > \eps\} 
    - \p\left(\eta(\tX) < s_0 - \eps \mid \tX\right)\\
  & \stackrel{(2)}{\ge} 1 - \alpha - b_2\left(2\eps + \Delta(\tX)\Indc
    \{\Delta(\tX) \le \eps\}\right) - \Indc\{\Delta(\tX) > \eps\} 
    - \p\left(\eta(\tX) < s_0 - \eps \mid \tX\right)\\
  & \ge 1-\alpha- 3b_2\eps - \Indc\{\Delta(\tX) > \eps\} 
    - \p\left(\eta(\tX) < s_0 - \eps \mid \tX\right).\label{eq:conditional_coverage_nonasym}
\end{align}
Above, step (1) is due to the definition of $\Delta(x)$, and step (2) follows from Assumption C2 (i). 

Next, we derive an upper bound on
  $\p(\eta(\td{X}) < s_0-\eps \mid \td{X})$. Let $G$ denote the
  cumulative distribution function of the random distribution
  $\sum_{i=1}^{n}\hat{p}_{i}(\td{X})\delta_{V_{i}} +
  \hat{p}_{\infty}(\td{X})\delta_{\infty}$. Again, $G$ implicitly
  depends on $N$, $n$ and $\td{X}$. Then $\eta(\td{X}) < s_0-\eps$
  implies $G(s_0-\eps)\ge 1 - \alpha$, and thus,
  \[\p\lb\eta(\td{X}) < s_0-\eps\mid \td{X}\rb\le \p\lb G(s_0-\eps)\ge 1 - \alpha \mid \td{X}\rb, \,\, \text{a.s.}.\]
  Let $G^{*}(s_0-\eps)$ denote the expectation of $G(s_0-\eps)$ conditional on
  $\D = \{(X_{i})_{i=1}^{n}, \td{X}\}$, namely,
  \[G^{*}(s_0-\eps) = \E[G(s_0-\eps)\mid \D] = \sum_{i=1}^{n}\hat{p}_{i}(\td{X})\p(V_{i}\le s_0-\eps\mid \D).\]
  For any $t > 0$, the triangle inequality implies that
  \begin{align}
    &\p\lb\eta(\td{X}) < s_0-\eps\mid \td{X}\rb\nonumber\\
    & \le \p\lb G(s_0-\eps) - G^{*}(s_0-\eps)\ge t \mid \td{X}\rb + \p\lb G^{*}(s_0-\eps)\ge 1 - \alpha - t \mid \td{X}\rb, \,\, \text{a.s.}.    \label{eq:etaX1}
  \end{align}
  To bound the first term, we note that
  \[G(s_0-\eps) - G^{*}(s_0-\eps) = \sum_{i=1}^{n}\hat{p}_{i}(\td{X})\left\{
    I(V_{i}\le s_0 - \eps) - \p(V_{i}\le s_0 -\eps\mid \D)\right\}.\] Conditional on
  $\D$, $G(s_0-\eps) - G^{*}(s_0-\eps)$ is sub-Gaussian with
  parameter
  \[\hat{\sigma}^{2} = \sum_{i=1}^{n}\hat{p}_{i}(\td{X})^{2}.\]
  For any $t > 0$,
  \[\p\lb G(s_0-\eps) - G^{*}(s_0-\eps) \ge t\mid \D\rb\le \exp\lb
    -\frac{t^{2}}{2\hat{\sigma}^{2}}\rb.\]
  Let $\gamma_{n}$ be any fixed sequence with $\gamma_{n} = O(1)$. Taking expectation over $\D\setminus \{\td{X}\}$, we obtain that
  % \ejc{Arrange display a bit}
  \begin{align}
    &\p\lb G(s_0-\eps) - G^{*}(s_0-\eps) \ge t \mid \td{X}\rb\\
    & \le \E\left[\exp\lb -\frac{t^{2}}{2\hat{\sigma}^{2}}\rb\mid \td{X}\right] \\
    & \le \exp\lb - \frac{t^2}{2\gamma_{n}}\rb + \p\lb\hat{\sigma}^2 \ge \gamma_{n}\mid \td{X}\rb\\
    & = \exp\lb - \frac{t^2}{2\gamma_{n}}\rb + \p\lb \frac{\sum_{i=1}^{n}\hat{w}(X_i)^2}{\lb\sum_{i=1}^{n}\hat{w}(X_i) + \hat{w}(\td{X})\rb^2}\ge \gamma_{n}\mid \td{X}\rb\\
    & \le \exp\lb - \frac{t^2}{2\gamma_{n}}\rb + \p\lb \frac{\sum_{i=1}^{n}\hat{w}(X_i)^2}{\lb\sum_{i=1}^{n}\hat{w}(X_i)\rb^2}\ge \gamma_{n}\mid \td{X}\rb\\
    & \stackrel{(1)}{=} \exp\lb - \frac{t^2}{2\gamma_{n}}\rb + \p\lb \frac{\sum_{i=1}^{n}\hat{w}(X_i)^2}{\lb\sum_{i=1}^{n}\hat{w}(X_i)\rb^2}\ge \gamma_{n}\rb\\
    & \le \exp\lb - \frac{t^2}{2\gamma_{n}}\rb + \p\lb\sum_{i=1}^{n}\hat{w}(X_i)\le \frac{n}{2}\rb + \p\lb\sum_{i=1}^{n}\hat{w}(X_i)^2\ge \frac{n^2\gamma_{n}}{4}\rb\\
    & \le \exp\lb - \frac{t^2}{2\gamma_{n}}\rb + \p\lb\bigg|\sum_{i=1}^{n}\big(\hat{w}(X_i) - 1\big)\bigg|\ge \frac{n}{2}\rb + \p\lb\sum_{i=1}^{n}\hat{w}(X_i)^2\ge \frac{n^2\gamma_{n}}{4}\rb, \label{eq:etaX2_0}
  \end{align}
  where (1) uses the fact that $\td{X}$ is independent of
  $(\hat{w}(X_i))_{i=1}^{n}$. Note that this bound holds uniformly with $\td{X}$. Throughout the rest of the proof, we write
  $a_{1,n}\preceq a_{2,n}$ if there exists a constant $\const$ that only depends on $r, b_1, b_2, \delta, M, k, \ell$ and, in particular, on $M$ polynomially such that
  $a_{1,n}\le \const a_{2,n}$ for all $n$.

  ~\\
  \noindent By Lemma \ref{lem:what},
  \begin{align}
    &\p\lb G(s_0-\eps) - G^{*}(s_0-\eps) \ge t \mid \td{X}\rb\nonumber\\
    & \preceq \exp\lb - \frac{t^2}{2\gamma_{n}}\rb + \frac{n^{(1 - \delta')/2}}{(n\gamma_n)^{(1 + \delta') / 2}} + \frac{1}{n^{(\delta' + \delta) / 2}}\nonumber\\
    & \preceq \exp\lb - \frac{t^2}{2\gamma_{n}}\rb + \frac{1}{n^{\delta'}\gamma_n^{(1 + \delta') / 2}} + \frac{1}{n^{(\delta' + \delta) / 2}}\nonumber\\
    & \preceq \exp\lb - \frac{t^2}{2\gamma_{n}}\rb + \frac{1}{n^{\delta'}\gamma_n^{(1 + \delta') / 2}},\label{eq:etaX2}
  \end{align}
  where the last step follows from $\gamma_{n} = O(1)$ and $(\delta + \delta') / 2 \ge \delta'$.

  ~\\
  \noindent Next, we almost surely bound the term
  $\p\lb G^{*}(s_0-\eps)\ge 1 - \alpha - t \mid \td{X}\rb$. When $\eps < r/3$,
\begin{align}
   G^*(s_0 - \eps) 
  & = \sum^n_{i=1} \hat{p}_i(\tX) \p\left(V_i \le s_0 - \eps \mid \calD\right) \\
  & = \sum^n_{i=1} \hat{p}_i(\tX) \left[ \p\left(V_i \le s_0 - \eps ,\Delta(X_i) > \eps/2 \mid \calD \right)
  + \p\left(V_i \le s_0 -  \eps, \Delta(X_i) \le \eps/2 \mid \calD \right)\right]\\
  & \stackrel{(1)}{\le} \sum^n_{i=1} \hat{p}_i(\tX) \left[\Indc \left\{\Delta(X_i) > \eps/2\right\}
  + \p\left(V_i \le s_0 -  \eps/2 - \Delta(X_i), \Delta(X_i) \le \eps/2 \mid \calD \right)\right]\\
  & \le \sum^n_{i=1} \hat{p}_i(\tX) \left[\Indc \left\{\Delta(X_i) > \eps/2\right\}
  + \p\left(Y_i \in \calF_{s_0 - \eps/4 - \Delta(X_i)} (X_i; \calZ_\tr)\mid \calD \right)\right]\\
  & \stackrel{(2)}{\le} \sum^n_{i=1} \hat{p}_i(\tX) \left[ \Indc\left\{\Delta(X_i) > \eps/2 \right\}
  + \p\left(Y_i \in \calO_{s_0 -  \eps/8}(X_i) \mid \calD \right)\right]\\ 
  & \stackrel{(3)}{\le} \sum^n_{i=1} \hat{p}_i(\tX) \Indc\left\{\Delta(X_i) > \eps/2\right\}
  + 1-\alpha - \frac{b_1\eps}{8},\label{eq:etaX3}
\end{align}
where step (1) holds because $\Delta(X_i)$ is deterministic
conditional on $\calD$, step (2) follows from the definition of
$\Delta(X_i)$, and step (3) follows from the Assumption C2 (i). For any $t \le b_1\eps / 16$,
\begin{align}
  &\p\lb G^*(s_0 - \eps) \ge 1 - \alpha - t\mid \td{X}\rb\\
  & \le \p\lb\sum^n_{i=1} \hat{p}_i(\tX) \Indc\left\{\Delta(X_i) > \eps/2\right\}\ge \frac{b_1\eps}{8} - t\mid \td{X}\rb\\
  & \le \p\lb\frac{\sum_{i=1}^{n}\hat{w}(X_i)I(\Delta(X_i)\ge \eps / 2)}{\sum_{i=1}^{n}\hat{w}(X_i)}\ge \frac{b_1\eps}{16}\mid \td{X}\rb\\
  & \le \p\lb\sum_{i=1}^{n}\hat{w}(X_i)\le \frac{n}{2}\rb + \p\lb \sum_{i=1}^{n}\hat{w}(X_i)I\lb \Delta(X_i)\ge \frac{\eps}{2}\rb\ge \frac{nb_1 \eps}{32}\rb\\
  & \le \p\lb\bigg|\sum_{i=1}^{n}\big(\hat{w}(X_i) - 1\big)\bigg|\ge \frac{n}{2}\rb + \p\lb \sum_{i=1}^{n}\hat{w}(X_i)I\lb \Delta(X_i)\ge \frac{\eps}{2}\rb\ge \frac{nb_1 \eps}{32}\rb.   \label{eq:etaX4}
\end{align}
By Markov's inequality,
\begin{align*}
  \p\lb \sum_{i=1}^{n}\hat{w}(X_i)I\lb \Delta(X_i)\ge \frac{\eps}{2}\rb\ge \frac{n b_1\eps}{32}\rb &\preceq \frac{1}{\eps}\E\left[\hat{w}(X)I\lb \Delta(X)\ge \frac{\eps}{2}\rb\right]\preceq \frac{\E[\hat{w}(X)\Delta^{k}(X)]}{\eps^{1+k}},
\end{align*}
where the last step uses the simple fact that $I(\Delta(X_i)\ge \eps / 2)\le (2/\eps)^{k}\Delta^{k}(X_i)$. Further, by Lemma \ref{lem:what}, for any $t\le \eps b_1 / 16$,
\begin{equation}
  \label{eq:etaX5}
  \p\lb G^*(s_0 - \eps) \ge 1 - \alpha - t\mid \td{X}\rb\preceq \frac{1}{n^{(\delta' + \delta) / 2}} + \frac{\E[\hat{w}(X)\Delta^{k}(X)]}{\eps^{1+k}}.
\end{equation}

~\\
  \noindent Combining \eqref{eq:etaX1} - \eqref{eq:etaX5} together and setting $t = \eps b_1 / 16$, we obtain that for any sequence $\gamma_{n} = O(1)$, 
  \begin{align}
    \p\lb\eta(\td{X}) < s_0-\eps\mid \td{X}\rb &\preceq \exp\lb - \frac{b_1^2}{512}\frac{\eps^2}{\gamma_{n}}\rb + \frac{1}{n^{\delta'}\gamma_{n}^{(1 + \delta') / 2}} + \frac{1}{n^{(\delta + \delta') / 2}} + \frac{\E[\hat{w}(X)\Delta^{k}(X)]}{\eps^{1+k}}\nonumber\\
                                            & \preceq \exp\lb - \frac{b_1^2}{512}\frac{\eps^2}{\gamma_{n}}\rb + \frac{1}{n^{\delta'}\gamma_{n}^{(1 + \delta') / 2}} + \frac{\E[\hat{w}(X)\Delta^{k}(X)]}{\eps^{1+k}}. \label{eq:etaX6}
  \end{align}
  Substitute $\eps$ with $\eps_{n}$ and assume $\eps_{n} \le r / 3$
  (recall the beginning of the proof). Set
  \begin{equation}
    \label{eq:gamman}
    \gamma_{n} = \frac{512}{b_1^2}\frac{\eps_{n}^2}{\log n}.
  \end{equation}
  Clearly, $\gamma_{n} = o(1)$. Then the first term of \eqref{eq:etaX6} is $1 / n$, and thus,
\[\p\lb\eta(\td{X}) < s_0-\eps_{n}\mid \td{X}\rb \preceq \frac{(\log n)^{(1 + \delta') / 2}}{n^{\delta'}\eps_{n}^{1 + \delta'}} + \frac{\E[\hat{w}(X)\Delta^{k}(X)]}{\eps_{n}^{1+k}}.\]
  Equivalently, there exists a constant $\const$ that only depends on $r, b_1, b_2, \delta, M, k, \ell$ and, in particular, on $M$ polynomially, such that
  \[\p\lb\eta(\td{X}) < s_0-\eps_{n}\mid \td{X}\rb \le \const\left\{\frac{(\log n)^{(1 + \delta') / 2}}{n^{\delta'}\eps_{n}^{1 + \delta'}} + \frac{\E[\hat{w}(X)\Delta^{k}(X)]}{\eps_{n}^{1+k}}\right\}, \,\, \text{a.s.}.\]
  Together with \eqref{eq:conditional_coverage_nonasym}, it implies that
  % \ejc{displays}
  \begin{multline*}
    \p(\td{Y} \in \hat{C}(\td{X}) \mid \td{X})\\
    \ge 1-\alpha- 3b_2\eps - \Indc\{\Delta(\tX) > \eps\} 
    - \const\left\{\frac{(\log n)^{(1 + \delta') / 2}}{n^{\delta'}\eps_{n}^{1 + \delta'}} + \frac{\E[\hat{w}(X)\Delta^{k}(X)]}{\eps_{n}^{1+k}}\right\},
  \end{multline*}
  almost surely. Assume $\const\ge 3b_2$ without loss of generality. Then
  \begin{multline}
    \p\lb\p(\td{Y} \in \hat{C}(\td{X}) \mid \td{X}) \le 1 - \alpha - \const\left\{\eps_{n} + \frac{(\log n)^{(1 + \delta') / 2}}{n^{\delta'}\eps_{n}^{1 + \delta'}} + \frac{\E[\hat{w}(X)\Delta^{k}(X)]}{\eps_{n}^{1+k}}\right\}\rb\\
    \quad \le \p(\Delta(\td{X}) > \eps_{n}).\label{eq:conditional_coverage2}
  \end{multline}
  For any $\beta\in (0, 1)$, let
  \[\eps_{n} = \frac{(\log n)^{(1 + \delta') / 2(2 + \delta')}}{n^{\delta' / (2 + \delta')}} + \lb\E[\hat{w}(X)\Delta^{k}(X)]\rb^{1 / (2+k)} + \frac{\lb\E[w(X)\Delta^{\ell}(X)]\rb^{1/\ell}}{\beta^{1/\ell}}.\]
  Then
  \begin{equation}
    \label{eq:eps_cond}
    \frac{(\log n)^{(1 + \delta') / 2}}{n^{\delta'}\eps_{n}^{1 + \delta'}}, \frac{\E[\hat{w}(X)\Delta^{k}(X)]}{\eps_{n}^{1+k}}\le \eps_{n},
  \end{equation}
  and by Markov's inequality and \eqref{eq:transformation},
  \[\p(\Delta(\td{X}) > \eps_{n})\le \frac{\E[\Delta^{\ell}(\td{X})]}{\eps_{n}^{\ell}} = \frac{\E[w(X)\Delta^{\ell}(X)]}{\eps_{n}^{\ell}}\le \beta.\]
  Furthermore, Assumption C2 (ii) implies that $\eps_{n}\le r / 2$ when $N\ge N(r)$ and $n\ge n(r)$ for some constants $N(r), n(r)$ that only depend on $r$. Replacing $\const$ by $3\const$, we obtain that, for $N\ge N(r), n\ge n(r)$, 
  \[\p\lb\p(\td{Y} \in \hat{C}(\td{X}) \mid \td{X}) \le 1 - \alpha - \const\eps_{n}\rb\le \beta.\]
  We can further enlarge $\const$ so that $B\eps_{n} \ge 1 - \alpha$ when $N < N(r)$ or $n < n(r)$, in which case \eqref{eq:conditional_coverage_rate_q} trivially holds.
  % Since it holds for any $\beta \in (0, 1)$, we conclude that
  % \begin{align*}
  %   &\left\{1 - \alpha - \p(\td{Y} \in \hat{C}(\td{X}) \mid \td{X})\right\}_{+} = O_\p(\eps_{n})\\
  %   & = O_\p\lb \frac{(\log n)^{(1 + \delta') / 2(2 + \delta')}}{n^{\delta' / (2 + \delta')}} + \lb\E[\hat{w}(X)\Delta^{k}(X)]\rb^{1 / (2+k)} + \lb\E[w(X)\Delta^{\ell}(X)]\rb^{1/\ell}\rb.
  % \end{align*}

  ~\\
  \noindent To prove the unconditional result, we note that \eqref{eq:conditional_coverage2} implies
  \[\p(\td{Y} \in \hat{C}(\td{X}))\ge \lb 1 - \alpha - \const\left\{\eps_{n} + \frac{(\log n)^{(1 + \delta') / 2}}{n^{\delta'}\eps_{n}^{1 + \delta'}} + \frac{\E[\hat{w}(X)\Delta^{k}(X)]}{\eps_{n}^{1+k}}\right\}\rb(1 - \p(\Delta(\td{X}) > \eps_{n})).\]
  Let
  \[\eps_{n} = \frac{(\log n)^{(1 + \delta') / 2(2 + \delta')}}{n^{\delta' / (2 + \delta')}} + \lb\E[\hat{w}(X)\Delta^{k}(X)]\rb^{1 / (2+k)} + \lb\E[w(X)\Delta^{\ell}(X)]\rb^{1/ (1+\ell)}.\]
  Then \eqref{eq:eps_cond} remains to hold. % \ejc{I don't understand what the last sentence mean.} \lihua{(A.19) implies that the first term $\eps_{n}$ denominates the other two terms in the expression $B\{\eps_{n} + \ldots + \ldots\}$. That's why we can rewrite the whole term as $3B\eps_{n}$. } 
  By Markov's inequality and \eqref{eq:transformation},
  \[\p(\Delta(\td{X}) > \eps_{n})\le \frac{\E[\Delta^{\ell}(\td{X})]}{\eps_{n}^{\ell}} = \frac{\E[w(X)\Delta^{\ell}(X)]}{\eps_{n}^{\ell}}\le \eps_{n}.\]
  Furthermore, Assumption (4) implies that $\eps_{n}\le r / 2$ when $N$ and $n$ are sufficiently large, in which case,
  \[\p(\td{Y} \in \hat{C}(\td{X}))\ge \lb 1 - \alpha - 3\const\eps_{n}\rb(1 - \eps_{n})\ge 1 - \alpha - (3\const+1)\eps_{n}.\]
  Similar to \eqref{eq:conditional_coverage_rate_q}, we can enlarge the constant to make \eqref{eq:unconditional_coverage_rate_q} hold % \ejc{To make (A.2)?}\lihua{``to make (A.2) hold''}
  when $N$ or $n$ is not sufficiently large.

  \subsection{Proof of Theorem \ref{thm:prediction_intervals}}
    % Using the same notation as in the proof of Theorem \ref{thm:double_robustness_q}, we shall prove a result similar to \eqref{eq:etaX5}:
    % \begin{equation}
    %   \p\lb\eta(\td{X}) > s_0+\eps\mid \td{X}\rb \preceq \exp\lb - \frac{b_1^2}{512}\frac{\eps^2}{\gamma_{n}}\rb + \frac{1}{n^{\delta'}\gamma_{n}^{(1 + \delta') / 2}} + \frac{\E[\hat{w}(X)\Delta^{k}(X)]}{\eps^{1+k}}.      \label{eq:pred_int_goal}
    % \end{equation}
    Let $G$ denote the
    cumulative distribution function of the random distribution
  $\sum_{i=1}^{n}\hat{p}_{i}(\td{X})\delta_{V_{i}} +
  \hat{p}_{\infty}(\td{X})\delta_{\infty}$ and $G^{*}$ denote the expectation of $G$ conditional on $\D = \{(X_{i})_{i=1}^{n}, \td{X}\}$. Clearly, for any $\eps > 0$,
    \[\eta(\td{X}) > s_0+\eps \Longrightarrow G(s_0+\eps)\le 1 - \alpha.\]
  %   Thus,
  % \[\p\lb\eta(\td{X}) > s_0+\eps\mid \td{X}\rb\le \p\lb G(s_0+\eps)\le 1 - \alpha + \max_{i\in \{1, \ldots,n, \infty\}}\hat{p}_{i}(\td{X})\mid \td{X}\rb, \,\, \text{a.s.}.\]
  Then, for any $t > 0$, the triangle inequality implies that, 
  \begin{align}
    &\p\lb\eta(\td{X}) > s_0+\eps\rb \\
    & \le \p\lb G(s_0+\eps) - G^{*}(s_0+\eps)\le -t \rb + \p\lb G^{*}(s_0+\eps)\le 1 - \alpha  + t \rb\label{eq:etaX0_pred}
  \end{align}
  Throughout the rest of the proof, we set $t = b_1 \eps / 16$. 

  ~\\
  \noindent
  % As shown in the proof of Theorem \ref{thm:double_robustness_q}, conditional on
  % $\D$, $G(s_0+\eps) - G^{*}(s_0+\eps)$ is sub-Gaussian with
  % parameter
  % \[\hat{\sigma}^{2} = \sum_{i=1}^{n}\hat{p}_{i}(\td{X})^{2}.\]
  % As a result, for any $t > 0$,
  % \[\p\lb G(s_0+\eps) - G^{*}(s_0+\eps) \le -t\mid \D\rb\le \exp\lb
  %   -\frac{t^{2}}{2\hat{\sigma}^{2}}\rb.\]
Analogous to \eqref{eq:etaX2}, with probability $1$,
  % \ejc{Arrange display a bit}
  \begin{align}
    &\p\lb G(s_0+\eps) - G^{*}(s_0+\eps) \le -t \mid \td{X}\rb \preceq \inf_{\gamma > 0}\left\{\exp\lb - \frac{b^{2}_1}{512}\frac{\eps^2}{\gamma}\rb + \frac{1}{n^{\delta'}\gamma^{(1 + \delta') / 2}}\right\}.
  \end{align}
  Let $\gamma = \eps^2 / c$ for some $c > 0$. Marginalizing over $\td{X}$ implies
  \begin{align}
    &\p\lb G(s_0+\eps) - G^{*}(s_0+\eps) \le -t \rb \preceq \inf_{c > 0}\left\{\exp\lb - \frac{b^{2}_1 c}{512}\rb + \frac{c^{(1 + \delta') / 2}}{n^{\delta'}\eps^{1 + \delta'}}\right\}.\label{eq:etaX2_pred}
  \end{align}

  ~\\
  \noindent Next, we almost surely bound the term
  $\p\lb G^{*}(s_0+\eps)\le 1 - \alpha + t \mid \td{X}\rb$. When $\eps < \min\{r/3, 8\alpha/b_1\}$,
\begin{align}
   G^*(s_0 + \eps) 
  & = \sum^n_{i=1} \hat{p}_i(\tX) \p\left(V_i \le s_0 + \eps \mid \calD\right) \nonumber\\
  & \ge \sum^n_{i=1} \hat{p}_i(\tX)\p\left(V_i \le s_0 + \eps, \Delta'(X_i) \le \eps/2 \mid \calD \right)\nonumber\\
  % & \ge (1 - \hat{p}_{\infty}(\td{X}))\p\left(V_i \le s_0 +  \eps/2 + \Delta'(X_i), \Delta'(X_i) \le \eps/2 \mid \calD \right)\nonumber\\
  & \ge \sum^n_{i=1} \hat{p}_i(\tX)\p\left(V_i \le s_0 + \eps/2 + \Delta'(X_i), \Delta'(X_i) \le \eps/2 \mid \calD \right)\nonumber\\
  & \ge \sum^n_{i=1}\hat{p}_i(\tX)\left[\p\left(V_i \le s_0 +  \eps/2 + \Delta'(X_i) \mid \calD \right) - \Indc\left(\Delta'(X_i) > \eps/2\mid \calD\right)\right]\nonumber\\
  & \ge \sum^n_{i=1} \hat{p}_i(\tX) \left[\p\left(Y_i \in \calF_{s_0 + \eps/4 + \Delta'(X_i)} (X_i; \calZ_\tr)\mid \calD \right) - \Indc \left\{\Delta'(X_i) > \eps/2\right\}
  \right]\nonumber\\
  & \ge \sum^n_{i=1} \hat{p}_i(\tX) \left[ \p\left(Y_i \in \calO_{s_0 + \eps/8}(X_i) \mid \calD \right) - \Indc\left\{\Delta'(X_i) > \eps/2 \right\}
  \right]\nonumber\\ 
  & \ge \lb 1-\alpha + \frac{b_1\eps}{8}\rb \left[1 - \hat{p}_{\infty}(\td{X})\right]-\sum^n_{i=1} \hat{p}_i(\tX) \Indc\left\{\Delta'(X_i) > \eps/2\right\}\nonumber\\
  & \ge 1-\alpha + \frac{b_1\eps}{8} - \td{p}_{\infty}(\td{X}) - \sum^n_{i=1} \hat{p}_i(\tX) \Indc\left\{\Delta'(X_i) > \eps/2\right\}.
\end{align}
Then
\begin{align}
  &\p\lb G^*(s_0 + \eps) \le 1 - \alpha + t\rb\\
  & \le \p\lb\sum^n_{i=1} \hat{p}_i(\tX) \Indc\left\{\Delta'(X_i) > \eps/2\right\}\ge \frac{b_1\eps}{32} \,\,\text{or }\,\,\td{p}_{\infty}(\td{X})\ge \frac{b_1\eps}{32}\rb\\
  & \le \p\lb\frac{\sum_{i=1}^{n}\hat{w}(X_i)I(\Delta'(X_i)\ge \eps / 2)}{\sum_{i=1}^{n}\hat{w}(X_i)}\ge \frac{b_1\eps}{32}\,\, \text{or }\,\,\frac{\hat{w}(\td{X})}{\sum_{i=1}^{n}\hat{w}(X_i)}\ge \frac{b_1\eps}{32}\rb\\
  & \le \p\lb\sum_{i=1}^{n}\hat{w}(X_i)\le \frac{n}{2}\rb + \p\lb \sum_{i=1}^{n}\hat{w}(X_i)I\lb \Delta'(X_i)\ge \frac{\eps}{2}\rb\ge \frac{nb_1 \eps}{64}\rb \\
  & \qquad \qquad+ \p\lb \hat{w}(\td{X})\ge \frac{nb_1 \eps}{64}\rb.\label{eq:etaX3_pred}
\end{align}
As shown in the proof of Theorem \ref{thm:double_robustness_q},
\[\p\lb\sum_{i=1}^{n}\hat{w}(X_i)\le \frac{n}{2}\rb + \p\lb \sum_{i=1}^{n}\hat{w}(X_i)I\lb \Delta'(X_i)\ge \frac{\eps}{2}\rb\ge \frac{nb_1 \eps}{64}\rb\preceq \frac{1}{n^{(\delta' + \delta) / 2}} + \frac{\E[\hat{w}(X)\Delta^{'k}(X)]}{\eps^{1+k}}.\]
By H\"{o}lder's inequality and Markov's inequality,
\begin{align}
  &\p\lb \hat{w}(\td{X})\ge \frac{nb_1 \eps}{64}\rb = \E\left[w(X)\Indc\lb \hat{w}(X)\ge \frac{nb_1 \eps}{64}\rb\right]\\
  & \le \lb\E[w(X)^{1 + \delta_0}]\rb^{1/(1 + \delta_0)}\p\lb\hat{w}(X)\ge \frac{nb_1 \eps}{64}\rb^{\delta_0/(1 + \delta_0)}\\
  % & \preceq M\lb\frac{64}{nb_1 \eps_n}\rb^{\delta_0/(1 + \delta_0)}.
      & \preceq \frac{1}{(n\eps)^{\delta_0 / (1 + \delta_0)}}.\label{eq:wtdX}
\end{align}
By \eqref{eq:etaX3_pred},
\begin{equation}
  \label{eq:etaX4_pred}
  \p\lb G^*(s_0 + \eps) \le 1 - \alpha + t\rb\preceq \frac{1}{n^{(\delta' + \delta) / 2}} + \frac{1}{(n\eps)^{\delta_0 / (1 + \delta_0)}} + \frac{\E[\hat{w}(X)\Delta^{'k}(X)]}{\eps^{1+k}}.
\end{equation}

% ~\\
% \noindent Finally, we bound the last term in \eqref{eq:etaX0_pred}. By definition,
% \begin{align}
%   &\p\lb \max_{i\in \{1, \ldots,n, \infty\}}\hat{p}_{i}(\td{X})\ge t\rb  \le \p\lb \frac{\max\{\hat{w}(\td{X}), \max_{i\in \{1, \ldots, n\}}\hat{w}(X_i)\}}{\sum_{i=1}^{n}\hat{w}(X_i)}\ge t\rb\\
%   % & \le \p\lb\sum_{i=1}^{n}\hat{w}(X_i)\le \frac{n}{2}\rb + \p\lb\max_{i\in \{1, \ldots, n\}}\hat{w}(X_i) \ge \frac{nt}{2}\rb\\
%   & \le \p\lb\sum_{i=1}^{n}\hat{w}(X_i)\le \frac{n}{2}\rb + \p\lb \hat{w}(\td{X})\ge \frac{nb_1 \eps}{64}\rb+ \p\lb\max_{i\in \{1, \ldots, n\}}\hat{w}(X_i) \ge \frac{nb_1 \eps}{64}\rb.\label{eq:etaX5_pred}
% \end{align}
% By Markov's inequality,
% \begin{align}
%   &\p\lb\max_{i\in \{1, \ldots, n\}}\hat{w}(X_i) \ge \frac{nb_1 \eps}{64}\rb\preceq \frac{1}{(n\eps)^{1 + \delta}}\E\left[\max_{i\in \{1, \ldots, n\}}\hat{w}^{1 + \delta}(X_i)\right]\\
%   & \le \lb\frac{64}{nb_1\eps}\rb^{1 + \delta}\sum_{i=1}^{n}\E\left[\hat{w}^{1 + \delta}(X_i)\right] \preceq \frac{1}{n^{\delta}\eps^{1 + \delta}}.
% \end{align}
% By Lemma \ref{lem:what}, \eqref{eq:wtdX}, and \eqref{eq:etaX5_pred},
% \begin{equation}
%   \label{eq:etaX6_pred}
%   \p\lb \max_{i\in \{1, \ldots,n, \infty\}}\hat{p}_{i}(\td{X})\ge t\rb  \preceq \frac{1}{n^{(\delta' + \delta)/2}} + \frac{1}{(n\eps)^{\delta_0 / (1 + \delta_0)}} + \frac{1}{n^{\delta}\eps^{1 + \delta}}.
% \end{equation}

 ~\\
 \noindent Putting \eqref{eq:etaX0_pred}, \eqref{eq:etaX2_pred}, and \eqref{eq:etaX4_pred} together, we obtain that, for any $c > 0$, 
\begin{align}
  &\p\lb\eta(\td{X}) > s_0+\eps\rb\\
  & \preceq \exp\lb - \frac{b^{2}_1 c}{512}\rb + \frac{c^{(1 + \delta')/2}}{n^{\delta'}\eps^{1 + \delta'}} + \frac{1}{n^{(\delta' + \delta) / 2}} + \frac{1}{(n\eps)^{\delta_0 / (1 + \delta_0)}} + \frac{\E[\hat{w}(X)\Delta^{'k}(X)]}{\eps^{1+k}}. \label{eq:etaX7_pred}
\end{align}
%   Substitute $\eps$ with $\eps_{n}$ and assume $1/n \le \eps_{n} \le \min\{r / 3, 8\alpha / b_1\}$ (recall the proof of \eqref{eq:etaX3_pred} for the upper bound of $\eps_{n}$). Set
% \[\gamma_{n} = \frac{512}{b_1^2}\frac{\eps_{n}^2}{\log n}.\]
% Then the first term in \eqref{eq:etaX7_pred} is $1/n$ and hence in a smaller order than the second term. As a result,
% \begin{align}
%   &\p\lb\eta(\td{X}) > s_0+\eps_n\mid \td{X}\rb\\
%   & \preceq \frac{1}{n^{\delta'}\eps_n^{1 + \delta'}} + \frac{1}{n^{(\delta' + \delta) / 2}} + \frac{1}{n^{\delta}\eps_n^{1 + \delta}} + \frac{\E[\hat{w}(X)\Delta^{'k}(X)]}{\eps_n^{1+k}} + \Indc\lb \hat{w}(\td{X})\ge \frac{nb_1 \eps_n}{64}\rb\\
%   & \preceq \frac{1}{n^{\delta'}\eps_n^{1 + \delta'}} + \frac{1}{n^{(\delta' + \delta) / 2}} + \frac{\E[\hat{w}(X)\Delta^{'k}(X)]}{\eps_n^{1+k}} + \Indc\lb \hat{w}(\td{X})\ge \frac{nb_1 \eps_n}{64}\rb,
% \end{align}
% where the last line is due to $\delta \ge \delta'$ and $n\eps_n \ge 1$. Marginalizing over $\td{X}$, 
% \[\p\lb\eta(\td{X}) > s_0+\eps_n\rb\preceq \frac{1}{n^{\delta'}\eps_n^{1 + \delta'}} + \frac{1}{n^{(\delta' + \delta) / 2}} + \frac{\E[\hat{w}(X)\Delta^{'k}(X)]}{\eps_n^{1+k}} + \p\lb\hat{w}(\td{X})\ge \frac{nb_1 \eps_n}{64}\rb.\]
Applying a union bound, we have
\begin{align}
  \label{eq:etaX8_pred}
  &\p\lb\eta(\td{X}) > s_0+\eps\text{ or }\Delta'(\td{X}) > \eps\rb\\
  & \preceq \exp\lb - \frac{b^{2}_1 c}{512}\rb + \frac{c^{(1 + \delta')/2}}{n^{\delta'}\eps^{1 + \delta'}} + \frac{1}{n^{(\delta' + \delta) / 2}} + \frac{1}{(n\eps)^{\delta_0 / (1 + \delta_0)}} \\
  &\qquad + \frac{\E[\hat{w}(X)\Delta^{'k}(X)]}{\eps^{1+k}} + \p(\Delta'(\td{X}) > \eps)\\
    & \preceq \exp\lb - \frac{b^{2}_1 c}{512}\rb + \frac{c^{(1 + \delta')/2}}{n^{\delta'}\eps^{1 + \delta'}} + \frac{1}{n^{(\delta' + \delta) / 2}} + \frac{1}{(n\eps)^{\delta_0 / (1 + \delta_0)}} \\
  &\qquad  + \frac{\E[\hat{w}(X)\Delta^{'k}(X)]}{\eps^{1+k}} + \frac{\E[w(X)\Delta^{'\ell}(X)]}{\eps^{\ell}}
  %\label{eq:etaX7_pred}
\end{align}
where the last line applies Markov's inequality which yields 
\begin{align}
  \p\lb \Delta'(\td{X}) > \eps\rb\le \frac{\E[\Delta^{'\ell}(\td{X})]}{\eps^{\ell}} = \frac{\E[w(X)\Delta^{'\ell}(X)]}{\eps^{\ell}}.
\end{align}
Thus, there exists a constant $\const$ that only depends on $r, b_1, \delta, M, k, \ell$ and, in particular, on $M$ polynomially such that
\begin{align}
  &\p\lb\eta(\td{X}) > s_0+\eps \text{ or }\Delta'(\td{X}) > \eps\rb\\
  & \le \const\bigg\{\exp\lb - \frac{b^{2}_1 c}{512}\rb + \frac{c^{(1 + \delta')/2}}{n^{\delta'}\eps^{1 + \delta'}} + \frac{1}{n^{(\delta' + \delta) / 2}} + \frac{1}{(n\eps)^{\delta_0 / (1 + \delta_0)}}\\
  & \quad\qquad + \frac{\E[\hat{w}(X)\Delta^{'k}(X)]}{\eps^{1+k}} + \frac{\E[w(X)\Delta^{'\ell}(X)]}{\eps^{\ell}}\bigg\}.
%\label{eq:etaX8_pred}
\end{align}
Let $\td{B} = 10B / \beta$, 
\[c = \frac{512\log \td{B}}{b^{2}_1},\]
and
\[\eps_n = \frac{\td{B}^{1/(1+\delta')}c^{1/2}}{n^{\delta' / (1 + \delta')}} + \frac{\td{B}^{(1 + \delta_0) / \delta_0}}{n} + \lb \td{B}\E[\hat{w}(X)\Delta^{'k}(X)]\rb^{1/(1+k)} + \lb \td{B}\E[w(X)\Delta^{'\ell}(X)]\rb^{1/\ell}.\]
Then
\[\p\lb\eta(\td{X}) > s_0+\eps_n \text{ or }\Delta'(\td{X}) > \eps_n\rb\le \frac{\beta}{2} + \frac{B}{n^{(\delta' + \delta) / 2}}.\]
For a sufficiently large $n$,
\[\p\lb\eta(\td{X}) > s_0+\eps_n \text{ or }\Delta'(\td{X}) > \eps_n\rb\le \beta.\]
Recall that
\begin{align}
  \hat{C}(x) &= \{y: V(x, y; \calZ_\tr)\le \eta(x)\} = \{y: \inf\{s\in \calS: y \in \calF_s(x; \calZ_\tr)\}\le \eta(x)\} \subset \calF_{\eta(x) + \eps_n}(x; \calZ_\tr).
\end{align}
Thus, on the event $\eta(\td{X}) \le s_0+\eps$ and $\Delta'(\td{X}) \le \eps$, which has probability at least $1 - \beta$,
\[\hat{C}(\td{X})\subset \calF_{s_0 + 2\eps_n}(\td{X}; \calZ_\tr)\subset \calO_{s_0 + 3\eps_n}(\td{X}).\]
The proof is then completed by noting that
\begin{align}
  \eps_{n}&\preceq \left\{ \frac{1}{\beta^{1/(1+\delta')}}\log\lb\frac{1}{\beta}\rb + \frac{1}{\beta^{\delta_0/(1+\delta_0)}} + \frac{1}{\beta^{1/(1+k)}} + \frac{1}{\beta^{1/\ell}}\right\}\\
          & \quad \cdot \left\{\frac{1}{n^{\delta' / (1 + \delta')}} + \lb\E[\hat{w}(X)\Delta^{'k}(X)]\rb^{1/(1+k)} + \lb\E[w(X)\Delta^{'\ell}(X)]\rb^{1/\ell}\right\}\\
          & \preceq \frac{1}{\beta^a}\log\lb\frac{1}{\beta}\rb \left\{\frac{1}{n^{\delta' / (1 + \delta')}} + \lb\E[\hat{w}(X)\Delta^{'k}(X)]\rb^{1/(1+k)} + \lb\E[w(X)\Delta^{'\ell}(X)]\rb^{1/\ell}\right\}.
\end{align}

\subsection{Proof of Lemma \ref{lem:what}}\label{subapp:lemma_what}
The result is almost the same as a step in the proof of Theorem A.1 in \cite{lei2020conformal}. We repeat the proof below for the sake of completeness.

We start with the following two Rosenthal-type inequalities for sums of independent random variables with finite $(1 + \delta)$-th moments.

\begin{proposition}[Theorem 3 of \cite{rosenthal1970subspaces}]\label{prop:rosenthal}
  Let $\{Z_{i}\}_{i = 1, \ldots, n}$ be independent mean-zero random
  variables. Then for any $\delta \ge 1$, there exists $L(\delta) > 0$
  that only depends on $\delta$ such that
\[\E \bigg|\sum_{i=1}^{n}Z_{i}\bigg|^{1 + \delta}\le L(\delta)\left\{\sum_{i=1}^{n}\E |Z_{i}|^{1 + \delta} + \lb\sum_{i=1}^{n}\E |Z_{i}|^{2}\rb^{(1 + \delta) / 2}\right\}.\]
\end{proposition}

\begin{proposition}[Theorem 2 of \cite{vonbahr65}]\label{prop:vonbahresseen}
 Let $\{Z_{i}\}_{i = 1, \ldots, n}$ be independent mean-zero random variables. Then for any $\delta \in [0, 1)$,
\[\E \bigg|\sum_{i=1}^{n}Z_{i}\bigg|^{1 + \delta}\le 2\sum_{i=1}^{n}\E |Z_{i}|^{1 + \delta}.\]
\end{proposition}

To prove the lemma, we consider two cases:
% \ejc{Arrange displays below}
\begin{enumerate}[label=(\arabic*)]
\item If $\delta \ge 1$, then by Markov's inequality,
  \[\p\lb\sum_{i=1}^{n}\hat{w}(X_i)^2\ge nt\rb\le \frac{\E[\sum_{i=1}^{n}\hat{w}(X_i)^2]}{nt} = \frac{\E[\hat{w}(X_1)^2]}{t}\le \frac{M^2}{t}.\]
  Since $\E[\hat{w}(X_i)\mid \Z_\tr] = 1$, $\E[\hat{w}(X_i)] = 1$. By Markov's inequality and Proposition \ref{prop:rosenthal},
  \begin{align}
    &\p\lb \bigg|\sum_{i=1}^{n}\hat{w}(X_i) - 1\bigg|\ge \frac{n}{2}\rb \nonumber\\
    & \le \frac{2^{1 + \delta}}{n^{1 + \delta}}\E \bigg|\sum_{i=1}^{n}\hat{w}(X_i) - \E[\hat{w}(X_i)]\bigg|^{1 + \delta}\nonumber\\
    & \le \frac{2^{1+ \delta}L(\delta)}{n^{1 + \delta}}\left\{ n\E|\hat{w}(X_i) - \E[\hat{w}(X_i)]|^{1 + \delta} + n^{(1 + \delta) / 2}\lb\E|\hat{w}(X_i) - \E[\hat{w}(X_i)]|^{2}\rb^{(1 + \delta) / 2}\right\}\nonumber\\
    & \stackrel{(i)}{\le} \frac{2^{2(1 + \delta)}L(\delta)}{n^{1 + \delta}}\left\{ n\E|\hat{w}(X_i)|^{1 + \delta} + n^{(1 + \delta) / 2}(\E|\hat{w}(X_i)|^2)^{(1 + \delta) / 2}\right\}\nonumber\\
    & \le \frac{2^{3+2\delta}L(\delta)M^{1 + \delta}}{n^{(1 + \delta) / 2}},\label{eq:denom_case1}
  \end{align}
  where (i) follows from the H\"{o}lder's
  inequality which gives 
  \[\E|\hat{w}(X_i) - \E[\hat{w}(X_i)]|^{1 + \delta}\le 2^{\delta}\lb \E|\hat{w}(X_i)|^{1+\delta} + |\E[\hat{w}(X_i)]|^{1+\delta}\rb\le 2^{1 + \delta}\E|\hat{w}(X_i)|^{1+\delta}.\]
\item If $\delta < 1$, then by Markov's inequality,
  \begin{align*}
    &\p\lb\sum_{i=1}^{n}\hat{w}(X_i)^2\ge nt\rb\le \frac{\E\left[\lb\sum_{i=1}^{n}\hat{w}(X_i)^2\rb^{(1 + \delta) / 2}\right]}{(nt)^{(1 + \delta) / 2}}
  \le \frac{\E\left[\sum_{i=1}^{n}\hat{w}(X_i)^{1 + \delta}\right]}{(nt)^{(1 + \delta) / 2}}\le \frac{n^{(1 - \delta) / 2 }M^{1 + \delta}}{t^{(1 + \delta)/2}},
  \end{align*}
  where the second last step follows from the simple fact that $\|x\|_{p}\le \|x\|_{1}$ for $p\ge 1$, with $p = 2 / (1 + \delta)$ and $x_i = \hat{w}(X_i)^{1 + \delta}$. By Markov's inequality and Proposition \ref{prop:vonbahresseen},
  \begin{multline}
    \p\lb \sum_{i=1}^{n}|\hat{w}(X_i) - 1|\ge \frac{n}{2}\rb \\
    \le \frac{2^{1 + \delta}}{n^{1 + \delta}}\E\bigg|\sum_{i=1}^{n}\hat{w}(X_i) - \E[\hat{w}(X_i)]\bigg|^{1 + \delta} \le \frac{2^{2 + \delta}n\E|\hat{w}(X_i) - \E[\hat{w}(X_i)]|^{1 + \delta}}{n^{1 + \delta}} \le \frac{2^{3+2\delta}M^{1+\delta}}{n^{\delta}}.\label{eq:denom_case2}
  \end{multline}
\end{enumerate}
The proof is then completed by setting $A = 2^{3+2\delta}L(\delta)\max\{M^2, M^{1+\delta}\}$.

\subsection{Asymptotic theory for weighted conformal inference}
As a direct consequence of Theorem \ref{thm:double_robustness_w}, Theorem \ref{thm:double_robustness_q}, and Theorem \ref{thm:prediction_intervals}, we can derive the following asymptotic results.
\begin{theorem}\label{thm:double_robustness}
  Let
  $(X_i,Y_i) \stackrel{i.i.d.}{\sim} (X,Y) \sim P_X \times P_{Y\mid
    X}$ and let $Q_X$ be another distribution on the domain of
  $X$. Set $N = |\calZ_\tr|$ and $n = |\calZ_\ca|$. Further, let
  $\{\calF_s(x; \calD): s\in \calS\}$ be any sequence of nested sets,
  $\hat{w}(x) = \hat{w}(x;\calZ_\tr)$ be an estimate of
  $w(x) = (dQ_X/dP_X)(x)$, and $\hat{C}(x)$ be the resulting conformal
  interval from Algorithm \ref{algo:weighted_split_general}. Assume that
  $\E\left[\hat{w}(X)\mid \calZ_\tr\right] = 1$ and $\E[w(X)] =
  1$. Assume that either B1 or B2 (or both) holds:
\begin{enumerate}[label = B\arabic*]
  \item $\underset{N\rightarrow \infty}{\lim} \E\Big[\big|\hat{w}(X) - w(X) \big|\Big] = 0$; 
  \item The assumptions C1 and C2 in Theorem \ref{thm:double_robustness_q} hold.
  \end{enumerate}
  Then 
  \begin{align}
    \label{eq:double_robust_marginal}
    \lim_{N,n\rightarrow\infty} \p_{(X,Y)\sim Q_X \times P_{Y \mid X}} \left( Y\in
    \hat{C}(X)\right) \ge 1-\alpha.
  \end{align}
  Furthermore, under B2, for any $\eps > 0$,
  \begin{align}\label{eq:conditional_coverage}
    \lim_{N,n\rightarrow \infty} \p_{X\sim Q_X}\left(\p\left(Y\in \hat{C}(X) \mid X\right)
    \le 1-\alpha-\eps \right) = 0.
  \end{align}
\end{theorem}

\begin{theorem}\label{thm:asym_prediction_intervals}
  In the settings of Theorem \ref{thm:double_robustness}, assume further that the assumptions C'1 and C'2 in Theorem \ref{thm:prediction_intervals} hold. Then for any $\eps > 0$,
  \begin{align}\label{eq:optimality}
    \lim_{N,n\rightarrow \infty} \p_{X\sim Q_X}\left(\hat{C}(X)\subset \calO_{s_0 + \eps}(X)\right) = 1.
  \end{align} 
  
\end{theorem}
}

 \section{Double robustness of conformalized survival analysis: asymptotic results}\label{app:double_robustness_asym}
\subsection{Proof of Theorem \ref{thm:double_robustness_CQR}}
\revise{
  Let $\xi = \min\{r, \eps / b_2\}$. Since $\hat{L}(x) \le c_0$ for any $x$, 
\begin{align}
  \p(T\wedge c_0 \ge \hat{L}(x) \mid X = x)&\ge \p(T\wedge c_0 \ge c_0 \mid X = x)\\
    & \ge \left\{
    \begin{array}{ll}
      1 - \alpha & \text{if }q_{\alpha}(x)\ge c_0\\
      1 - \alpha - b_2 \xi &  \text{if }q_{\alpha}(x)\in [c_0 - \xi, c_0]
    \end{array}
  \right.,
\end{align}
where the first case is proved by the definition of $q_{\alpha}(x)$ and second case is proved by the condition A2 (i). Thus, whenever $q_{\alpha}(x)\ge c_0 - \xi$, 
\[\p(T\wedge c_0 \ge \hat{L}(x) \mid X = x)\ge 1 - \alpha - \eps.\]
As a result,
\begin{align}
  &\p_{X\sim Q_{X}}\lb \p(T\wedge c_0 \ge \hat{L}(X) \mid X = x) < 1 - \alpha - \eps\rb\\
  & = \p_{X\sim Q_{X}}\lb \p(T\wedge c_0 \ge \hat{L}(X) \mid X = x) < 1 - \alpha - \eps, q_{\alpha}(X) < c_0 - \xi\rb\\
  & \le \p_{X\sim Q_{X}}\lb \p(T\wedge c_0 \ge \hat{L}(X) \mid X = x) < 1 - \alpha - \eps
  ~\big|~ q_{\alpha}(X) < c_0 - \xi\rb,
\end{align}
and
\begin{align}
  \p_{X\sim Q_{X}}\lb T\wedge c_0 \ge \hat{L}(X)\rb &\ge (1 - \alpha - \eps)\p_{X\sim Q_{X}}(q_{\alpha}(X) \ge c_0 - \xi)\\
  &\quad + \p_{X\sim Q_{X}}\lb T\wedge c_0 \ge \hat{L}(X)\mid q_{\alpha}(X) < c_0 - \xi\rb \p(q_{\alpha}(X) < c_0 - \xi).
\end{align}
Since $\eps$ is arbitrary, it remains to prove
\[\lim_{N, n\rightarrow \infty}\p_{X\sim Q_{X}}\lb \p(T\wedge c_0 \ge \hat{L}(X) \mid X = x) < 1 - \alpha - \eps\mid q_{\alpha}(X) < c_0 - \xi\rb = 0,\]
and
\[\lim_{N, n\rightarrow \infty}\p_{X\sim Q_{X}}\lb T\wedge c_0 \ge \hat{L}(X)\mid q_{\alpha}(X) < c_0 - \xi\rb = 1 - \alpha.\]
The results to be proved are equivalent to Theorem \ref{thm:double_robustness_CQR} with $r$ replaced by $\xi$ and an additional assumption that
\begin{equation}
  \label{eq:break_tie_CQR}
  q_{\alpha}(X) < c_0 - r \text{ almost surely under }P_X \text{ and }Q_X.
\end{equation}
The latter is due to that $w(X) < \infty$ almost surely under $P_X$. Throughout the rest of the proof, we will assume \eqref{eq:break_tie_CQR}.
}

Recall that $Y_i = T_i \wedge c_0$, $w(x) = 1 / \c(x)$ and $\hat{w}(x) = 1 / \hat{\c}(x)$. First we show that Assumption A1 of Theorem \ref{thm:double_robustness_CQR} implies Assumption B1 of Theorem \ref{thm:double_robustness}. Since $\E[1 / \hat{\c}(X)\mid \calZ_{\tr}] < \infty$ almost surely and
$\E[1 / \c(X)] < \infty$, we set
$$\hat{w}_{N}(x) = \frac{1 / \hat{\c}(x)}{\E[1 / \hat{\c}(X)\mid
  \calZ_{\tr}]}, \,\, w(x) = \frac{dP_{X}(x)}{dP_{X\mid T=1}(x)} = \frac{1 / \c(x)}{\E[1 / \c(X)]}$$
and observe 
$$\E[\hat{w}_{N}(X)\mid \calZ_{\tr}] = 1 = \E[w(X)].$$ 
Thus, Assumption B1 reduces to
\[\lim_{N\rightarrow \infty}\E\bigg|\frac{1 / \hat{\c}(X)}{\E[1 / \hat{\c}(X)\mid \calZ_{\tr}]} - \frac{1 / \c(X)}{\E[1 / \c(X)]}\bigg| = 0.\]
In fact,
\begin{align*}
  &\lim_{N\rightarrow \infty}\E\bigg|\frac{1 / \hat{\c}(X)}{\E[1 / \hat{\c}(X)\mid \calZ_{\tr}]} - \frac{1 / \c(X)}{\E[1 / \c(X)]}\bigg| \\
  &\le \limsup_{N\rightarrow \infty}\frac{1}{\E[1 / \hat{\c}(X)\mid \calZ_{\tr}]}\E\bigg|\frac{1}{\hat{\c}(X)} - \frac{1}{\c(X)}\bigg| + \limsup_{N\rightarrow \infty}\E \left[\frac{1}{\c(X)}\right]\E\bigg|\frac{1}{\E[1 / \hat{\c}(X)\mid \calZ_{\tr}]} - \frac{1}{\E[1 / \c(X)]}\bigg|\\
  &\stackrel{(1)}{\le} \limsup_{N\rightarrow \infty}\E\bigg|\frac{1}{\hat{\c}(X)} - \frac{1}{\c(X)}\bigg| + \limsup_{N\rightarrow \infty}\E \left[\frac{1}{\c(X)}\right]\E \bigg|\E\left[\frac{1}{\hat{\c}(X)}\mid \calZ_{\tr}\right] - \E\left[\frac{1}{\c(X)}\right]\bigg|\\
  & \le \limsup_{N\rightarrow \infty}\E\bigg|\frac{1}{\hat{\c}(X)} - \frac{1}{\c(X)}\bigg| + \limsup_{N\rightarrow \infty}\E \left[\frac{1}{\c(X)}\right]\E\left( \E\left[\bigg|\frac{1}{\hat{\c}(X)} - \frac{1}{\c(X)}\bigg|\mid \calZ_{\tr}\right]\right)\\
  &\stackrel{(2)}{=} \left( 1 + \E \left[\frac{1}{\c(X)}\right]\right)\limsup_{N\rightarrow \infty}\E\bigg|\frac{1}{\hat{\c}(X)} - \frac{1}{\c(X)}\bigg| = 0,
\end{align*}
where step (1) uses the fact that $\c(x), \hat{\c}(x)\in [0, 1]$, and step (2) uses Assumption A1 and the condition $\E [1 / \c(X)] < \infty$.

Next we show that Assumption A2 implies Assumption B2. Remark that
Assumption C1 is satisfied since
$\E[1 / \hat{c}(X)^{1 + \delta}] < \infty$). It remains to prove that
$\text{A2} \Longrightarrow \text{C2}$. Let
\[\calO_{s}(x) = [q_{\alpha}(x; c_0) - s, \infty).\]
Clearly, Assumption A2 (i) implies Assumption C2 (i) with $s_0 = 0$. Moreover,
\begin{align}
  \Delta_{s}(x) & = \inf\left\{\Delta: \hat{q}_{\alpha}(x; c_0) - s + \Delta \ge q_{\alpha}(x; c_0) - s \text{ and } q_{\alpha}(x; c_0) - s + \Delta \ge \hat{q}_{\alpha}(x; c_0) - s\right\}\\
                &= \inf\left\{\Delta: \hat{q}_{\alpha}(x; c_0) + \Delta \ge q_{\alpha}(x; c_0) \text{ and } q_{\alpha}(x; c_0) + \Delta \ge \hat{q}_{\alpha}(x; c_0)\right\}\\
                &= |\hat{q}_{\alpha}(x; c_0) - q_{\alpha}(x; c_0)| = \calE(x).
\end{align}
Thus, $\Delta(x) = \sup_{s\in [s_0 - r, s_0]}\Delta_s(x) \le \calE(x)$ and Assumption C2 (ii) holds.

\subsection{Proof of Theorem \ref{thm:double_robustness_CDR}}
\revise{
As in the proof of Theorem \ref{thm:double_robustness_CQR}, we first show that it remains to prove the result when $q_{\alpha + r}(x; c_0) < c_0$ almost surely. Let $\xi = \min\{r, \eps\}$. Since $\hat{L}(x) \le c_0$ for any $x$,
\[\p(T\wedge c_0 \ge \hat{L}(x) \mid X = x)\ge \p(T\wedge c_0 \ge c_0 \mid X = x) = \p(T \ge c_0 \mid X = x).\]
Then the condition (i) implies 
\[\p(T\ge c_0 \mid X = x)\ge \p(T \ge q_{\alpha + \eps}(x) \mid X = x) = 1 - \alpha - \xi, \quad \text{if }q_{\alpha+\xi}(x) \ge c_0.\]
As a result,
\begin{align}
  &\p_{X\sim Q_{X}}\lb \p(T\wedge c_0 \ge \hat{L}(X) \mid X = x) < 1 - \alpha - \eps\rb\\
  & = \p_{X\sim Q_{X}}\lb \p(T\wedge c_0 \ge \hat{L}(X) \mid X = x) < 1 - \alpha - \eps, q_{\alpha+\xi}(X) < c_0\rb\\
  & \le \p_{X\sim Q_{X}}\lb \p(T\wedge c_0 \ge \hat{L}(X) \mid X = x) < 1 - \alpha - \eps\mid q_{\alpha+\xi}(X) < c_0\rb,
\end{align}
and
\begin{align}
  \p_{X\sim Q_{X}}\lb T\wedge c_0 \ge \hat{L}(X)\rb &\ge (1 - \alpha - \eps)\p_{X\sim Q_{X}}(q_{\alpha+\xi}(X) \ge c_0)\\
  &\quad + \p_{X\sim Q_{X}}\lb T\wedge c_0 \ge \hat{L}(X)\mid q_{\alpha+\eps}(X) \ge c_0\rb \p(q_{\alpha+\xi}(X) < c_0).
\end{align}
Since $\eps$ is arbitrary, it remains to prove
\[\lim_{N, n\rightarrow \infty}\p_{X\sim Q_{X}}\lb \p(T\wedge c_0 \ge \hat{L}(X) \mid X = x) < 1 - \alpha - \eps\mid q_{\alpha + \xi}(X) < c_0\rb = 0,\]
and
\[\lim_{N, n\rightarrow \infty}\p_{X\sim Q_{X}}\lb T\wedge c_0 \ge \hat{L}(X)\mid q_{\alpha + \xi}(X) < c_0\rb = 1 - \alpha.\]
The results to be proved are equivalent to Theorem \ref{thm:double_robustness_CDR} with $r$ replaced by $\xi$ and an additional assumption that
\begin{equation}
  \label{eq:break_tie_CDR}
  q_{\alpha + r}(X) < c_0 \text{ almost surely under }P_X \text{ and }Q_X.
\end{equation}
The latter is due to that $w(X) < \infty$ almost surely under $P_X$. Throughout the rest of the proof, we will assume \eqref{eq:break_tie_CDR}.
}

Using the same argument as in the proof of Theorem
\ref{thm:double_robustness_CQR}, it suffices to show that A2 from
Theorem \ref{thm:double_robustness_CDR} implies Assumption C2. Let
\[\calO_{s}(x) = [q_{\alpha - s}(x; c_0), \infty),\]
Clearly, Assumption A2 (i) implies Assumption C2 (i) with $s_0 = 0$.

To compute $\Delta(x)$, we first replace $r$ by $2r$. Assumption A2 (i) with $\eps = 0$ implies that $F(q_{\alpha - s}(x; c_0)\mid X = x) = 1 - \alpha + s$ for any $s \in [s_0 - 2r, s_0] = [ - 2r, 0]$. % Recalling that $\Delta(
% x) = \sup_{y\in \R}|\hat{F}(y\mid X = x) - F(y\mid X = x)|$ in Theorem \ref{thm:double_robustness_LPB}, 
Then for any $s \in [- r, 0]$, 
\begin{align}
  \Delta_{s}(x)
  &= \inf\left\{\Delta: q_{\alpha - s + \Delta}(x; c_0)\ge \hat{q}_{\alpha - s}(x; c_0)\text{ and } \hat{q}_{\alpha - s + \Delta}(x; c_0) \ge q_{\alpha - s}(x; c_0)\right\}.
\end{align}
Let $\event$ denote the event that $\calE(X)\le r$. Then on $\event$, $\alpha - s + \calE(X) \in [\alpha, \alpha + 2r]$, and by A2 (i), 
\begin{align}
  &  \alpha - s + \calE(X)= F(q_{\alpha - s}(X; c_0)\mid X) + \calE(X)\text{ and } F(q_{\alpha - s + \calE(x)}(X; c_0)\mid X)  = \alpha - s + \calE(X)\\
  & \Longrightarrow \alpha - s + \calE(X)\ge F(\hat{q}_{\alpha - s}(X; c_0)\mid X)\text{ and } F(\hat{q}_{\alpha - s + \calE(X)}(X; c_0)\mid X) \ge \alpha - s\\
  & \Longrightarrow q_{\alpha - s + \calE(X)}(X; c_0)\ge \hat{q}_{\alpha - s}(X; c_0)\text{ and } \hat{q}_{\alpha - s + \calE(X)}(X; c_0) \ge q_{\alpha - s}(X; c_0)\\
  & \Longrightarrow \Delta_{s}(X)\le \calE(X).
\end{align}
Thus on $\event$,
\[\Delta(X) = \sup_{s\in [s_0 - r, s_0]}\Delta_{s}(X)\le \calE(X).\]
On the other hand, $\Delta(X)\le 1$ almost surely. Since $c(X)\le 1$ almost surely, A2 (ii) implies that $\E [\Delta(X)]\rightarrow 0$. By Markov's inequality,
\[\p(\event) \rightarrow 1.\]
As a result,
\begin{align}
  \lim_{N\rightarrow \infty}\E[\Delta(X) / \hat{c}(X)] 
  & \le \limsup_{N\rightarrow \infty}\E[\Delta(X)I_{\event} / \hat{c}(X)] + \E[\Delta(X)I_{\event^c} / \hat{c}(X)]\\
  & \le \limsup_{N\rightarrow \infty}\E [\calE(X) / \hat{c}(X)] + \alpha \E[I_{\event^c} / \hat{c}(X)]\\
  & \le \limsup_{N\rightarrow \infty}\E [\calE(X) / \hat{c}(X)] + \alpha \p(\event^c)^{1 + 1 / \delta} \E [1 / \hat{c}(X)^{1 + \delta}] = 0,
\end{align}
where the second last step follows from H\"{o}lder's inequality. Similarly, we have
\[\lim_{N\rightarrow \infty}\E[\Delta(X) / \c(X)] = 0.\]
Since $\E[1 / \hat{\c}(X)], \E[1 / \c(X)]\ge 1$, we conclude that
\[\lim_{N\rightarrow \infty}\E[\hat{w}(X)\Delta(X)] = \lim_{N\rightarrow \infty}\E[w(X)\Delta(X)] = 0.\]
In conclusion, A2 implies C2 (ii). This completes the proof of Theorem \ref{thm:double_robustness_CDR}.

\revise{
  \subsection{Adaptivity of conformalized survival analysis: asymptotic results}\label{app:adaptivity}
  Following the steps of Theorem \ref{thm:double_robustness_CQR} and Theorem \ref{thm:double_robustness_CDR}, we can show that
  \[\Delta'_{s}(X) \le \calE(X).\]
  By Theorem \ref{thm:prediction_intervals} with $k = \ell = 1$, for CQR,
  \[\lim_{N, n\rightarrow \infty}\p_{X\sim Q_X}\lb [\hat{L}(X), \infty) \subset [q_{\alpha}(x; c_0) - \eps, \infty)\rb = 1,\]
  and for CDR, 
  \[\lim_{N, n\rightarrow \infty}\p_{X\sim Q_X}\lb [\hat{L}(X), \infty) \subset [q_{\alpha - \eps}(x; c_0), \infty)\rb = 1.\]
  Theorem \ref{thm:adaptivity} is then proved.

  Next, we discuss the more involved case where $c_0$ is growing with $n$. Intuitively, as $c_0$ approaches the upper endpoint of the domain of $T$, both CQR and CDR lower predictive bound should approach the oracle $\alpha$-th conditional quantile of $T$. Nevertheless, we cannot invoke Theorem \ref{thm:double_robustness} directly because $1/\p(C\ge c_0 \mid X)$ would diverge as $c_0$ grows. However, the nonasymptotic result provides sufficient conditions under which $\hat{L}(X)$ converges to $q_{\alpha}(X)$ with high probability under $Q_{X}$. We say a sequence $a_{n} = O\lb\mathrm{polyLog}(n)\rb$ iff there exists $\gamma > 0$ such that $a_{n} = O((\log n)^\gamma)$.
  \begin{theorem}\label{thm:adaptivity_growing_c0}
    Assume that $c_0 = c_{0, n}$ grows with $n$ with
    \begin{equation}
      \label{eq:c0n_limit}
      \p_{X\sim Q_{X}}\lb\lim_{n\rightarrow \infty}c_{0, n}\ge q_{\alpha}(X)\rb = 1,
    \end{equation}
    and, for some $\delta > 0$,
    \begin{equation}
      \label{eq:polylogn}
      \E\left[\frac{1}{\hat{\p}(C\ge c_{0, n}\mid X)^{1 + \delta}} + \frac{1}{\p(C\ge c_{0, n}\mid X)^{1 + \delta}}\right] = O\lb\mathrm{polyLog}(n)\rb.
    \end{equation}
    \begin{enumerate}
    \item For CQR, assume further that
      \begin{enumerate}
      \item there exists $b_1, r > 0$ such that, for any $\eps\in [0, r]$,
        \[\p\lb T \ge q_{\alpha}(X) - \eps\rb\ge 1 - \alpha + b_1 \eps.\]
      \item $\hat{q}_{\alpha}(x; c_0) = \hat{q}_{\alpha}(x) \wedge c_0$ for some estimate $\hat{q}_{\alpha}(x)$ of the $\alpha$-th conditional quantile of $T$.
      \item For some $\gamma, \omega > 0$, 
      \[\E\left[\calE^{*\gamma}(X)\right] = O\lb\frac{1}{n^{\omega}}\rb, \quad \text{where }\calE^{*}(x) = |\hat{q}_{\alpha}(x) - q_{\alpha}(x)|.\]
      \end{enumerate}
      Then for any $\eps > 0$,
      \[\lim_{N, n\rightarrow \infty}\p_{X\sim Q_X}(\hat{L}(X) \ge q_{\alpha}(X) - \eps) = 1.\]      
    \item For CDR, assume further that
      \begin{enumerate}
      \item there exists $r > 0$ such that, for any $\eps\in [0, r]$,
        \[\p\lb T \ge q_{\alpha - \eps}(X)\rb\ge 1 - \alpha + \eps.\]
      \item $\hat{q}_{s}(x; c_0) = \hat{q}_{s}(x) \wedge c_0$ for $s\in [\alpha - r, \alpha]$ and some estimate $\hat{q}_{s}(x)$ of the $s$-th conditional quantile of $T$.
      \item For some $\gamma, \omega > 0$, 
      \[\E\left[\calE^{*\gamma}(X)\right] = O\lb\frac{1}{n^{\omega}}\rb, \quad \text{where }\calE^{*}(x) = \sup_{s\in [\alpha - r, \alpha]}|F(\hat{q}_{s}(x) \mid x) - F(q_{s}(x)\mid x)|.\]
      \end{enumerate}
      Then for any $\eps > 0$,
      \[\lim_{N, n\rightarrow \infty}\p_{X\sim Q_X}(\hat{L}(X) \ge q_{\alpha - \eps}(X)) = 1.\]      
    \end{enumerate}
  \end{theorem}
  \begin{remark}
    The condition \eqref{eq:c0n_limit} implies that essential supremum of $C$ must be at least as large as the essential supremum of $q_{\alpha}(X)$. Apparently, it is a necessary condition since otherwise $q_{\alpha}(x)$ would be strictly larger than $\hat{L}(x)$ for a nonnegligible fraction of units. 
  \end{remark}
  \begin{remark}
    The moment condition on $\calE^{*}(X)$ is satisfied for most parametric, nonparametric, and sparse models under standard regularity conditions.
  \end{remark}
  \begin{proof}
    We shall verify the assumptions of Theorem \ref{thm:prediction_intervals}. First, the condition \eqref{eq:polylogn} implies the condition C'1 with 
    \begin{equation}
      \label{eq:M_polylogn}
      \delta_0 = \delta, \quad M = O(\mathrm{polyLog}(n)).
    \end{equation}
    Next, we note that, for any $x$, 
    \[q_{\alpha}(x; c_0) = q_{\alpha}(x) \wedge c_0.\]
    Since $z\mapsto z\wedge c_0$ is a contraction and $F(\cdot \mid x)$ is non-decreasing, it is easy to show
    \[\calE(x)\le \calE^{*}(x),\]
    for both CQR and CDR where $\calE(x)$ is defined in Theorem \ref{thm:double_robustness_CQR} for CQR and in Theorem \ref{thm:double_robustness_CDR} for CDR. Using the same arguments as in the proofs of Theorem \ref{thm:double_robustness_CQR} and Theorem \ref{thm:double_robustness_CDR}, we can show that
    \[\Delta'(x)\le \calE(x).\]
    Let $k = \ell = \gamma\delta / (1 + \delta)$. By H\"{o}lder's inequality,
    \[\E[\hat{w}(X)\Delta^{'k}(X)]\le \E[\hat{w}(X)\calE^{*k}(X)]\le\lb\E[\hat{w}(X)^{1 + \delta}]\rb^{1 / (1 + \delta)}\lb\E[\calE^{*\gamma}(X)]\rb^{\delta / (1 + \delta)} = O\lb\frac{\mathrm{polyLog(n)}}{n^{\omega \delta / (1 + \delta)}}\rb.\]
    Similarly,
    \[\E[w(X)\Delta^{'\ell}(X)] = O\lb\frac{\mathrm{polyLog(n)}}{n^{\omega \delta / (1 + \delta)}}\rb.\]
    Finally, we verify the condition C'2 (i) for CQR and CDR separately.
    \begin{itemize}
    \item For CQR, let
      \[\calO_{s}(x) = [q_{\alpha}(x)\wedge c_0 - s, \infty).\]
      For any $s > 0$,
      \[T \ge q_{\alpha}(x) - s\Longrightarrow T\wedge c_0 \ge q_{\alpha}(x) \wedge c_0 - s.\]
    Thus,
    \[\p\lb T\wedge c_0 \in \calO_{s}(x)\rb\ge \p\lb T \ge q_{\alpha}(X) - s\rb\ge 1 - \alpha + b_1 s.\]
  \item For CDR, let
    \[\calO_{s}(x) = [q_{\alpha - s}(x)\wedge c_0, \infty).\]
    For any $s > 0$,
    \[T \ge q_{\alpha - s}(x)\Longrightarrow T\wedge c_0 \ge q_{\alpha - s}(x) \wedge c_0.\]
    Thus,
    \[\p\lb T\wedge c_0 \in \calO_{s}(x)\rb\ge \p\lb T \ge q_{\alpha}(X) - s\rb\ge 1 - \alpha + s.\]
  \end{itemize}
  Thus, we have shown that all assumptions of Theorem \ref{thm:prediction_intervals} are satisfied. Let $\beta = 1 / \log n$. Since $\const_3$ depends on $M$ polynomially and other constants, \eqref{eq:M_polylogn} implies that
  \[\const_3 = O(\mathrm{polyLog(n)}).\]
  By Theorem \ref{thm:prediction_intervals},
  \[\p_{X \sim Q_{X}}\lb \hat{C}(X)\subset \mathcal{O}_{s_0 + \eps_{n}}(X)\rb = 1 - o(1),\]
  where
  \[\eps_{n} = O\lb\frac{\mathrm{polyLog}(n)}{n^{(\omega \wedge 1)\delta' / (1 + \delta')}}\rb = o(1).\]
  Since $\lim_{n\rightarrow \infty}c_{0, n}\ge q_{\alpha}(X)$ almost surely under $Q_{X}$, for any $s\ge 0$,
  \[\p_{X\sim Q_{X}}\lb\lim_{n\rightarrow \infty}(q_{\alpha}(X) - s)\wedge c_{0, n}\ge (q_{\alpha}(X) - s)\rb = 1,\]
  and
  \[\p_{X\sim Q_{X}}\lb\lim_{n\rightarrow \infty}q_{\alpha - s}(X)\wedge c_{0, n}\ge q_{\alpha - s}(X)\rb = 1.\]
  The proof is then completed by plugging in $\calO_{s}(x)$ for CQR and CDR, respectively.
  \end{proof}
}

\subsection{When is CMR-LPB doubly robust?}\label{subapp:CMR}

For CMR, a natural oracle nested set is given by
\[\calO_{s}(x) = [m(x; c_0) - s, \infty), \quad \text{where }m(x; c_0) = \E[T\wedge c_0 \mid X = x].\]
However, Assumption B2 (b) (i) with $\eps = 0$ requires the existence of $s_0$ such that
\[\p(T\wedge c_0\in \calO_{s_0}(X)\mid X) = 1 - \alpha, \quad \text{almost surely}.\]
This implies that for some $s_0$,
\[\p(T\wedge c_0\ge m(X; c_0) - s_0\mid X) = 1 - \alpha, \quad \text{almost surely}.\]
The above equality does not hold in general. One exception
is the additive case with homoscedastic errors:
\[T\wedge c_0 = m(X; c_0) + \nu, \quad \mathrm{Var}[\nu \mid X] =
  \mathrm{Var}[\nu].\] In this case, we can derive the double
robustness of the CMR-LPB based on Theorem
\ref{thm:double_robustness}. % However, the above model is rarely
% plausible due to the censoring time $c_0$. 
% \ejc{I don't understand
%   this sentence.}~\llzr{This is indeed a confusing argument. We deleted it. }

\revise{
\section{Additional results}
\subsection{Additional details for the selection of $c_0$}
\label{appx:c_up}
We consider two cases. 
\begin{itemize}
\item When there is prior information on the consoring 
mechanism:
for example, if the researcher
knows that there exists $\bar{c}$ such that 
$\p(C \ge \bar{c} ~|~ X) \ge \eta$ almost
surely, for some $\eta > 0$. With any consistent
estimator $\tilde{c}(\cdot)$ for $c(\cdot)$, 
we can again consider the truncated estimator
$\hat{c}(x) = \tilde{c}(x) \vee \frac{\eta}{2}$
and for any $\varepsilon >0$,
\begin{align*}
& \E\bigg[\Big|\frac{1}{\hat{c}(X)} - \frac{1}{c(X)}\Big|\bigg]\\
= & \E\bigg[\Big|\frac{1}{\hat{c}(X)} - \frac{1}{c(X)}\Big| 
\cdot \Indc\Big\{\big|\hat{c}(X) - c(X)\big| \ge \varepsilon\Big\}\bigg]
+ \E\bigg[\Big|\frac{1}{\hat{c}(X)} - \frac{1}{c(X)}\Big| 
\cdot \Indc\Big\{\big|\hat{c}(X) - c(X)\big| < \varepsilon\Big\}\bigg]\\
\le & \frac{4}{\eta}\cdot \p\Big(\big|\hat{c}(X) - c(X) \big| \ge \varepsilon\Big)
+ \frac{4\varepsilon}{\eta^2}.
\end{align*}
The above implies that Assumpation \textbf{A}1 is satisfied.
\item 
When there is no available 
prior information on such $\bar c$, researchers
can instead use the data to determine $\bar c$.
Suppose the covariate $X$ takes value in a finite
set $\{x_1,x_2,\ldots,x_m\}$, where $\p(X = x_l) >0$
for any $l\in[m]$. Let $\Delta = \min_{l\in[m]} \p(X =x_l)$.
Further divide 
$\calZ_{\tr}$ into two disjoint $\calZ_{\tr}^1$
and $\calZ_{\tr}^2$, and the upper bound on $c_0$ can
be obtained via:
\begin{align*}
  \bar{c} = \min_{l\in[m]} \bar{c}_l:=\sup 
\left\{c \in \R: \frac{\sum_{i \in \calZ_{\tr}^1} \Indc\{C_i \ge c, X_i = x_l\}}
{\sum_{i \in \calZ_{\tr}^1} \Indc\{X_i = x_l\}} \ge 2\eta \right\}.
\end{align*}
Note that $\bar{c}\ge 0$ for any $\eta < 0.5$.
We then claim that the event $A := \{\p(C \ge \bar{c} ~|~ X = x_l, \calZ^1_{\tr}) \ge \eta,
\forall l \in [m]\}$ occurs with probability at least 
$1 - m\exp(-\eta^2|\calZ_{\tr}^1|\Delta) - 
m \exp(-\Delta^2|\calZ_{\tr}^1|/2)$.

To prove this claim, for simplicity, we assume the censoring time has a continuous
distribution; the results can be applied to the general case
with a silimar but slightly more complicated argument.
Fix $\eta >0$. For any $l\in [m]$, define 
\begin{align*}
  c_l^* = \sup \big\{c\in\R: \p(C \ge c ~|~ X = x_l) \ge \eta\big\}.
\end{align*}
By the continuity of $C$, we have $\p(C \ge c_l^* ~|~ X = x_l) = \eta$
for $\forall l\in[m]$. Next, 
\begin{align*}
& \p\big(\exists l\in[m], \p(C \ge \bar{c} ~|~ X = x_l) < \eta \big)
\le \sum_{l\in[m]} \p\big( \p(C \ge \bar{c} ~|~ X = x_l) <\eta \big)\\
\le & \sum_{l\in[m]} \p(\bar{c} > c_l^*)
\le \sum_{l \in [m]} \p(\bar{c}_l > c_l^*)
\le \sum_{l \in [m]} \p\lb\frac{\sum_{i\in \calZ_{\tr}^1} \Indc\{X_i = x_l, 
C_i \ge c_l^*\}}{\sum_{i\in\calZ_{\tr}^1} \Indc\{X_i = x_l\}} \ge 2\eta \rb\\
\le & \sum_{l \in [m]} \E\left[\p\lb\frac{\sum_{i\in \calZ_{\tr}^1} \Indc\{X_i = x_l, 
C_i \ge c_l^*\}}{\sum_{i\in\calZ_{\tr}^1} \Indc\{X_i = x_l\}} 
- \p(C \ge c_l^* ~|~ X = x_l) \ge \eta ~\Big|~ \{X_{i}:i\in\calZ_{\tr}^1\}
 \rb\right]\\
        \stackrel{\rm (1)}{\le} & m\cdot \E\Big[\exp\big(-2 \eta^2 \cdot |\calZ_{tr}^{1,l}|\big)\Big],
\end{align*}
where $|\calZ_{\tr}^{1,l}| = \sum_{i\in \calZ_{\tr}^1} \Indc\{X_i = x_l\}$
and step (1) is due to Hoeffding's inequality. Finally, we have
\begin{align*}
\E\Big[\exp\big(-2\eta^2 \cdot |\calZ_{\tr}^{1,l}|\big)\Big] = &
\E\left[\exp\big(-2\eta^2 \cdot |\calZ_{\tr}^{1,l}|\big)\cdot  \Indc\left
\{|\calZ_{\tr}^{1,l}| \ge \frac{\Delta}{2} |\calZ_{\tr}^1|\right\}\right] \\
& \qquad  + \E\left[\exp\big(-2\eta^2 \cdot |\calZ_{\tr}^{1,l}|\big) 
                                                                               \cdot \Indc\left\{|\calZ_{\tr}^{1,l}| < \frac{\Delta}{2} |\calZ_{\tr}^1| \right\}\right] \\
  \le &  \E\Big[\exp\big(-\eta^2 \Delta |\calZ_{\tr}^1|\big)\Big] + \p\lb |\calZ_{\tr}^{1,l}| < \frac{\Delta}{2} |\calZ_{\tr}^1|\rb\\
\le &  \E\Big[\exp\big(-\eta^2 \Delta |\calZ_{\tr}^1|\big)\Big] + \exp\big(-|\calZ_{\tr}^1| \Delta^2/2\big),
\end{align*}
where the last line applies Hoeffding's inequality again. This completes the proof for the lower bound on $\p(A)$.
Now for any $c_0 \le \bar{c}$,
and any consistent estimator $\tilde{c}(\cdot)$,
the truncated estimator $\hat{c}(x) = \tilde{c}(x)\vee \frac{\eta}{2}$
satisfies that for any $\varepsilon > 0$,
\begin{align*}
& \p\Bigg(\E\bigg[\Big|\frac{1}{\hat{c}(X)} - \frac{1}{c(X)}\Big| ~\Big|~ 
\calZ_{\tr}^1\bigg] \ge \varepsilon \Bigg)\\
= & \p\Bigg(\E\bigg[\Big|\frac{1}{\hat{c}(X)} - \frac{1}{c(X)}\Big| ~\Big|~ 
\calZ_{\tr}^1\bigg] \ge \varepsilon, A \Bigg)
+ \p\Bigg(\E\bigg[\Big|\frac{1}{\hat{c}(X)} - \frac{1}{c(X)}\Big| ~\Big|~ 
\calZ_{\tr}^1\bigg] \ge \varepsilon, A^c \Bigg)\\
\stackrel{\rm (a)}{\le} &  \p\bigg(\p\Big(\big|\hat{c}(X) - c(X)\big| 
\ge \frac{\eta^2 \varepsilon}{8} ~\Big|~ \calZ_{\tr}^1\Big) \ge \frac{\eta^2 \varepsilon}{8},A\bigg)
+ \p(A^c) \\
    \stackrel{\rm (b)}{\le} & \p\bigg(\p\Big(\big|\hat{c}(X) - c(X)\big| 
\ge \frac{\eta^2 \varepsilon}{8} ~\Big|~ \calZ_{\tr}^1\Big) \ge \frac{\eta^2 \varepsilon}{8}\bigg) + 
m\cdot \exp\Big(- \eta^2 |\calZ_{\tr}^1|\Delta\Big) 
+ m \exp\Big(-\frac{\Delta^2|\calZ_{\tr}^1|}{2}\Big).
\end{align*}
Above, step (a) is because on event $A$,
\begin{align*}
  & \E\bigg[\Big|\frac{1}{\hat{c}(X)} - \frac{1}{c(X)}\Big| ~\bigg|~ \calZ_{\tr}\bigg]
  \le  \frac{4}{\eta^2} \cdot \E\Big[\big|\hat{c}(X) - c(X)]\big| ~\Big|~ \calZ_{\tr}\Big]\\
= & \frac{4}{\eta^2} \cdot \bigg(\E\Big[\big|\hat{c}(X) - c(X)\big| 
  \cdot \ind\Big\{\big|\hat{c}(X) - c(X)\big|  \ge \frac{\eta^2 \eps}{8}\Big\} ~\Big|~ \calZ_{\tr}\Big]\\
  &\qquad \qquad+ \E\Big[\big|\hat{c}(X) - c(X)\big| 
\cdot \ind\Big\{\big|\hat{c}(X) - c(X)\big|  < \frac{\eta^2 \eps}{8}\Big\}~\Big|~\calZ_{\tr}\Big]\bigg) \\
    \le & \frac{4}{\eta^2}\cdot \p\Big(\big|\hat{c}(X) - c(X)\big| \ge \frac{\eta^2 \eps}{8} ~\Big|~ \calZ_{\tr}\Big)
    + \frac{\eps}{2};
\end{align*}
step (b) is due to the lower bound on $\p(A)$ we have
proved previously.
The above quantity goes to zero as $N,n \rightarrow \infty$.
Hence we arrive at a condition slightly weaker than Assumption
\textbf{A}1, and would imply a slightly weaker coverage guarantee: for any 
$\varepsilon >0$,
\begin{align*}
  \lim_{N,n\rightarrow \infty}\p\Big(\p\big(Y_{n+1} \in \hat{C}(X_{n+1}) ~|~ \calZ_{\tr}^1 \big) 
  \le 1-\alpha-\varepsilon  \Big) = 0.
\end{align*}
\end{itemize}

\subsection{Selecting $c_0$ based on the calibration set}
\label{sec:uniform_c}
Throughout this section, we only consider the case where $\hat{w}(\cdot) = w(\cdot)$ is known. Without loss of generality, assume $\calI_\ca = \{1, \ldots, n\}$. Let $\hat{c}_0$ be any choice of $c_0$ that potentially depends on both $\Z_\tr$ and $\Z_\ca$. In particular, we choose $\hat{c}_0$ by maximizing the average lower prediction bounds on $\Z_\ca$ over a candidate set $\calC \subset \R^{+}$, i.e.,
\begin{equation}
  \label{eq:hatc0}
  \hat{c}_0 = \argmax{c_0\in \calC}~\frac{1}{n}\sum_{i=1}^{n}\hat{L}_{c_0}(X_i),
\end{equation}
where $\hat{L}_{c_0}(X_i)$ denotes the lower prediction bound with threshold $c_0$. We shall prove that $\hat{L}_{\hat{c}_0}(X)$ is approximately valid under regularity conditions on $\mathcal{C}$ and the conformity score.

\begin{theorem}\label{thm:finite_c0}
  If $|\mathcal{C}| < \infty$ and $\E[w(X_i)^r] < \infty$ for some $r > 2$, then, as $n\rightarrow \infty$,
  \[\p\lb T\ge \hat{L}_{\hat{c}_0}(X)\rb\ge 1 - \alpha - o(1).\]
\end{theorem}

\begin{theorem}\label{thm:infinite_c0}
  Assume $\calC = \R^{+}$ and $\E[w(X_i)^r] < \infty$ for some $r > 4$. If there exists a positive integer $M$ such that, for any $(x_1, \td{t}_1)$ and $(x_2, \td{t}_2)$, both $\{c_0\in \R^{+}: V(x_1, \td{t}_1; c_0) > V(x_2, \td{t}_2; c_0)\}$ and $\{c_0\in \R^{+}: V(x_1, \td{t}_1; c_0) < V(x_2, \td{t}_2; c_0)\}$ are unions of at most $M$ intervals, then, as $n\rightarrow \infty$,
  \[\p\lb T\ge \hat{L}_{\hat{c_0}}(X)\rb\ge 1 - \alpha - o(1).\]
\end{theorem}
\begin{remark}
  For CQR, if we estimate the $(1-\alpha)$-th quantile of $T\wedge c_0$ by $\hat{q}_{\alpha}(x) \wedge c_0$, where $\hat{q}_{\alpha}(x)$ is an estimate of the $(1-\alpha)$-th quantile of $T$ using survival estimation techniques, then
  \[V(x, \td{t}; c_0) = \hat{q}_{\alpha}(x) \wedge c_0 - \td{t}\wedge c_0.\]
  Clearly, $V(x, \td{t}; c_0)$ is a piecewise linear function where the first and the third pieces are constants $0$ and $\hat{q}_{\alpha}(x) - \td{t}$, respectively, connecting by a linear piece with slope $\mathrm{sign}(\hat{q}_{\alpha}(x) - \td{t})$. Thus, the condition of Theorem \ref{thm:infinite_c0} holds with $M = 1$.
\end{remark}

\subsubsection{Proof of Theorem \ref{thm:finite_c0} and \ref{thm:infinite_c0}}
We start with the following lemma that bounds the weighted empirical process associated with the weighted conformal inference procedure indexed by $c_0$ and $v$. 

\begin{lemma}\label{lem:weighted_conformal_conditional}
Let $V(x, y; c_0)$ be any conformity score that depends on $c_0$ and $\Z_\tr$. Further let
  \[A_{c_0, v} = \{(x, y): V(x, y; c_0)\le v\}, \quad \calA(\calC) = \{A_{c_0, v}: c_0 \in \calC, v\in \R\}.\]
  If $\E[w(X_i)^r] < \infty$ for some $r > 2$,
  \begin{align*}
    &\sup_{c_0\in \calC, v\in \R}\bigg|\frac{\sum_{i=1}^{n}w(X_i)I(V(X_i, \td{T}_i; c_0)\le v)}{\sum_{i=1}^{n}w(X_i)} - \E\left[w(X)I(V(X, \td{T}; c_0)\le v)\right]\bigg|= O_\p\lb\frac{N(\calA(\calC))^{1/r}}{\sqrt{n}}\rb
  \end{align*}
  where $N(\calA(\calC))$ denotes the shattering number of $\calA(\calC)$, i.e.,
  \[N(\calA(\calC)) = \sup_{(x_i, \td{t}_i)\in \mathrm{Domain}(X_i, \td{T}_i)}\bigg|\left\{\bigg( I[(x_1, \td{t}_1)\in A_{c_0, v}], \ldots, I[(x_n, \td{t}_n)\in A_{c_0, v}]\bigg): c_0\in \calC, v\in \R\right\}\bigg|.\]
\end{lemma}

The proof is lengthy and deferred to Section \ref{subsubapp:proof_weighted_conditional}. Next, we prove a lemma relating the empirical process bound to the coverage.
% Applying the first We are left to prove that, under the conditions of Theorem \ref{thm:},
% \begin{equation}
%   \label{eq:goal_NAC}
%   N(\calA(\calC))
% \end{equation}

\begin{lemma}\label{lem:coverage_lower_bound}
  Let $\hat{c}_0$ be selected from $\calC$ based solely on $\Z_\tr$ and $\Z_\ca$ (e.g., \eqref{eq:hatc0}). If
  \[\sup_{c_0\in \calC, v\in \R}\bigg|\frac{\sum_{i=1}^{n}w(X_i)I(V(X_i, \td{T}_i; c_0)\le v)}{\sum_{i=1}^{n}w(X_i)} - \E\left[w(X)I(V(X, \td{T}; c_0)\le v)\right]\bigg| = o_\p(1),\]
  then
  \[\p(T\ge \hat{L}_{\hat{c}_0}(X))\ge 1 - \alpha - o(1).\]
\end{lemma}
\begin{proof}
  Let $\eta_{c_0}(x)$ denote the cutoff for the conformity score defined in Algorithm \ref{algo:weighted_split} that corresponds to the threshold $c_0$ and
  \[\td{\eta}_{c_0} = \quantile\left(1 - \alpha; \frac{\sum_{i=1}^{n}w(X_i)\delta_{V(X_i, \td{T}_i; c_0)}}{\sum_{i=1}^{n}w(X_i)} \right).\]
  Clearly, 
  \[\frac{\sum_{i=1}^{n}w(X_i)\delta_{V(X_i, \td{T}_i; c_0)}}{\sum_{i=1}^{n}w(X_i)}\preceq \frac{\sum_{i=1}^{n}w(X_i)\delta_{V(X_i, \td{T}_i; c_0)}}{\sum_{i=1}^{n}w(X_i) + w(x)} + \frac{w(x)\delta_{\infty}}{\sum_{i=1}^{n}w(X_i) + w(x)}.\]
  Thus, for any $x$,
  \begin{equation}
    \label{eq:tdeta}
    \td{\eta}_{c_0} \le \eta_{c_0}(x).
  \end{equation}
  Let
  \[G(c_0, v) = \E\left[w(X)I(V(X, \td{T}; c_0)\le v)\right].\]
  By \eqref{eq:tdeta},
  \begin{align}
    \p\lb T\ge \hat{L}_{\hat{c}_0}(X)\mid \Z_\ca\rb &= \p\lb V(X, T; \hat{c}_0)\le \eta_{\hat{c}_0}(X)\mid \Z_\ca\rb\\
                                                    & \ge \p\lb V(X, T; \hat{c}_0)\le \td{\eta}_{\hat{c}_0}\mid \Z_\ca\rb  = G(\hat{c}_0, \td{\eta}_{\hat{c}_0}).\label{eq:Gtdeta}
  \end{align}
  Noting that $\hat{c}_0\in \calC, \td{\eta}_{\hat{c}_0}\in \R$, the condition implies 
  \begin{align}
    \bigg|\frac{\sum_{i=1}^{n}w(X_i)I(V(X_i, \td{T}_i; \hat{c}_0)\le \td{\eta}_{\hat{c}_0})}{\sum_{i=1}^{n}w(X_i)} - G(\hat{c}_0, \td{\eta}_{\hat{c}_0})\bigg| = o_\p(1).
  \end{align}
  By definition,
  \[\frac{\sum_{i=1}^{n}w(X_i)I(V(X_i, \td{T}_i; \hat{c}_0)\le \td{\eta}_{\hat{c}_0})}{\sum_{i=1}^{n}w(X_i)} \ge 1 - \alpha.\]
  Therefore,
  \[G(\hat{c}_0, \td{\eta}_{\hat{c}_0}) \ge 1 - \alpha - o_\p(1).\]
  Marginalizing over $\Z_\ca$ in \eqref{eq:Gtdeta} and apply the Dominated Convergence Theorem, we conclude
  \[\p\lb T\ge \hat{L}_{\hat{c}_0}(X)\rb\ge \E[G(\hat{c}_0, \td{\eta}_{\hat{c}_0})] \ge 1 - \alpha - o(1).\]
\end{proof}

\begin{proof}[\textbf{Theorem \ref{thm:finite_c0}}]
  For any $c_0 \in \calC$,
  \[\bigg|\left\{\bigg( I[(x_1, \td{t}_1)\in A_{c_0, v}], \ldots, I[(x_n, \td{t}_n)\in A_{c_0, v}]\bigg): v\in \R\right\}\bigg|\le n + 1.\]
  Thus,
  \[N(\calA(\calC)) \le |\calC| \cdot (n + 1) = O(n).\]
  The result is then implied by the first bound in Lemma \ref{lem:weighted_conformal_conditional} and Lemma \ref{lem:coverage_lower_bound}.
\end{proof}

\begin{proof}[\textbf{Theorem \ref{thm:infinite_c0}}]
  For any subset $\calS \subset \{1, \ldots, n\}$, let $e_{\calS}$ denote the binary vector with the $k$-th entry equal to $1$ iff $k\in \calS$. Given $(x_i, \td{t}_{i})_{i=1}^{n}$, let $\calC_{0, jk}$ denote all boundary points of $\{c_0\in \R^{+}: V(x_j, \td{t}_j; c_0) > V(x_k, \td{t}_k; c_0)\}$ and $\calC_0$ denote the union of all $C_{0, jk}$'s, allowing the same value to appear for multiple times. Then
  \[|\calC_0| \le 2Mn^2 = O(n^2).\]
  Let $c_{0, 1} \le c_{0, 2} \le \cdots \le c_{0, |\calC_0|}$ be the elements of $\calC_0$ and
  \[\calN_j = \left\{\bigg( I[(x_1, \td{t}_1)\in A_{c_{0, j}, v}], \ldots, I[(x_n, \td{t}_n)\in A_{c_{0, j}, v}]\bigg): v\in \R\right\}.\]
  Then
  \[\left\{\bigg( I[(x_1, \td{t}_1)\in A_{c_0, v}], \ldots, I[(x_n, \td{t}_n)\in A_{c_0, v}]\bigg): c_0\in \calC, v\in \R\right\} = \bigcup_{j=1}^{|\calC_0|}\calN_j.\]
  By definition, the ordering of the conformity scores $\{V(x_j, \td{t}_j; c_0)\}_{j=1}^{n}$ can only change at $c_0 = c_{0, j}$ for some $j = 1, \ldots, |\calC_0|$. Since we allow for multiplicity in $\calC_0$, at each $c_{0,j}$, only a pair of adjacent conformity scores exchange while all other elements maintain their ranks. Given $j < |\calC_0|$, suppose that
  \[V(x_{i_1}, \td{t}_{i_1}; c) \le V(x_{i_2}, \td{t}_{i_2}; c) \le \cdots \le V(x_{i_n}, \td{t}_{i_n}; c), \quad \text{for any }c \in (c_{0, j}, c_{0, j+1}),\]
  and $V(x_{i_k}, \td{t}_{i_k}; c)$ would exchange with $V(x_{i_{k+1}}, \td{t}_{i_{k+1}}; c)$ at $c = c_{0, (j + 1)}$. Then
  \[\calN_j = \{e_\calS: \calS = \{i_1, \ldots, i_r\}, r = 1, \ldots, n\},\]
  and
  \[\calN_{j+1} = \calN_j \cup \{\{i_1, \ldots, i_{k-1}, i_{k+1}\}\} \setminus \{\{i_1, \ldots, i_{k-1}, i_{k}\}\}.\]
  As a result,
  \[|\calN_{j+1}\setminus \calN_{j}| = 1.\]
  Clearly, $|\calN_1| = n + 1$. Thus
  \[\bigg|\bigcup_{j=1}^{|\calC_0|}\calN_j\bigg|\le |\calN_1| + |\calC_0| - 1 = n + |\calC_0| = O(n^2).\]
  The theorem is then implied by the first bound in Lemma \ref{lem:weighted_conformal_conditional} and Lemma \ref{lem:coverage_lower_bound}.
\end{proof}

\subsubsection{Proof of Lemma \ref{lem:weighted_conformal_conditional}}\label{subsubapp:proof_weighted_conditional}
Let $(X'_1, \td{T}'_1), \ldots, (X'_n, \td{T}'_n)$ be independent copies of $(X_1, \td{T}_1), \ldots, (X_n, \td{T}_n)$ and $(\eta_1, \ldots, \eta_n)$ be i.i.d. Rademacher random variables. Further, let
\[\mathcal{N}((x_i, \td{t}_i)_{i=1}^{n}) = \left\{\bigg( w(x_1)I[(x_1, \td{t}_1)\in A_{c_0, v}], \ldots, w(x_n)I[(x_n, \td{t}_n)\in A_{c_0, v}]\bigg): c_0\in \calC, v\in \R\right\}.\]
Clearly,
\[|\mathcal{N}((x_i, \td{t}_i)_{i=1}^{n})|\le N(\calA(\calC)).\]
  Then for any increasing convex function $\Phi: \R\mapsto \R^{+}$, 
  \begin{align}
    &\E\left[\Phi\lb \sup_{c_0\in \calC, v\in \R}\bigg|\frac{1}{n}\sum_{i=1}^{n}w(X_i)I(V(X_i, \td{T}_i; c_0)\le v) - \E\left[w(X)I(V(X, \td{T}; c_0)\le v)\right]\bigg|\rb\right]\\
    & \le \E\left[\Phi\lb \sup_{c_0\in \calC, v\in \R}\bigg|\frac{1}{n}\sum_{i=1}^{n}w(X_i)I(V(X_i, \td{T}_i; c_0)\le v) - \E\left[\frac{1}{n}\sum_{i=1}^{n}w(X'_i)I(V(X'_i, \td{T}'_i; c_0)\le v)\right]\bigg|\rb\right]\\
    & \stackrel{(1)}{\le}\E\left[\Phi\lb \sup_{c_0\in \calC, v\in \R}\bigg|\frac{1}{n}\sum_{i=1}^{n}w(X_i)I(V(X_i, \td{T}_i; c_0)\le v) - \frac{1}{n}\sum_{i=1}^{n}w(X'_i)I(V(X'_i, \td{T}'_i; c_0)\le v)\bigg|\rb\right]\\
    & = \E\left[\Phi\lb \sup_{c_0\in \calC, v\in \R}\bigg|\frac{1}{n}\sum_{i=1}^{n}\left\{w(X_i)I(V(X_i, \td{T}_i; c_0)\le v) - w(X'_i)I(V(X'_i, \td{T}'_i; c_0)\le v)\right\}\bigg|\rb\right]\\
    & \stackrel{(2)}{=} \E\left[\Phi\lb \sup_{c_0\in \calC, v\in \R}\bigg|\frac{1}{n}\sum_{i=1}^{n}\eta_{i}\left\{w(X_i)I(V(X_i, \td{T}_i; c_0)\le v) - w(X'_i)I(V(X'_i, \td{T}'_i; c_0)\le v)\right\}\bigg|\rb\right]\\
    & \stackrel{(3)}{\le} \E\left[\Phi\lb \sup_{c_0\in \calC, v\in \R}\bigg|\frac{1}{n}\sum_{i=1}^{n}\eta_{i} w(X_i)I(V(X_i, \td{T}_i; c_0)\le v)\bigg| +  \sup_{c_0\in \calC, v\in \R}\bigg| \frac{1}{n}\sum_{i=1}^{n}\eta_i w(X'_i)I(V(X'_i, \td{T}'_i; c_0)\le v)\bigg|\rb\right]\\
    & \stackrel{(4)}{\le} \E\left[\Phi\lb \sup_{c_0\in \calC, v\in \R}\bigg|\frac{2}{n}\sum_{i=1}^{n}\eta_{i} w(X_i)I(V(X_i, \td{T}_i; c_0)\le v)\bigg|\rb\right]\\
    & = \E\left\{\E\left[\Phi\lb \sup_{c_0\in \calC, v\in \R}\bigg|\frac{2}{n}\sum_{i=1}^{n}\eta_{i} w(X_i)I(V(X_i, \td{T}_i; c_0)\le v)\bigg|\rb\right]\mid (X_i, \td{T}_i)_{i=1}^{n}\right\}\\
    & \stackrel{(5)}{\le} \E\left\{\sum_{a\in \mathcal{N}((X_i, \td{T}_i)_{i=1}^{n})} \E\left[\Phi\lb \bigg|\frac{2}{n}\sum_{i=1}^{n}\eta_{i}a_{i}\bigg|\rb\mid (X_i, \td{T}_i)_{i=1}^{n}\right]\right\},\label{eq:EPhi_generic}
  \end{align}
  where (1) applies Jensen's inequality, (2) follows from the distributional symmetry of the summand $w(X_i)I(V(X_i, \td{T}_i; c_0)\le v) - w(X'_i)I(V(X'_i, \td{T}'_i; c_0)\le v)$, (3) is due to the triangle inequality and the monotonicity of $\Phi$, (4) applies Jensen's inequality again together with exchangeability of the two terms, and (5) uses the definition of the set $\mathcal{N}(\cdot)$. For any $a\in \R^{n}$, $\frac{2}{n}\sum_{i=1}^{n}\eta_{i}a_{i}$ is sub-Gaussian with parameter $\frac{4}{n^2}\sum_{i=1}^{n}a_i^2$. Given $(X_i, \td{T}_i)_{i=1}^{n}$, for any $a\in \mathcal{N}((X_i, \td{T}_i)_{i=1}^{n})$,
  \[\frac{4}{n^2}\sum_{i=1}^{n}a_i^2\le \frac{4}{n^2}\sum_{i=1}^{n}w(X_i)^2.\]
  Thus, for any $a\in \mathcal{N}((X_i, \td{T}_i)_{i=1}^{n})$,
  \begin{equation}
    \label{eq:conditional_subgaussian}
    \frac{2}{n}\sum_{i=1}^{n}\eta_{i}a_{i} \text{ is sub-Gaussian with parameter }\frac{4}{n^2}\sum_{i=1}^{n}w(X_i)^2 \text{ conditional on }(X_i, \td{T}_i)_{i=1}^{n}.
  \end{equation}
  Now, we set $\Phi(y) = y^{r}$ which is increasing and convex since $r > 2 \ge 1$. It is well-known that the $r$-th moment of a $\sigma^2$-sub-Gaussian random variable is equivalent to that of a Gaussian random variable with variance $\sigma^2$ up to a constant that only depends on $r$ \citep[e.g.][Proposition 2.5.2]{vershynin2018high}. Thus, there exists a constant $C(r) > 0$ such that
    \[\E\left[\Phi\lb \bigg|\frac{2}{n}\sum_{i=1}^{n}\eta_{i}a_{i}\bigg|\rb\mid (X_i, \td{T}_i)_{i=1}^{n}\right]\le C(r)\lb\frac{1}{n^2}\sum_{i=1}^{n}w(X_i)^2\rb^{r/2}.\]
    By H\"{o}lder's iequality,
    \[\E\left[\Phi\lb \bigg|\frac{2}{n}\sum_{i=1}^{n}\eta_{i}a_{i}\bigg|\rb\mid (X_i, \td{T}_i)_{i=1}^{n}\right]\le \frac{C(r)}{n^{r/2}} \lb\frac{1}{n}\sum_{i=1}^{n}w(X_i)^r\rb.\]
    % Note that
    % \[N(\calA(\calC)) = \sup_{(x_i, \td{t}_i)\in \mathrm{Domain}(X_i, \td{T}_i)}\big|\mathcal{N}\lb (x_i, \td{t}_i)_{i=1}^{n}\rb\big|.\]
    Then,
    \begin{align}
      &\E\left\{\sum_{a\in \mathcal{N}((X_i, \td{T}_i)_{i=1}^{n})} \E\left[\Phi\lb \bigg|\frac{2}{n}\sum_{i=1}^{n}\eta_{i}a_{i}\bigg|\rb\mid (X_i, \td{T}_i)_{i=1}^{n}\right]\right\}\\
      & \le \E\left\{|\mathcal{N}((X_i, \td{T}_i)_{i=1}^{n})| \cdot \frac{C(r)}{n^{r/2}} \lb\frac{1}{n}\sum_{i=1}^{n}w(X_i)^r\rb\right\}\\
      & \le C(r)\frac{N(\calA(\calC))}{n^{r/2}}\E[w(X_1)^{r}].\label{eq:EPhi_case1}
    \end{align}
    By Markov's inequality, \eqref{eq:EPhi_generic}, and \eqref{eq:EPhi_case1},
    \begin{align*}
      &\sup_{c_0\in \calC, v\in \R}\bigg|\frac{1}{n}\sum_{i=1}^{n}w(X_i)I(V(X_i, \td{T}_i; c_0)\le v) - \E\left[w(X)I(V(X, \td{T}; c_0)\le v)\right]\bigg|\\
      &= O_\p\lb\left\{\E\sup_{c_0\in \calC, v\in \R}\bigg|\frac{1}{n}\sum_{i=1}^{n}w(X_i)I(V(X_i, \td{T}_i; c_0)\le v) - \E\left[w(X)I(V(X, \td{T}; c_0)\le v)\bigg|^{r}\right]\right\}^{1/r}\rb\\
      &= O_\p\lb \left\{\frac{N(\calA(\calC))}{n^{r/2}}\right\}^{1/r}\rb = O_\p\lb \frac{N(\calA(\calC))^{1/r}}{\sqrt{n}}\rb.
    \end{align*}
  By Kolmogorov's strong law of large numbers,
  \[\bigg|\frac{1}{n}\sum_{i=1}^{n}w(X_i) - 1\bigg| = O_\p\lb\frac{1}{\sqrt{n}}\rb.\]
  By Slusky's Lemma,
  \[\bigg|\lb\frac{1}{n}\sum_{i=1}^{n}w(X_i)\rb^{-1} - 1\bigg| = O_\p\lb\frac{1}{\sqrt{n}}\rb.\]
  Thus,
  \[\sup_{c_0\in \calC, v\in \R}\bigg|\frac{\sum_{i=1}^{n}w(X_i)I(V(X_i, \td{T}_i; c_0)\le v)}{\sum_{i=1}^{n}w(X_i)} - \frac{\E\left[w(X)I(V(X, \td{T}; c_0)\le v)\right]}{(1/n)\sum_{i=1}^{n}w(X_i)}\bigg| = O_\p\lb \frac{N(\calA(\calC))^{1/r}}{\sqrt{n}}\rb.\]
  On the other hand,
  \begin{align}
    &\sup_{c_0\in \calC, v\in \R}\bigg|\frac{\E\left[w(X)I(V(X, \td{T}; c_0)\le v)\right]}{(1/n)\sum_{i=1}^{n}w(X_i)} - \E\left[w(X)I(V(X, \td{T}; c_0)\le v)\right]\bigg|\\
    & \le \E[w(X)]\bigg|\lb\frac{1}{n}\sum_{i=1}^{n}w(X_i)\rb^{-1} - 1\bigg|\\
    &= O_\p\lb\frac{1}{\sqrt{n}}\rb = O_\p\lb \frac{N(\calA(\calC))^{1/r}}{\sqrt{n}}\rb. 
  \end{align}
  By triangle inequality,
  \[\sup_{c_0\in \calC, v\in \R}\bigg|\frac{\sum_{i=1}^{n}w(X_i)I(V(X_i, \td{T}_i; c_0)\le v)}{\sum_{i=1}^{n}w(X_i)} - \E\left[w(X)I(V(X, \td{T}; c_0)\le v)\right]\bigg| = O_\p\lb \frac{N(\calA(\calC))^{1/r}}{\sqrt{n}}\rb.\]

\subsection{Coverage results under complete independent censoring}
\label{appx:coverage_independent_censoring}
Suppose now weights $W_i = 1/\hat{c}(X_i)$ is non-decreasing in 
the conformity scores $V_i = V(X_i,\td{T}_i; c_0)$ in the sense that 
\begin{align*}
  V_i \le V_j \Rightarrow W_i \le W_j,
\end{align*}
for any $i,j \in \calI_{\ca}\cup \{n+1\}$.
The predictive interval given by our algorithm is
\begin{align*}
\hat{C}(X_{n+1}) = \Bigg\{y: V(X_{n+1},y) \le \quantile\bigg(\frac{\sum_{i \in \calI_{\ca}}
W_i \cdot \delta_{V_i} + W_{n+1} \cdot \delta_{\infty}}
{\sum_{i\in\calI_{\ca}} W_i + W_{n+1} }, 1-\alpha\bigg)\Bigg\}.
\end{align*}
Without loss of generality, we let $\calI_{\ca} = [n]$, so
$\calI_{\ca}\cup \{n+1\} = [n+1]$.
Let $\{\pi_1,\ldots,\pi_{n+1}\}$ be a permutation of $[n+1]$
such that $V_{\pi_1} \le \ldots \le V_{\pi_{n+1}}$. By the
assumption we also have $W_{\pi_1} \le \ldots \le W_{\pi_{n+1}}$.
Next, for any $k\in[n+1]$, we claim that
\begin{align*}
  \frac{\sum^k_{i=1} W_{\pi_i}}{\sum^{n+1}_{j=1}W_j}
  \le \frac{k}{n+1}.
\end{align*}
Assume otherwise. Then for any $j>k$, by the monotonicity
\begin{align*}
  W_{\pi_j} \ge \frac{1}{k} \sum_{i=1}^k W_{\pi_i}.
\end{align*}
Consequently, 
\begin{align*}
\frac{  \sum_{j=k+1}^{n+1} W_{\pi_j}}{\sum_{i=1}^{n+1}W_i}
  \ge \frac{n+1-k}{k}\frac{\sum^{k}_{j=1} W_{\pi_j}}{\sum_{i=1}^{n+1}W_i}
  > \frac{n+1 - k}{n+1},
\end{align*}
which is a contradiction. Hence, we conclude that
for any $t \in \R$,
\begin{align*}
  \frac{\sum_{i =1}^{n+1}
  W_i \cdot \mathbf{1}\{V_i \le t\}}
{\sum_{j=1}^{n+1} W_j }
\le \frac{\sum_{i =1}^{n+1}
\mathbf{1}\{V_i \le t\}}
{n+1},
\end{align*}
and correspondingly,
\begin{align*}
\quantile\bigg(\frac{\sum_{i \in [n+1]}
W_i \cdot \delta_{V_i}}
{\sum_{j\in[n+1]} W_j}, 1-\alpha\bigg)
\ge \quantile\bigg(\frac{\sum_{i \in [n+1]}
\delta_{V_i}}{n+1}, 1-\alpha\bigg).
\end{align*}
Finally, we have
\begin{align*}
\p\big(Y_{n+1} \in \hat{C}(X_{n+1})\big) = & 
\p\Bigg(Y_{n+1} \le\quantile\bigg(\frac{\sum_{i \in \calI_{\ca}}
W_i \cdot \delta_{V_i} + W_{n+1} \cdot \delta_{V_{n+1}}}
{\sum_{i\in\calI_{\ca}} W_i + W_{n+1} }, 1-\alpha\bigg) \Bigg)\\
\ge &\p\Bigg(Y_{n+1} \le\quantile\bigg(\frac{\sum_{i \in \calI_{\ca}}
\delta_{V_i} +  \delta_{V_{n+1}}}{n+1}, 1-\alpha\bigg)\Bigg)\\
\ge & 1-\alpha.
\end{align*}

\subsection{Time complexity analysis}
\label{sec:time_comp}
We consider the following cases separately:
\begin{itemize}
\item The censoring mechanism is known and 
the threshold $c_0$ is determined a priori: the 
time complexity of the comformalized method can be decomposed into
\begin{align*}
  & \mbox{TC(conformalized survival analysis)}\\
  & =  \mbox{TC(fitting the quantile of $T$)} \\
  &\quad + \mbox{negligible cost to compute conformity scores and weighted empirical quantile}.
\end{align*}

%\item The censoring mechanism is known and the threshold
%$c_0$ is selected on the data: the time complexity of the
%the conformalized method comprises of the time complexity of 
%quantile-fitting base method, the computation of conformity
%scores, weights, and the empirical quantile on the training fold for all the candidates 
%for $c_0$ and the computation of conformity scores, weights
%and the empirical quantile on the calibration fold:
%\begin{align*}
%\mbox{TC(conformalized)} = 
%\mbox{TC(base method for fitting $T$ quantile)}
%+ (\#\mbox{ of candidate }c_0) \times O(n).
%\end{align*}

\item The censoring mechanism is unknown and the 
threshold $c_0$ is determined a priori: the time
complexity can be decomposed into
\begin{align*}
  & \mbox{TC(conformalized survival analysis)}\\
  & = \mbox{TC(fitting the quantile of $T$)}\\
& \quad  + \mbox{TC(fitting $\mathbb{P}(C \ge c_0 ~|~X)$)}\\
& \quad + \mbox{negligible cost to compute the conformity scores and weighted empirical quantile}.
\end{align*}

%\item The censoring mechanism is unknown and
%the threshold $c_0$ is selected on the data:
%the time complexity of the conformalized method
%equals to the sum of the time complexity of 
%the base method for fitting $T$ quantile,
%the base method for fitting $\mathbb{P}(C\ge c_0~|~X)$
%for all candidate $c_0$'s and the computation
%of conformity scores, weights and the empirical
%quantile for all candidate $c_0$'s:
%\begin{align*}
%\mbox{TC(conformalized)} = 
%& \mbox{TC(base method for fitting $T$ quantile)}\\
%& + \sum_{c_0} \mbox{TC(base method for fitting $\mathbb{P}(C\ge c_0 ~|~X)$)}\\
%& + (\# \mbox{ of candidate }c_0) \times O(n).
%\end{align*}
\end{itemize}

\subsection{An illustrating simulation in the presence of additional censoring (Section~\ref{subsec:beyond})}
\label{sec:add_sim}
To evaluate the validity and efficiency of our method under the 
more general setting discussed in Section~\ref{subsec:beyond}, we
set up our simulation as follows: the covariate $X \sim \calU(0,4)$
and the survival time $T$ satisfies $\log T \mid X \sim \calN(\mu(X),\sigma^2(X))$,
where $\mu(x) = 2 + 0.37\cdot \sqrt{x}$ and $\sigma(x) = 1 + x/5$. 
There are censoring times. The end-of-study censoring time
$\Cend \sim \calE(0.4)$; the loss-to-follow-up censoring time
$\Closs$ is generated from the following model:
\begin{align*}
  \log \Closs \mid X \sim \calN(2 + 0.05 \log T + 0.09\cdot (X-2)(X-3)(X-4),1).
\end{align*}
Clearly, $\Closs$ is not independent of $T$ even conditional on $X$,
but Assumption~\ref{eq:Cend} is satisfied in 
this example. We then apply our method to the observable
data $(X, T \wedge \Cend \wedge \Closs, \Cend)$ with a
target level $90\%$;
the implementation details are the same as in Section~\ref{sec:sim}.

Figure~\ref{fig:sim3_coverage} plots the empirical coverage of
$T$ and $T\wedge \Closs$ of three variants of our proposed methods
and the naive CQR method. The naive CQR method is very conservative,
showing an almost $100\%$ coverage for both $T$ and $T\wedge \Closs$; 
the three variants of the proposed methods are all less conservative 
than the naive CQR since they are able to remove the consoring from 
$\Cend$---both CDR-conTree and CQR-cRF are able to achieve exact coverage
for $T\wedge \Closs$; the coverage for $T$ is higher than the target level,
where the conservativeness comes from the censoring of $\Closs$. 
Figure~\ref{fig:sim3_conditional_coverage} further plots 
the empirical conditional coverage of $T$ and $T\wedge \Closs$ as 
functions of $\mathrm{Var}(T\mid X)$. The three variants of
our method are less conservative than the naive CQR;
CDR-conTree and CQR-cRF achieves conditional coverage approximately.
Figure~\ref{fig:sim3_length} plots  the ratio 
between the LPB to the theoretical conditional quantile
of three variants of our proposed methods and the naive CQR method,
where we again see that the naive CQR provides non-informative lower
bounds.

\begin{figure}[ht]
\centering
\begin{minipage}{0.45\textwidth}
\centering
\includegraphics[width = \textwidth]{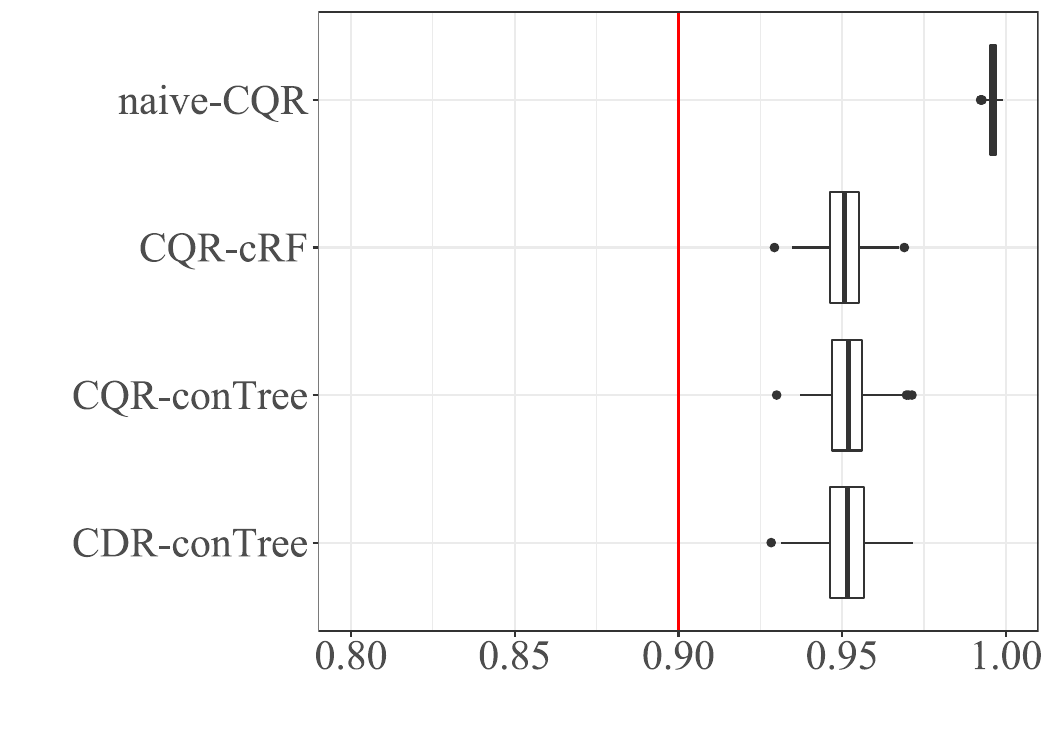}\\
(a)
\end{minipage}
\begin{minipage}{0.45\textwidth}
\centering
\includegraphics[width = \textwidth]{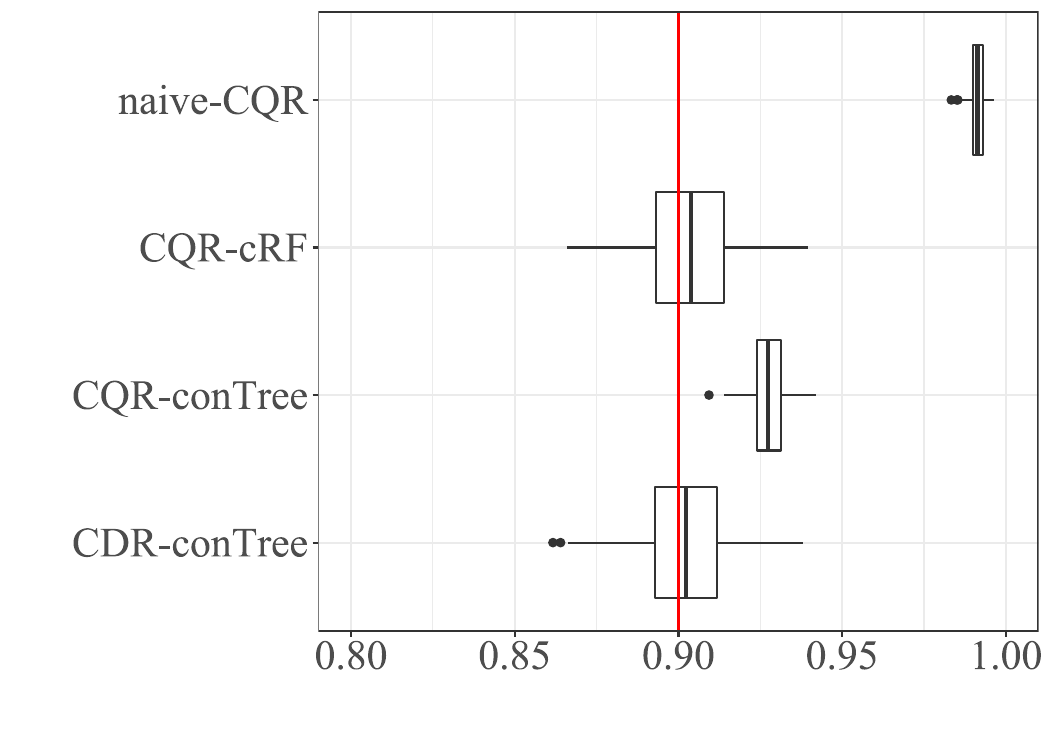}\\
(b)
\end{minipage}
\caption{Empirical $90\%$ coverage of (a) the uncensored survival
time $T$ and (b) the partially censored survival time $\Closs \wedge T$. 
The abbreviations are the same as in Figure~\ref{fig:sim1_marginal_coverage}.}.
\label{fig:sim3_coverage}
\end{figure}

\begin{figure}[ht]
\centering
\includegraphics[width = 0.9\textwidth]{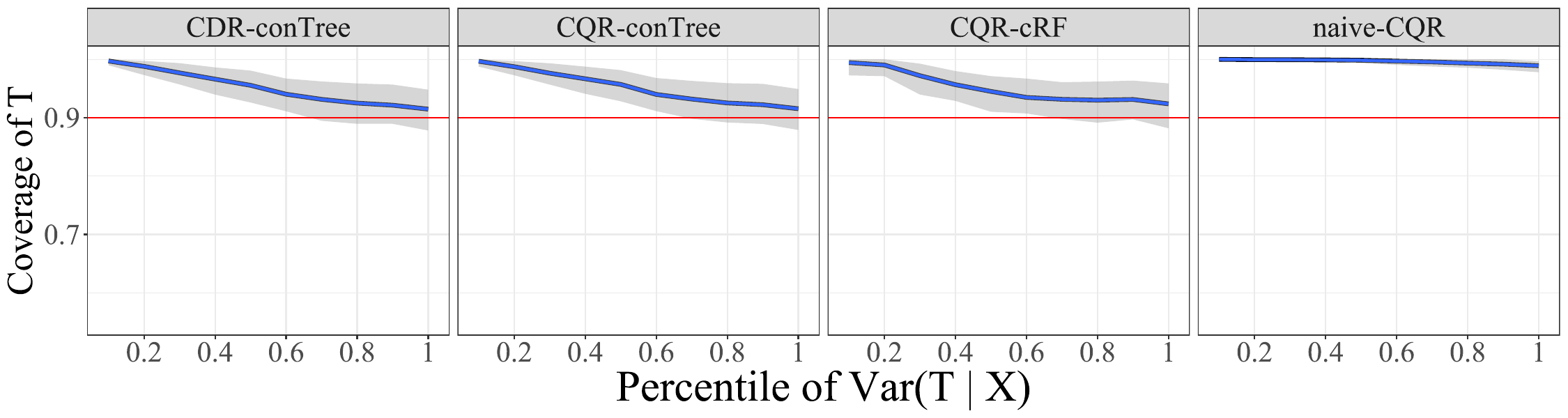}\\
(a)\\
\includegraphics[width = 0.9\textwidth]{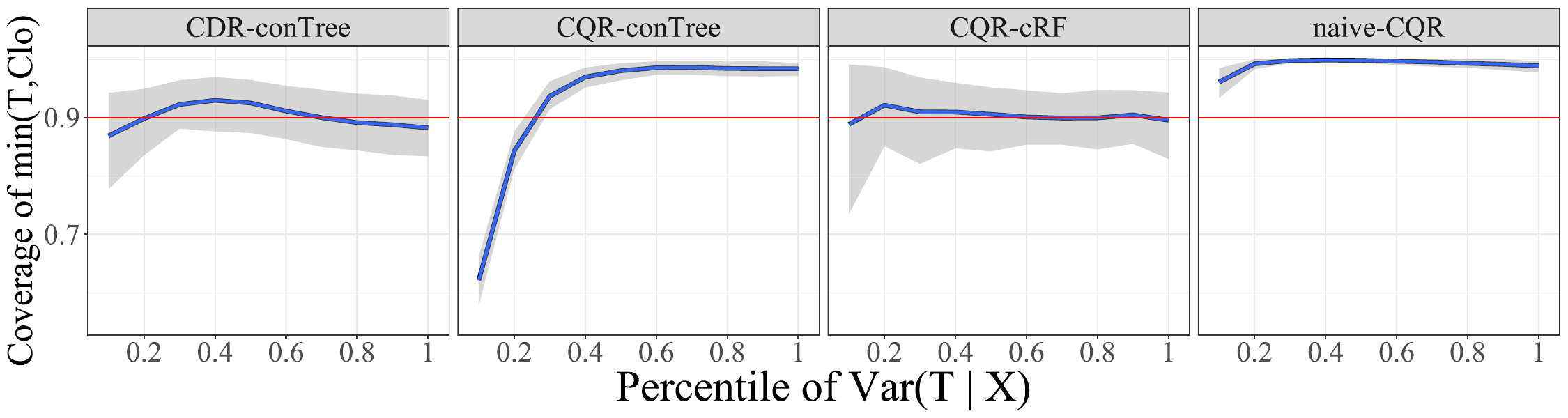}\\
(b)
\caption{Empirical $90\%$ conditional coverage of (a) the uncensored survival
  time $T$ and (b) the paritally censored survival time $T \wedge \Closs$
  as functions of $\mathrm{Var}(T\mid X)$. The
other details are the same as in Figure~\ref{fig:sim1_conditional_coverage}.}.
\label{fig:sim3_conditional_coverage}
\end{figure}

\begin{figure}[ht]
\centering
\includegraphics[width = 0.9\textwidth]{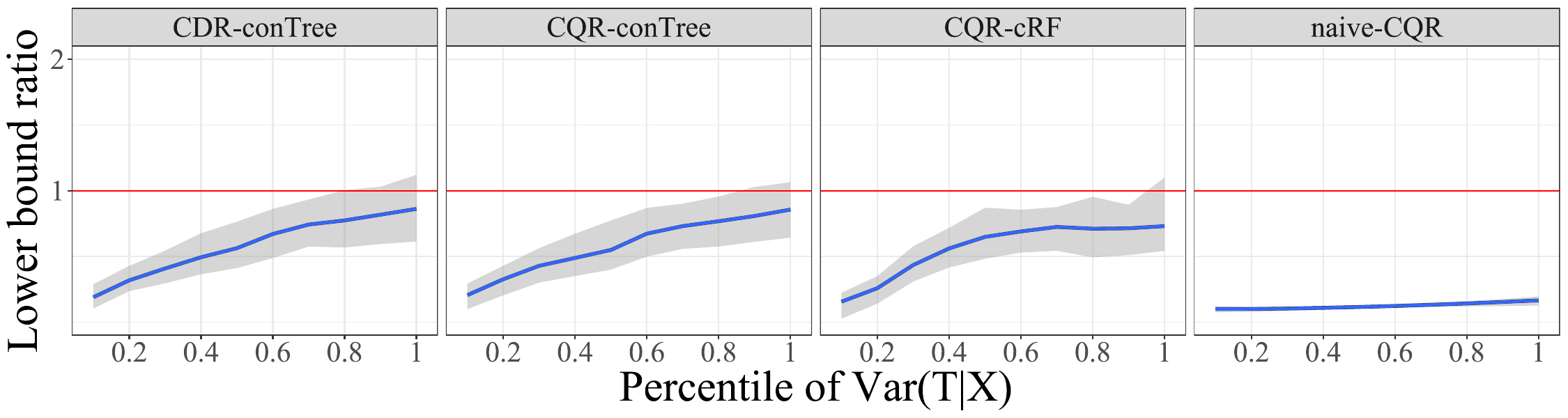}
\caption{Ratio between the LPB and the 
theoretical quantile as a function of $\mathrm{Var}(T \mid X)$.
The other details are the same as in Figure~\ref{fig:sim1_conditional_coverage}.}.
\label{fig:sim3_length}
\end{figure}
}
%%% Local Variables:
%%% mode: latex
%%% TeX-master: "supplement"
%%% End:

\end{document}